%% file: tech-report.tex
\documentclass[11pt]{article}
\usepackage[margin=1in]{geometry}
\usepackage{authblk}

\title{A Formal Model for Secure Multiparty Computations} 
\date{}

\author[1]{Amy Rathore}
\author[1]{Marina Blanton}
\author[2]{Marco Gaboardi}
\author[1]{Lukasz Ziarek}
\affil[1]{Department of Computer Science and Engineering, University at Buffalo}
\affil[2]{Department of Computer Science, Boston University}

\include{prelude}


\begin{document}

\lstset{language=C, basicstyle=\footnotesize\ttfamily, numbers=left, numbersep=5pt, tabsize=2,mathescape=true}
\lstset{emph={[1]public,private,smcinput,smcoutput,smcopen,pfree,pmalloc},emphstyle={[1]\color{red}}}

\maketitle

\begin{abstract}
Although Secure Multiparty Computation (SMC) has seen considerable development in recent years, its use is challenging, resulting in complex code which obscures whether the security properties or correctness guarantees hold in practice. 
For this reason, several works have investigated the use of formal methods  to provide guarantees for
SMC systems. However, these approaches have been 
applied mostly to domain specific languages (DSL),
neglecting general-purpose approaches. 
In this paper, we consider a formal
model for an SMC system for  annotated C programs. 
We choose C due to its popularity
in the cryptographic community and being the only general-purpose language
for which SMC compilers exist.
Our formalization supports all
key features of C -- including private-conditioned branching statements, mutable arrays
(including out of bound array access), pointers to private data, etc. 
We use this formalization to  characterize correctness and security properties of annotated C, with the
latter being a form of non-interference on execution traces.  We realize our formalism as an implementation in the PICCO SMC compiler and provide
evaluation results on SMC programs written in C.
\end{abstract}

\section{Introduction} \label{Sec: Introduction}

\input{intro}

\section{Related Work}
\input{rel-work}

\section{Overview}
\label{sec:background}

\input{background}

\section{Semantics} \label{Sec: Semantics}
\input{semantics}

\section{Metatheory}\label{Sec: Metatheory}

\input{selected-proofs}

\section{Implementation}
\label{Sec:implementation}
\input{implementation}

\section{Evaluation}
\label{Sec:evaluation}
\input{evaluation}

\section{Conclusions}
\label{sec:conclusions}

\begin{figure*}  \footnotesize
\begin{lstlisting}[firstnumber = 22]
avgFemaleSalPub=smcopen(avgFemaleSalary); 
femaleCountPub=smcopen(femaleCount);
avgMaleSalPub=smcopen(avgMaleSalary); maleCountPub=smcopen(maleCount);
avgFemaleSalPub=(avgFemaleSalPub/femaleCountPub)/2+historicFemaleSalAvg/2; 
avgMaleSalPub=(avgMaleSalPub/maleCountPub)/2+historicMaleSalAvg/2; 

for (i=1; i<numParticipants+1; i++) 
	smcoutput(avgFemaleSalPub, i);  smcoutput(avgMaleSalPub, i); 
\end{lstlisting}
\caption{Securely calculating the gender pay gap for 100 organizations with additional information released.}
\label{Fig: salary vs gender smcopen}
\end{figure*}

In this paper we have presented a formal model for a general SMC compiler, supporting both safe and unsafe features of C.  
Our model does not artificially restrict what C features can be present in private branches -- restrictions are instead guided by which operations
our model has shown to be unsafe. 
Our extension supports additional tracking meta-data to provide support for features unsafe in
current SMC techniques.  The intuition, shown in our motivation, is that state-of-the-art SMC techniques cannot track complex memory indirections that can occur when using pointers.  By providing this tracking, these operations can be made safe.
As future work we plan on extending our model to support explicit declassification, through a primitive PICCO calls \texttt{smcopen}.   Consider Figure~\ref{Fig: salary vs gender smcopen} which highlights a modification to our original
gender based salary computation (lines 16-17) from Figure~\ref{Fig: salary vs gender}.  Explicitly declassifying the sum and count earlier in the program, allows us to change the average computation to a public computation.  This reduces the number of high cost communications and cryptographic computations in the program.  To support explicit declassification in our model we would need to extend our semantics with gradual 
release~\cite{GR}.

\section*{Acknowledgements}

This work was supported in part by a Google Faculty Research Award and US National Science Foundation grants 1749539, 1845803, 2040249, and 2213057.
Any opinions, findings, and conclusions or recommendations expressed in this publication are those of the authors and do not necessarily reflect the views of the funding sources.

\bibliography{main}

\appendix

\input{appendix}

\end{document}

%% file: prelude.tex
\usepackage{pgfplotstable}
\pgfplotstableset{precision=5}
\usepackage{booktabs}

\usepackage{enumitem}
\usepackage{subcaption} 
\usepackage{mathpartir}
\usepackage{xcolor}

\usepackage{bbding}
\usepackage{amsmath,amssymb,amsfonts}
\usepackage{amsthm}
\usepackage{algorithmic}
\usepackage{algorithm}
\usepackage{bm}
\usepackage{multirow}
\usepackage{color}
\usepackage{textcomp}
\usepackage{graphicx}
\usepackage{stmaryrd}
\usepackage{float}
\usepackage{hyperref}

\usepackage{listings}
\usepackage[T1]{fontenc}

\newcommand{\fc}{\CircleSolid}
\newcommand{\hc}{\HalfCircleRight}
\newcommand{\ec}{\CircleShadow}

\newcommand{\bsq}{\Box}

\colorlet{darkpink}{red!60!purple!50}
\colorlet{lightpurple}{purple!60!blue!60}
\colorlet{purple}{purple!70!blue!100}
\colorlet{teal}{green!70!blue!90!black!100}
\colorlet{darkgreen}{green!50!black!100}
\colorlet{orange}{orange!70!red!80}
\colorlet{yellow}{yellow!100!orange!70!red!80!black!80}

\newcommand{\purple}[1]{\textcolor{purple}{#1}}
\newcommand{\teal}[1]{\textcolor{teal}{#1}}

\newcommand{\orange}[1]{\textcolor{orange}{#1}}
\newcommand{\cyan}[1]{\textcolor{cyan}{#1}}
\newcommand{\red}[1]{\textcolor{red}{#1}}
\newcommand{\blue}[1]{\textcolor{blue}{#1}}

\newcommand{\TT}[1]{\texttt{#1}}

\newcommand{\ExprC}[1]{\blue{#1}}

\newcommand{\extC}[1]{\teal{#1}}

\newcommand{\initC}[1]{\red{#1}}

\newcommand{\sC}[1]{\cyan{#1}}

\newcommand{\ssC}[1]{\textcolor{lightpurple}{#1}}

\newcommand{\restC}[1]{\textcolor{yellow}{#1}}

\newcommand{\resoC}[1]{\orange{#1}}

\theoremstyle{plain}
\newtheorem{lemma}{Lemma}
\newtheorem{theorem}{Theorem}
\newtheorem{axiom}{Axiom}

\theoremstyle{definition}
\newtheorem{definition}{Definition}

\newcommand{\eval}{\Downarrow}
\newcommand{\Expr}{e}
\newcommand{\stmt}{s}
\newcommand{\val}{v}
\newcommand{\valL}{V}
\newcommand{\var}{\mathit{var}}

\newcommand{\x}{x}
\newcommand{\vl}{\mathit{list}}
\newcommand{\res}{\mathit{res}}
\newcommand{\binop}{\mathit{bop}}
\newcommand{\unop}{\mathit{uop}}
\newcommand{\preop}{\mathit{uop}}

\newcommand{\decl}{\mathit{decl}}
\newcommand{\Type}{\mathit{ty}}

\newcommand{\btype}{\mathit{bty}}
\newcommand{\llabel}{a}
\newcommand{\builtin}{\mathit{prim}}
\newcommand{\Elist}{E}
\newcommand{\Tlist}{\mathit{tyL}}
\newcommand{\plist}{P}

\newcommand{\bytelen}{\kappa}
\newcommand{\n}{n}

\newcommand{\byte}{\omega}

\newcommand{\tagb}{j}
\newcommand{\tagbL}{J}
\newcommand{\perm}{p}

\newcommand{\Const}{\mathrm{const}}
\newcommand{\Skip}{\mathrm{skip}}

\newcommand{\Void}{\mathrm{void}}

\newcommand{\Int}{\mathrm{int}}
\newcommand{\Float}{\mathrm{float}}

\newcommand{\Priv}{\mathrm{private}}
\newcommand{\Pub}{\mathrm{public}}

\newcommand{\If}{\mathrm{if}}
\newcommand{\Then}{\mathrm{then}}
\newcommand{\Else}{\mathrm{else}}

\newcommand{\free}{\mathrm{free}}
\newcommand{\pfree}{\mathrm{pfree}}
\newcommand{\Malloc}{\mathrm{malloc}}
\newcommand{\PMalloc}{\mathrm{pmalloc}}
\newcommand{\Null}{\mathrm{NULL}}
\newcommand{\While}{\mathrm{while}}

\newcommand{\smcinput}{\mathrm{smcinput}}
\newcommand{\smcoutput}{\mathrm{smcoutput}}

\newcommand{\PermF}{\mathrm{Freeable}}

\newcommand{\PermN}{\mathrm{None}}

\newcommand{\Decrypt}{\mathrm{decrypt}}

\newcommand{\Encrypt}{\mathrm{encrypt}}

\newcommand{\GetLoc}{\mathrm{GetLocation}}

\newcommand{\Initialize}{\mathrm{InitializeVariables}}
\newcommand{\Resolve}{\mathrm{ResolveVariables}}
\newcommand{\ResolveR}{\mathrm{ResolveVariables\_Retrieve}}
\newcommand{\ResolveS}{\mathrm{ResolveVariables\_Store}}

\newcommand{\Restore}{\mathrm{RestoreVariables}}
\newcommand{\resolve}{\mathrm{resolve}}

\newcommand{\sizeof}{\mathrm{sizeof}}

\newcommand{\Free}{\mathrm{Free}}

\newcommand{\MPC}[1]{\mathrm{MPC}_\mathit{{#1}}}
\newcommand{\PFree}{\MPC{free}}

\newcommand{\Retrieve}{\mathrm{Retrieve\_vals}}

\newcommand{\DerefPtrPub}{\mathrm{DerefPtr}}

\newcommand{\DeclassifyPtr}{\mathrm{DeclassifyPtr}}
\newcommand{\rrDeclassifyPtr}{\RT\DeclassifyPtr}

\newcommand{\Update}{\mathrm{UpdateVal}}

\newcommand{\UpdateArr}{\mathrm{UpdateArr}}

\newcommand{\UpdatePtr}{\mathrm{UpdatePtr}}

\newcommand{\InputVal}{\mathrm{InputValue}}

\newcommand{\OutputArr}{\mathrm{OutputArray}}

\newcommand{\DecodeArr}{\mathrm{DecodeArr}}

\newcommand{\Decode}{\mathrm{DecodeVal}}

\newcommand{\DecodePtr}{\mathrm{DecodePtr}}

\newcommand{\DecodeFun}{\mathrm{DecodeFun}}

\newcommand{\Encode}{\mathrm{EncodeVal}}

\newcommand{\EncodeArr}{\mathrm{EncodeArr}}

\newcommand{\EncodePtr}{\mathrm{EncodePtr}}

\newcommand{\EncodeFun}{\mathrm{EncodeFun}}

\newcommand{\ReadOOB}{\mathrm{ReadOOB}}

\newcommand{\WriteOOB}{\mathrm{WriteOOB}}

\newcommand{\Mem}{\mathrm{Mem}}

\newcommand{\plpl}{\mathrel{+}\mathrel{+}}

\newcommand{\Acc}{\mathrm{acc}}
\newcommand{\AccZ}{0}

\newcommand{\AccPP}{\Acc+1}

\newcommand{\qq}{\qquad\qquad}

\newcommand{\textred}[1]{\textcolor{red}{#1}}
\newcommand{\RT}[1]{\textred{\ensuremath{#1}}}
\newcommand{\rAcc}{\RT\Acc}
\newcommand{\rExpr}{\RT\Expr}
\newcommand{\rstmt}{\RT\stmt}

\newcommand{\rn}{\RT{n}}
\newcommand{\rrgamma}{\RT\gamma}
\newcommand{\rrsigma}{\RT\sigma}
\newcommand{\rgamma}[1]{\RT{\gamma{#1}}}
\newcommand{\rsigma}[1]{\RT{\sigma{#1}}}

\newcommand{\rVoid}{\RT\Void}

\newcommand{\rlabel}{\RT\llabel}
\newcommand{\rPriv}{\RT\Priv}
\newcommand{\rPub}{\RT\Pub}

\newcommand{\rloc}{\RT\loc}
\newcommand{\rval}{\RT\val}
\newcommand{\rType}{\RT\Type}
\newcommand{\rbtype}{\RT\btype}
\newcommand{\rbyte}{\RT\byte}

\newcommand{\rElist}{\RT\Elist}
\newcommand{\rTlist}{\RT\Tlist}
\newcommand{\rplist}{\RT\plist}

\newcommand{\Veval}{\eval'}

\newcommand{\N}{\mathbb{N}}

\newcommand{\inputFun}{\mathrm{mcinput}}
\newcommand{\outputFun}{\mathrm{mcoutput}}

\newcommand{\PermL}{\mathrm{PermL}}
\newcommand{\VarPermL}{\PermL}
\newcommand{\FunPermL}{\PermL\_\mathrm{Fun}}
\newcommand{\ArrPermL}{\PermL}

\newcommand{\PtrPermL}{\PermL\_\mathrm{Ptr}}

\newcommand{\getIndirection}{\mathrm{GetIndirection}}
\newcommand{\indir}{i}
\newcommand{\hindir}{\hat{\indir}}

\newcommand{\offset}{\ensuremath{\mu}}
\newcommand{\Hoffset}{\hat{\offset}}
\newcommand{\nl}{\alpha}
\newcommand{\hnl}{\hat{\nl}}

\newcommand{\ind}{i}
\newcommand{\hind}{\hat{\ind}}

\newcommand{\hgamma}{\hat{\gamma}}
\newcommand{\hsigma}{\hat{\sigma}}
\newcommand{\hx}{\hat{\x}}

\newcommand{\hstmt}{\hat{\stmt}}
\newcommand{\hExpr}{\hat{\Expr}}
\newcommand{\hType}{\hat{\Type}}
\newcommand{\hbtype}{\hat{\btype}}
\newcommand{\hval}{\hat{\val}}

\newcommand{\hbyte}{\hat{\byte}}

\newcommand{\hn}{\hat{n}}

\newcommand{\hloc}{\hat{\loc}}

\newcommand{\piccoC}{SMC$^2$}

\newcommand{\DynamicPicco}{\piccoC}
\newcommand{\vanillaC}{Vanilla C}

\newcommand{\codeV}{\hat{\code}}
\newcommand{\codeVL}{\hat{\codeL}}
\newcommand{\codeVLL}{\hat{\codeLL}}
\newcommand{\codeVV}[1]{\mathit{\hat{#1}}}
\newcommand{\codeVS}[1]{(\pid, [\codeVV{#1}])}
\newcommand{\codeVM}[1]{(\mathrm{ALL}, [\codeVV{#1}])}

\newcommand{\piccoCodes}{SmcC}

\newcommand{\vanillaCodes}{VanC}

\newcommand{\erasure}{\mathrm{Erase}}

\newcommand{\loweq}{\simeq_L}
\newcommand{\deriv}{\triangleright}

\newcommand{\Code}[1]{\texttt{#1}}

\newcommand{\changeMap}{location map}
\newcommand{\DMap}{\Delta}
\newcommand{\dmap}{\delta}

\newcommand{\lxor}{\oplus}

\newcommand{\UpdateOffset}{\mathrm{UpdateOffset}}

\newcommand{\UpdateBytesFree}{\mathrm{UpdateBytesFree}}

\newcommand{\locDefault}{\loc_\mathit{default}}
\newcommand{\hlocDefault}{\hloc_\mathit{default}}

\newcommand{\SelectFreeable}{\mathrm{CheckFreeable}}
\newcommand{\UpdatePtrLocs}{\mathrm{UpdatePointerLocations}}

\newcommand{\LocMap}{map}
\newcommand{\Pcong}{\ensuremath{\cong_{\psi}}}

\newcommand{\Deval}[2]{\eval^{#1}_\mathit{#2}}

\newcommand{\Dyn}{}
\newcommand{\DynExtract}{\mathrm{\Dyn Extract}}
\newcommand{\DynInit}{\mathrm{\Dyn Initialize}}
\newcommand{\DynUpdate}{\mathrm{DynamicUpdate}}
\newcommand{\DynRestore}{\mathrm{\Dyn Restore}}
\newcommand{\DynResolve}{\mathrm{\Dyn Resolve}}
\newcommand{\DynResolveR}{\mathrm{\Dyn Resolve\_Retrieve}}
\newcommand{\DynResolveS}{\mathrm{\Dyn Resolve\_Store}}

\newcommand{\isPriv}{\vdash}
\newcommand{\isPub}{\nvdash}
\newcommand{\Config}{C}
\newcommand{\hConfig}{\hat{C}}
\newcommand{\pidstyle}[1]{\mathrm{{#1}}}
\newcommand{\pid}{\pidstyle{p}}
\newcommand{\pidA}{\pidstyle{1}}
\newcommand{\pidZ}{\pidstyle{q}}

\newcommand{\loc}{l}
\newcommand{\locL}{L}
\newcommand{\locLL}{\mathcal{L}}
\newcommand{\addL}{::}		

\newcommand{\code}{d}
\newcommand{\codeL}{D}

\newcommand{\codeLL}{\mathcal{D}}

\newcommand{\addC}{::}		

\newcommand{\codeMP}[1]{(\mathrm{ALL}, [\mathit{#1}])}

\newcommand{\codeSP}[1]{(\pid, [\mathit{#1}])}

\newcommand{\Mid}{\-\ \Vert\ }

%% file: intro.tex
Secure Multiparty Computation (SMC) allows multiple parties to jointly compute over private data, revealing only the outcomes of the computation to designated recipients. Secure computation is needed in many domains, particularly the medical, military, and financial sectors. 
SMC is commonly implemented using low-level techniques like secret sharing~\cite{Shamir79}, garbled circuits~\cite{Yao86}, or homomorphic encryption~\cite{ElGamal85,Paillier99}. 
These low-level techniques are designed to enable computations among parties which are secure and efficient. While these low-level techniques provide efficiency, their use makes programming SMC applications challenging and error prone. 
To address this concern, several works have proposed high-level languages, DSLs, or language extensions providing abstractions, which can then be compiled down to low-level SMC primitives, to support programmers in writing SMC applications. 
As a result, there now exists a plethora of languages providing different expressivity, offering different features and performance trade-offs, using different threat models, and suitable for different domains.
 Similarly, there exist a number of SMC DSLs.  Although DSLs can make it easier to write SMC programs, there remains a disconnect between the DSL and its integration with the language of the underlying implementation of the compiled protocol. 

In the effort to unify our knowledge in this space, a recent work~\cite{HastingsHNZ19} compared several compilers and tools in terms of their expressivity and usability. We highlight two items among the lessons learned and recommendations to the community. First, there were numerous correctness issues and undocumented limitations present in the works surveyed. This finding is also echoed in~\cite{Mood16}, which found correctness issues in several two-party compilers. Second, the authors of~\cite{HastingsHNZ19} recommend that the community take a more principled approach to language design and verification, e.g., by defining and implementing type rules. This would help with ensuring correctness as well as reduce security-related corner cases overlooked by the compiler designers. 

To help achieves these goals, in this paper, we present a formalism of a general-purpose SMC system designed for C. We choose C because it provides the low level language framework targeted by most DSLs and there exist numerous direct language extensions for which multiple SMC compilers have been developed. This allows programmers to write secure multiparty programs in C, which the compiler will translate into secure computation protocols, avoiding managing the interactions between different languages. Given the maturity of SMC compilers today, modern implementations provide support for all C features, allowing private-conditioned branches (i.e., \texttt{if-else} statements whose conditional expressions use private data), use of private arrays, private indexing into arrays, and working with pointers to private data; all while ensuring that no private data is leaked over the execution of any given program. However, formally modeling semantics and translation of these features as done by SMC compilers presents non-trivial challenges not attempted in any prior work. Furthermore, what is interesting about C (and not present in the available well-typed DSLs) is that features such as pointer manipulation allow one to write programs that erroneously access unintended regions in memory. However, even in those circumstances, it is possible to show that the compiled protocol will not reveal any unintended information about private data that it handles.

Our contributions in this paper are:
\begin{enumerate}
\item a formal model for a general-purpose secure multiparty
  computation compiler, formalizing state of the art SMC techniques in 
  C. Our formal model supports distributed multi-party computation in the
  presence of private-conditioned branches, pointers to private data, 
  pointer arithmetic and general pointer operations in  private-conditioned branches. 
\item a formal proof based on this model that common SMC approaches
  guarantee correctness and a strong form of non-interference over execution
  traces consisting of multiple computing parties.  This shows that pointer
  operations can be safely managed with no restrictions on the program. 
 \item an implementation of our formal model in the PICCO SMC compiler and evaluation over micro-benchmarks and SMC programs.
\end{enumerate}

We present an overview of related work in the next section. Section~\ref{sec:background}  provides an overview of SMC, specifically highlighting features present in C extensions that support SMC, it presents SMC compiler background, and a motivating examples.
Section~\ref{Sec: Semantics} defines our semantics, introducing a model that ensures correctness and non-interference for SMC programs with full support for pointers.  Section~\ref{Sec: Metatheory} formalizes our metatheory and provides proof sketches.  Section~\ref{Sec:implementation} details our implementation of our model in the PICCO SMC compiler and  discusses a bug in the private-conditioned branch implementation that our
implementation solved.  Section~\ref{Sec:evaluation} provides evaluation results of our implementation.  We conclude and provide remarks on future work in Section~\ref{sec:conclusions}. The appendices provide additional information including definitions, algorithms, theorems and proofs. Additional details, results, and proofs are available in~\cite{amys-dissertation}.

%% file: rel-work.tex
\paragraph*{SMC compilers}
Work on SMC compilers was initiated in 2004 and a significant body of work has been developed. Notable examples include two-party computation compilers and tools Fairplay~\cite{Malkhi04}, TASTY~\cite{Henecka10}, ABY~\cite{Demmler15a}, PCF~\cite{Kreuter13}, TinyGarble~\cite{Songhori15}, Frigate~\cite{Mood16}, SCVM~\cite{Liu14}, and ObliVM~\cite{Liu15}; three-party Sharemind~\cite{Bogdanov08}; and multi-party FairplayMP~\cite{BenDavid08},  VIFF~\cite{DamgardGKN09}, and more recently SCALE-MAMBA, which evolved from \cite{BendlinDOZ11,DamgardPSZ12,NielsenNOB12}.
These compilers use custom DSLs to represent user programs, and notable exceptions are CBMC-GC~\cite{Holzer12} (intended to support general-purpose ANSI-C programs in the two-party setting, but not all features were realized at the time) and PICCO~\cite{Zhang13,Zhang18} (takes programs written in an extension of C, supports all C features, and produces multi-party protocols).
The above compilers did not come with a formalism of their type
system~\footnote{The ObliVM publication~\cite{Liu15} suggests that
there is a type system behind the ObliVM language, but no further
information could be found.}, while this was later developed for
Sharemind~\cite{sokk16}. There are also SMC DSLs with formal models,
such as Wysteria~\cite{RastogiHH14} with a formal model based on an
operational semantics and Wys*~\cite{RastogiSH17} which provides
support for SMC by means of an embedded DSL hosted in F*, a
dependently typed language supporting full verification. A different
approach is given in~\cite{PettaiL15} with an automated technique to
prove SMC protocols secure.

We provide a summary of significant features supported
in recent compilers in Table~\ref{tab:smc-compilers} (Wys*~\cite{RastogiSH17} inherits its expressivity from Wysteria and is
omitted).
\begin{table} \small \centering \setlength{\tabcolsep}{1ex} \footnotesize
\begin{tabular}{|c|c|c|c|c|c|c|} \hline
  \multirow{3}{*}{Compiler} & \multicolumn{6}{|c|}{Supported features} \\\cline{2-7}
  & \multirow{2}{*}{loops} & priv.  & mixed & float. & dyn. & seman. \\
  & & cond. & mode & point & mem. & formal. \\ \hline
  Fairplay~\cite{Malkhi04} & \hc & \fc & \ec & \ec & \ec & \ec \\ \hline
  Sharemind~\cite{Bogdanov08,jagomagis2010secrec} & \fc & \hc & \fc & \fc & \ec & \fc \\ \hline
  CBMC-GC~\cite{Holzer12} & \hc & \fc & \ec & \hc & \ec & \ec \\ \hline 
  PICCO~\cite{Zhang13,Zhang18} & \fc & \fc & \fc & \fc & \fc & \ec \\ \hline
  SCALE-MAMBA & \fc & \fc & \fc & \fc & \ec & \ec \\ \hline
  Wysteria~\cite{RastogiHH14} & \fc & \fc & \fc & \ec & \ec & \fc \\ \hline
  Frigate~\cite{Mood16} & \fc & \fc & \ec & \ec & \ec & \ec \\ \hline
  ABY~\cite{Demmler15a} & \hc & \hc & \fc & \fc & \ec & \ec \\ \hline
  ObliVM~\cite{Liu15} & \fc & \fc & \fc & \fc & \fc & \ec \\ \hline
  SCVM~\cite{Liu14} & \fc & \fc & \fc & \ec & \ec  & \fc \\ \hline
\end{tabular}
\caption{Language features supported in SMC compilers.} \label{tab:smc-compilers}
\end{table}
They are supporting loops, private-conditioned branches, supporting both private and public values (mixed mode), floating point arithmetic on private values, dynamically allocated memory, and having semantic formalism. 
Note that compilers that translate computation into Boolean circuits such as CBMC-GC need to unroll loops and thus can only support a bounded number of loop iterations, denoted as \hc\ in the table. 
ABY also appears to have this limitation and for that reason expects input sizes at compile time. 
Recent compilers that work with a circuit representation (e.g., Wysteria, ObliVM) store compiled programs using intermediate representation and perform loop unrolling and circuit generation at runtime. 
To the best of our knowledge, Sharemind permits updating only a single variable in a private-conditioned branch (i.e., \texttt{if (cond) a = b; else a = c;}). 
Similarly, in ABY the programmer has to encode all logic associated with conditional statements using multiplexers. 
CBMC-GC did not support floating point arithmetic based on open-source software at the time of publication. 

Dynamic memory management is often not discussed in prior work. CBMC-GC is said to support dynamic memory allocation, as long as this can be encoded as a bounded program, but the use of dynamic arrays and memory deallocation is not mentioned. PICCO explicitly supports C-style memory allocation and deallocation as well as dynamic arrays. ObliVM does not explicitly discuss dynamically allocated arrays, but we believe they are supported.
Similarly, out-of-bounds array access in user programs is also not typically discussed in the SMC literature. Therefore, it is difficult to tell what the behavior might be, i.e., whether the compiler checks for this and, if not, whether the behavior of the corresponding compiled program is undefined. Wysteria and PICCO are two notable exceptions: Wysteria has a strongly typed language and will prevent such programs from compiling (recall that it supports only static sizes). PICCO will compile programs with out-of-bounds memory accesses. While the behavior of such programs is undefined in C (and no correctness guarantees can be provided), its analysis demonstrates that no privacy violations take place. We formalize this behavior in this work.

\paragraph*{Non-interference}
Non-interference is a standard information flow security property guaranteeing that information about private data does not directly affect publicly observable data. We will show non-interference over executions of programs using the formal
model and its extension developed in this paper to prove security when SMC techniques and C language primitives are composed. Non-interference and its several variants have been extensively studied by means of language-based techniques, including type systems~\cite{VolpanoS97,AbadiBHR99}, runtime monitor~\cite{AustinF09,SabelfeldR09}, and multi-execution~\cite{DevrieseP10}, to cite a few. One of the challenges in guaranteeing non-interference when attackers can inspect the state of the computation is to guarantee that private information is not implicitly leaked by means of the control flow path, i.e., that the computation is data-oblivious. Several language-based methods have been designed to guarantee that systems are secure against leakage from branching statements, including timing analysis~\cite{Ford12} and multi-path execution~\cite{PlanulM13,Mitchell0SZ12,LaudP16}.
In particular, \cite{Mitchell0SZ12} considered an approach similar to the one we use here. However, these approaches do not prevent private data leakage from explicit memory management. Building on these early works,  several recent works~\cite{PatrignaniG17,AbateBCD0HPTT20} have shown that in the context of secure compilation the natural notion that one needs to consider is a form of non-interference extended to traces. Inspired by this work, this is the notion we use in this paper when reasoning about non-interference.

%% file: background.tex
\begin{figure*} \centering
\begin{tabular}{c}
\begin{lstlisting}
public int numParticipants=100, maxInput=100, inputSize[numParticipants];
public int historicFemaleSalAvg, historicMaleSalAvg, i, j;
private int salary[numParticipants][maxInput], maleCount=0; 
private int gender[numParticipants][maxInput], femaleCount=0;
private int avgMaleSalary=0, avgFemaleSalary=0; 

smcinput(historicFemaleSalAvg,1); smcinput(historicMaleSalAvg,1);
for (i=0; i<numParticipants; i++) 
	smcinput(inputSize[i], i+1); 
	smcinput(gender[i], inputSize[i], i+1); 
	smcinput(salary[i], inputSize[i], i+1); 
	
for (i=0; i<numParticipants; i++) 
	for (j=0; j<inputSize[i]; j++) 
		if (gender[i][j]==0) {
			avgFemaleSalary+=salary[i][j]; 
			femaleCount++;} 
		else { 
			avgMaleSalary+=salary[i][j]; 
			maleCount++;} 
			
avgFemaleSalary=(avgFemaleSalary/femaleCount)/2+historicFemaleSalAvg/2; 
avgMaleSalary=(avgMaleSalary/maleCount)/2+historicMaleSalAvg/2; 

for (i=1; i<numParticipants+1; i++) 
	smcoutput(avgFemaleSalary, i);  smcoutput(avgMaleSalary, i); 
\end{lstlisting}
\end{tabular}
\caption{Securely calculating the gender pay gap for 100 organizations.}
\label{Fig: salary vs gender}
\end{figure*}

We next motivate the need for formalization and discuss the challenges formalizing main language features presents. Let us consider an example SMC program that demonstrates how computation with private values can be specified and carried out. Figure~\ref{Fig: salary vs gender} presents an SMC program that securely computes the average salary of employees in a particular field by gender. This is representative of a real-world SMC deployment in the City of Boston gender pay gap study~\cite{Lapets18}, which evaluated pay inequalities by gender and race. While we list a simple program, the computation can be extended to securely compute comprehensive statistical information by gender, race, and other relevant attributes.

In this example, there are 100 organizations (line 1), each contributing a number of records about their employees, specified in the form of salary-gender pairs (lines 3--4). All records are private and must be protected from all participants. The computation proceeds by checking the gender field in a record and using the corresponding salary in the computation of either female or male average salary. Once the average salaries are computed privately, they are combined with publicly available historical averages using weighted average computation (lines 22-23).

When we talk about secure computation, we need to distinguish between the values which must be protected throughout the computation (and are not revealed to the participants) and the values requiring no protection which are therefore observable during the computation. In our specification, the former are marked as \emph{private} (e.g., salary-gender pairs and information derived from them) and the latter are marked as \emph{public} (e.g., the number of records that each organization contributes and historical average salaries). The ability to combine computation with private and public values is called \emph{mixed-mode} execution.

In our example, all inputs (public or private) are entered into the computation using the \texttt{smcinput} interface that expects the variable name, its dimension(s) (for array variables) and the index of the party supplying the input. When a private input enters the computation, it uses a different, cryptographically protected representation in accordance with the underlying secure multi-party computation techniques. That representation is maintained throughout the execution, which means that all operations on private data are carried out using the corresponding secure multi-party protocols and their true values are not observable. All public values are handled in the same way as in conventional computation, i.e., an SMC compiler does not modify the computation and the values are observable by any participant running the computation. Upon computation completion, any output is communicated to the intended recipient via the \texttt{smcoutput} interface. If \texttt{smcoutput} is used with a private variable, its true value can be recovered from its cryptographically protected representation by the output recipient only (and remains unknown to all other parties). Because output recovery happens after the computation completes, the disclosed output is the intended outcome of secure computation and is not subject to the security guarantees maintained during the computation (such as non-interference). 

The above interface distinguishes between different types of participants: those who \emph{supply inputs}, those who \emph{learn outputs}, and the parties \emph{carrying out the computation}. The computational parties can be different from input owners and output recipients and their selection may be based on the properties of the underlying secure multi-party computation techniques. For example, there are often constraints with respect to what fraction of computational parties can collude, i.e., combine their individual views during the computation, in order to maintain security of private values. This means that in our example, the participants can select a subset of them to run the computation or employ other parties such as cloud computing providers.

The main property this design guarantees is that  a computational party that is not an output recipient should learn nothing about the private values it handles during the computation.  To formalize this property about SMC programs we will show non-interference between private and public computations in mixed-mode programs. Non-interference ensures that private data does not directly affect publicly observable data and is crucial for mixed-mode execution. For the example given in Figure~\ref{Fig: salary vs gender}, non-interference ensures that the private data (e.g., \texttt{salary[i][j]}) does not affect any public data (e.g., \texttt{historicFemaleSalaryAvg}).

This is a distributed, mixed-mode computation, computed between multiple participants. Each individual salary should be kept private and none of the participants should be able to deduce  the salary (up to some number of colluding computational parties -- this is dependent on the cryptographic protocols used).
The computation is mixed-mode as it includes portions which are {\em protected} -- computations over private data (e.g., finding the sum of all the salaries and counting the total number of males and females), and portions which are {\em unprotected} -- local computations over public data (e.g., dividing the historic salaries by 2 in line 22 and 23). 
Secure computations by definition are distributed and governed by cryptographic protocols (e.g., conditional addition \texttt{sumFemaleSalary += salary[i][j];} on line 11).   This occurs because to do the increment operation to compute the total sum, each participant must interact during the operation.

Given a program such as the one in our example, an SMC compiler will perform certain transformations such as substituting an operation on private values with the corresponding SMC protocol for performing that operation securely.  In addition, the SMC compiler also has to perform structural transformations to programs to ensure that there is no information flow from private to public variables based on the instructions that a computational party executes. This is called \emph{data-oblivious} (or data-independent) execution. This means that for private-conditioned branches, the execution must not reveal which branch gets executed.  
Given that our formal model must also ensure \emph{data-oblivious} execution, we have designed our formal model to embody these transformations, using big-step semantics to transform the original operations into the secure operations.  
We formalize this by showing that non-interference holds not only when considering the input-output functional behavior of a program, but also when considering the traces of execution of programs.

%% file: semantics.tex
In this section we introduce our semantic and memory model. Our goal is to
formalize a general purpose SMC compiler, showing correctness with respect to standard C semantics
and non-interference to guarantee no leakage of private data. 
We therefore introduce two models, C (referred to as \vanillaC) as well as the semantics for the SMC compiler (referred to as \piccoC). 
We do not abstract away memory, instead we introduce a byte-level memory model, 
inspired by the memory model used by CompCert~\cite{leroy2012compcert}, a formally verified C compiler. 
Specifically, we build from their approach of byte-level representation of data and permissions. 

\begin{figure}
\begin{minipage}{0.48\textwidth}\footnotesize
$\begin{array}{l c l}
	\Type &::=& \RT{\llabel\ \btype} \mid \RT{\llabel\ \btype*} \mid \btype \mid {\btype*} \mid \Tlist \to \Type 
\\	\btype &::=& \Int \mid \Float \mid \Void 
\\	\llabel &::=& \rPriv \mid \rPub
\\	\Tlist &::=&  {[\ ]} \mid {\Type::\Tlist} 
\\ \\	\stmt &::=& {\var = \Expr} \mid {*\x = \Expr} \mid {\stmt_1; \stmt_2} \mid {\If (\Expr)\ \stmt_1\ \Else\ \stmt_2}

\\	  && \mid {\Type\ \x (\plist)\ \{ \stmt\}} \mid {\While (\Expr)\ \stmt} 
 \mid {\{ \stmt \}} \mid \decl \mid \Expr
\\	\Expr &::=& \Expr\ \binop\ \Expr \mid {\preop\ \x} \mid \var
		\mid {\x(\Elist)} \mid \builtin 
	\\	  && \mid {(\Type)\ \Expr}  \mid {( \Expr )}  \mid \val
\\	\decl &::=&  \Type\ \var \mid {\Type\ \x ( \plist )}
\\	\var &::=& \x \mid {\x[\Expr]} 
\\	\val &::=& n \mid (\loc, \offset) \mid [{\val_0},\ {...},\ {\val_n}] \mid \Null \mid \Skip
\\ \\ 	\builtin &::=& {\Malloc(\Expr)} \mid \RT{\PMalloc(\Expr,\ \Type)} \mid {\sizeof(\Type)}
	\\	  && \mid {\free(\Expr)} \mid \RT{\pfree(\Expr)} 
	\\ && \mid \RT{\smcinput(\var, e)} \mid \RT{\smcoutput(\var, e)}
\\ 	\binop &::=&{-} \mid {+} \mid \div \mid \cdot \mid {==}  \mid {!=} \mid {<}
\\	\preop &::=& \& \mid {*} \mid ++
\\	\Elist &::=& {\Elist,\ \Expr} \mid \Expr \mid \Void
\\	\plist &::=&  \plist,\ \Type\ \var \mid {\Type\ \var} \mid \Void
\end{array}$
\captionof{figure}{Combined \vanillaC /\piccoC\ Grammar. The color \red{red} denotes terms specific to programs written in \piccoC. 
} 	
\label{Fig: \piccoC grammar}
\end{minipage}
\qquad
\begin{tabular}{l}
\begin{minipage}{0.44\textwidth}\footnotesize
$\begin{array}{r l}	
\Config ::= & \epsilon \mid (\pid, \gamma, \sigma, \DMap, \Acc, \stmt) \Mid \Config \\
\\
	\gamma ::=& [ \-\ ]\ \mid\ \gamma[\x\ \to\ (\loc,\ \Type)]							\\
	\sigma ::=& [ \-\ ]\ \mid\ \sigma [\loc\ \to\ (\byte,\ \Type,\ n,\ \PermL)  	\\
	\PermL ::=& [ \-\ ]\ \mid\ [(0,\ \llabel_0,\ \perm_0), ..., (\bytelen, \llabel_\bytelen, \perm_\bytelen)] \\
	\perm ::=& \PermF \mid \PermN		 									\\
	\bytelen ::= & \tau(\Type)\cdot n - 1		\\
	\DMap ::= & [ \-\ ] \mid \dmap::\DMap \\
	\dmap ::= & [ \-\ ] \mid ((\loc, \offset)\to(\val_1, \val_2, \tagb, \Type))::\dmap \\
\\ 
	\locLL ::= & \epsilon \mid (\pid, \locL) \Mid \locLL \\
	\locL ::= & [ \-\ ] \mid (\loc, \offset)::\locL \\
	\codeLL ::= & \epsilon \mid (\pid, \codeL) \Mid \codeLL \\
	\codeL ::= & [ \-\ ] \mid \code::\codeL 
	\end{array}
$
\captionof{figure}{Configuration: party identifier $\pid$, environment $\gamma$, memory $\sigma$, \changeMap\ $\DMap$, accumulator $\Acc$, and statement $\stmt$.}
\label{Fig: mem model}
\end{minipage}
\end{tabular}
\end{figure}

Figure~\ref{Fig: \piccoC grammar} gives the combined \vanillaC\ and \piccoC\ grammar, which is a subset of the ANSI C grammar. 
We include one dimensional arrays, branches, loops, dynamic memory allocation, and pointers. 
Arrays are zero-indexed, and it is possible to overshoot their bounds. 
We chose not to include structs or multi-dimensional arrays, as they are an extension of this core subset.
The interested reader can find the full semantics within the scope of the grammar given in Appendix~\ref{app: semantics}.
We use the color \red{red} to denote terms in the \piccoC\ grammar that are not present in \vanillaC, including annotated types ($\llabel\ \btype$, $\llabel\ \btype *$), privacy labels ($\Pub$, $\Priv$), and primitive functions ($\PMalloc$, $\pfree$) for allocation and deallocation of 
memory for private pointers.

We denote types as $\Type$, basic types as $\btype$ ($\btype *$ as a pointer type), privacy annotations as $\llabel$, and function types as $\Tlist \to \Type$ 
(where $\Tlist$ is a type list, $[\ ]$ an empty list, and $::$ list concatenation). 
Values $\val$ include numbers $n$, locations ($\loc$, $\offset$) consisting of a memory block identifier and an offset, lists of values, $\Null$, and $\Skip$ (to show a statement being reduced to completion).
Declarations include variable and function declarations, where $\plist$ is the function parameter list.
For unary operations, we include: $\&$, to obtain the address of a variable; $*$, to allow dereferencing pointers; and $++$, to allow pre-incrementing and to model a basic pointer arithmetic.

%
%
\input{mem-model}

%
%

\input{vanillaCsemantics}

\input{vanCdescr}

%
%
\input{piccoCsemantics}

\input{piccoCdescr}

%
%

\subsubsection{\piccoC\ If Else} \label{subsec: \piccoC if else}
\input{if-else-desc}

\subsection{Overshooting Memory Bounds} \label{Sec: Overshooting}
It is possible to overshoot memory bounds in both \vanillaC\ and \piccoC. 
When overshooting occurs, we read the bytes of data as the type we expected it to be (i.e., bytes containing private data accessed by a public variable would be decoded as though they are public - no encryption or decryption occurs, but computations using the variable beyond that point will be garbage). 
This ensures that no information about private data can be leaked when overshooting. 
This is discussed further in Appendix~\ref{app: array oob} for the interested reader. 
We can only prove correctness over well-aligned accesses (i.e., those that iterate only over aligned elements of the same type, as with one array spilling into a subsequent array), as these would still produce readable data that is not garbage. 
When proving noninterference, we must prove that these cases (particularly those involving private data) cannot leak any information about the private data that is affected. We discuss this in more detail in the following section.

%% file: mem-model.tex
\subsection{Memory Model} \label{subsec: mem model}
We assume the existance $n$ communicating parties, each with a separate memory.
Our memory model encodes each memory as a contiguous region of {\em blocks}, which are sequences of bytes and metadata.
We introduce an execution environment $\gamma$ and memory $\sigma$, shown in Figure \ref{Fig: mem model}.
Each block is assigned an identifier $\loc$ (to be discussed more later in this section). 
Blocks are never recycled nor cleared when they are freed. We chose this view of memory to preserve all allocated data, which, in conjunction with data-oblivious execution, represents the worst case for maintaining privacy. 
Direct memory access through pointers or manipulation of array indices allows programs to access any block for which the memory address is computable (e.g., as an offset or direct pointer access). 
To obtain the byte representations of data we leverage functions similar to CompCert, using $\Encode$ for values, $\EncodePtr$ for pointers, and $\EncodeFun$ for functions. Likewise, to obtain the human-readable data back, we use respective decode functions such as $\Decode$. We use a specialized version ($\DecodeArr$) for obtaining a specific index within an array data block.
We introduce the particulars in the following subsections.


\subsubsection{Environment} 
The environment, $\gamma$, maintains a mapping of each live variable $\x$ to its memory block identifier (where the data $\x$ is stored) and its type. 
At the start of a program, the environment is empty, i.e., $\gamma = [\-\ ]$. 
Variables that are no longer live are removed from the environment, based on scoping.  
We use the environment to facilitate the lookup of variables (i.e., for reads, writes, function calls) in memory $\sigma$.

\subsubsection{Memory Blocks and Identifiers} \label{subsec: mem block}
The memory, $\sigma$ (shown in Figure~\ref{Fig: mem model}), is a mapping of each identifier $\loc$ to its memory block, which contains 
the byte representation $\byte$ of data stored there and metadata about the block.
Metadata consists of  
a type $\Type$ associated with the block, 
the number of elements $n$ of that type stored in $\byte$, 
and a list of byte-wise permissions tuples $[(0,\ \llabel_0,\ \perm_0), ..., (\bytelen,\ \llabel_\bytelen,\ \perm_\bytelen)]$, where $\bytelen = \tau(\Type)\cdot n - 1$ and function $\tau$ provides the size of the given type in bytes. 
A new memory block identifier is obtained from function $\phi$. 
These identifiers are monotonically increasing with each allocation. 
Every block is added to memory $\sigma$ on allocation, and is never cleared of data nor removed from $\sigma$ upon deallocation. 
Metadata cannot be accessed or modified directly by the program (the semantic rules control modification).
A memory block can be of an arbitrary size, which is constant and determined at allocation (with the exception of private pointers, to be discussed later in subsection~\ref{subsec: picco ptr eval}). 
We represent a memory location as a two-tuple of a memory block identifier and an offset. This allows us to use pointers to refer to any arbitrary memory location, as in C.

\subsubsection{Permissions} 
\label{subsec: permissions}
A permission $\perm$ can either be $\PermF$ (i.e., can be written to, read from, etc.) or 
$\PermN$ (i.e., already freed). 
These byte-wise permissions are modeled after a subset of those used by CompCert, and we extend their permission model by including a privacy label.
Each memory block has a list of permission tuples, one for each byte of data 
stored in that block. 
A permission tuple consists of the position of the byte that it corresponds to, and the privacy label $\llabel$ and permission $\perm$ for that byte of data. 
These permissions are important in reading and writing data to memory, especially when it comes to overshooting arrays and other out-of-bound memory accesses possible through the use of pointers. 
In particular, permissions allow us to keep track of deallocated memory (e.g. a block freed - note that the memory stored in the block itself is not overwritten or cleared and the block can still be accessed indirectly through direct
memory manipulations). 
All permission tuples corresponding to a memory block of a function type will have public privacy labels, as the instructions for a function are accessible from the program itself. 
Those for a normal variable or an array will have privacy labels corresponding to their type (i.e., public for public types, private for private types); those for pointers are more complex and will be discussed later in subsection~\ref{subsec: picco ptr eval}.

\subsubsection{Malloc and Free}
Allocation of dynamic memory in C is provided by \TT{malloc}, which takes a number of bytes as its argument. 
When \Code{malloc} is called, 
a new memory block with identifier $\loc$ is obtained, initialized as a void type of the given size, and returned. 
This block then needs to be cast to the desired type. 
However, when dealing with private data, the programmer is unlikely to know the internal representation and the size of the private data types. 
For that reason, when allocating memory for private data, we adopt PICCO's \TT{pmalloc} functionality which takes two arguments: the type and the number of elements of that type to be allocated. 
The semantics then handle sizing the new memory block for the given private type. 
When \Code{free} or \Code{pfree} is called, if the argument is a variable of a pointer type, the permissions for all bytes of the location the pointer refers to will be set to $\PermN$, but the data stored there will not be erased. 
When \Code{pfree} is called with a pointer with a single location, it behaves identically to \Code{free}. 
The use of \TT{pfree} with multiple locations is a bit more involved, to be discussed later in Subsection~\ref{subsec: picco mem alloc/dealloc}. 
It is important to note that memory allocation and deallocation are public side effects, and therefore are {\em not allowed} within private-conditioned branches.

\subsubsection{Public vs. Private Blocks}
To distinguish between public and private blocks, we assume that private blocks will be encrypted 
and we will use basic private primitives implementing specific operations to manage them. 
For modeling purposes, 
these primitives can decrypt the required blocks, perform 
the operations they are meant to implement, and encrypt the result. 
In our model, a program can also access private blocks by means of standard non-private operations or through of pointers. In this case, the operation will just interpret the encrypted value as a public value. This approach gives us a conceptual distinction between a \emph{concrete memory} and its corresponding \emph{logical content}, i.e. public values and values of the private data prior to encryption. Our model as described in the next section will work on concrete memories, but in Section~\ref{sec: noninterference}, for the proof of noninterference, it will be convenient to refer to the logical content of a memory. We will use the notation $\sigma\ell$ to denote the logical content of $\sigma$.

%% file: vanillaCsemantics.tex
\begin{figure*}
\footnotesize
\begin{tabular}{l}
Multiparty Binary Operation  \\
  	\inferrule{\begin{array}{l l}
		\begin{array}{l}
		\qquad((\pidA, \hgamma, \hsigma,\ \bsq, \bsq, \hExpr_1)\ \Mid ...\Mid
		(\pidZ, \hgamma, \hsigma,\ \bsq, \bsq, \hExpr_1))
			\crcr\Veval_{\codeVLL_1} 
			((\pidA, \hgamma, \hsigma_1, \bsq, \bsq, \hn_1) \Mid ...\Mid
			(\pidZ, \hgamma, \hsigma_1, \bsq, \bsq, \hn_1))
		\crcr \qquad ((\pidA, \hgamma, \hsigma_1, \bsq, \bsq, \hExpr_2)\ \Mid ...\Mid
			(\pidZ, \hgamma, \hsigma_1, \bsq, \bsq, \hExpr_2))
			\crcr\Veval_{\codeVLL_2} 
			((\pidA, \hgamma, \hsigma_2, \bsq, \bsq, \hn_2)\Mid ...\Mid
			(\pidZ, \hgamma, \hsigma_2, \bsq, \bsq, \hn_2))
		\end{array}
	& \begin{array}{l}
		 \hn_1 \binop\ \hn_2 = \hn_3
		\crcr \binop\in\{\cdot, +, -, \div\}
	\end{array}\end{array}}
	{\begin{array}{l}
	((\pidA, \hgamma, \hsigma,\ \bsq, \bsq, \hExpr_1 \binop\ \hExpr_2)\Mid ...\Mid
	(\pidZ, \hgamma, \hsigma,\ \bsq, \bsq, \hExpr_1 \binop\ \hExpr_2)) 
		\Veval_{\codeVLL_1\addC\codeVLL_2\addC\codeVM{mpb}} 
		\crcr((\pidA, \hgamma, \hsigma_2, \bsq, \bsq, \hn_3) 
			\qquad\ \Mid ...\Mid 
		(\pidZ, \hgamma, \hsigma_2, \bsq, \bsq, \hn_3))
		\end{array}}
\\ \\
 Write \\
\inferrule{\begin{array}{l}
		((\pid, \hgamma, \hsigma, \bsq, \bsq, \hExpr)\Mid \hConfig)  \Veval_{\codeVLL} 
			((\pid, \hgamma, \hsigma_1, \bsq, \bsq, \hn)\Mid \hConfig_1)
		\crcr \hgamma(\hx) = (\hloc, \hbtype)  
		\qq \Update(\hsigma_1, \hloc, \hn, \hbtype) = \hsigma_2
	\end{array}}
	{((\pid, \hgamma, \hsigma, \bsq, \bsq, \hx = \hExpr) \Mid \hConfig)
		\Veval_{\codeVLL\addC[\codeVS{w}]} 
		((\pid, \hgamma, \hsigma_2, \bsq, \bsq, \Skip)\Mid \hConfig_1)}
\\ \\
Malloc 	\\ 
\inferrule{\begin{array}{l}
		((\pid, \hgamma, \hsigma, \bsq, \bsq, \hExpr) \Mid \hConfig) \Veval_{\codeVLL} 
			((\pid, \hgamma, \hsigma_1, \bsq, \bsq, \hat{n})\Mid \hConfig_1)
		\qq \hloc = \phi()
		\crcr \hsigma_2 = \hsigma_1 \big[\hloc \to \big(\Null, \Void*, \hat{n}, 
			\PermL(\PermF, \Void*, \Pub, \hn)\big) \big] 
	\end{array} }
	{((\pid, \hgamma, \hsigma, \bsq, \bsq, \Malloc (\hExpr))\Mid \hConfig) 
		\Veval_{\codeVLL\addC[\codeVS{mal}]} 
		((\pid, \hgamma, \hsigma_2, \bsq, \bsq, (\hloc, 0))\Mid \hConfig_1)}
\\ \\
Multiparty Free \\
	\inferrule{\begin{array}{l}
		\hgamma(\hx) = (\hloc, \hbtype*) 
		\qq \hsigma(\hloc) = (\hbyte, {\hbtype*}, 1, \PtrPermL(\PermF, \hbtype*, \Pub, 1))
		\crcr \DecodePtr({ \hbtype*}, 1, \hbyte) = [1, [({\hloc_1}, 0)], [1], 1]
		\crcr \SelectFreeable(\hgamma, [({\hloc_1}, 0)], [1], \hsigma) = 1
		\qq \Free(\hsigma, \hloc_1) = \hsigma_1 
	\end{array}}
	{\begin{array}{l}
	((\pidA, \hgamma, \hsigma, \bsq, \bsq, \free (\x)) \Mid ... \Mid 
	(\pidZ, \hgamma, \hsigma, \bsq, \bsq, \free (\x)))
		\Veval_{\codeVM{fre}} \crcr
		((\pidA, \hgamma, \hsigma_1, \bsq, \bsq, \Skip)\Mid ... \Mid 
		(\pidZ, \hgamma, \hsigma_1, \bsq, \bsq, \Skip))
		\end{array}}
\\ \\
Pointer Declaration 	\\
\inferrule{\begin{array}{l l}
	\begin{array}{l}
		(\hType = \hbtype *)
		\crcr \hloc = \phi() 
		\crcr \hgamma_1 = \hgamma[\hx \to (\hloc, \hType)]
		\end{array}
		&\begin{array}{l}
		\getIndirection(*) = \hindir \crcr
		\hbyte = \EncodePtr(\hType*, [1, [(\hlocDefault, 0)], [1], \hindir])
		\crcr \hsigma_1 = \hsigma[\hloc \to (\hbyte, \hType, 0, \PtrPermL(\PermF, \hType, \Pub, 0))]
	\end{array}\end{array}}					
	{((\pid, \hgamma, \hsigma, \bsq, \bsq, \hType \hx) \Mid \hConfig)
		\Veval_{[\codeVS{dp}]} 
		((\pid, \hgamma_1, \hsigma_1, \bsq, \bsq, \Skip)\Mid \hConfig)}
\\ \\ 
Multiparty Array Read \\ 
\inferrule{\begin{array}{l}
		\hgamma(\hx) = (\hloc, \Const\ \hbtype*) \crcr
		((\pidA, \hgamma, \hsigma, \bsq, \bsq, \hExpr) \Mid ... \Mid 
		(\pidZ, \hgamma, \hsigma, \bsq, \bsq, \hExpr)) \Veval_{\codeVLL} 
			((\pidA, \hgamma, \hsigma_1, \bsq, \bsq, \hind)\Mid ... \Mid 
			(\pidZ, \hgamma, \hsigma_1, \bsq, \bsq, \hind)) 
		\crcr \hsigma_1(\hloc) = (\hbyte, \Const\ \hbtype*, 1, \PtrPermL(\PermF, \Const\ \hbtype*, \Pub, 1)) 
		\crcr \DecodePtr(\Const\ \hbtype*, 1, \hbyte) = [1, [(\hloc_1, 0)], [1], 1] 
		\crcr \hsigma_1(\hloc_1) = (\hbyte_1, \hbtype, \hnl, \ArrPermL(\PermF, \hbtype, \Pub, \hnl))  
		\qq 0 \leq \hind \leq \hnl - 1 
		\crcr \DecodeArr({\hbtype},\ \hind,\ {\hbyte_1}) = \hn_{\hind} 
	\end{array}}
	{\begin{array}{l}
	((\pidA, \hgamma, \hsigma,\ \bsq, \bsq, \hx[\hExpr]) \Mid ... \Mid 
	(\pidZ, \hgamma, \hsigma,\ \bsq, \bsq, \hx[\hExpr]))
		\Veval_{\codeVLL\addC\codeVM{mpra}} \crcr
		((\pidA, \hgamma, \hsigma_1, \bsq, \bsq, \hn_{\hind})\ \ \Mid ... \Mid 
		(\pidZ, \hgamma, \hsigma_1, \bsq, \bsq, \hn_{\hind}))
		\end{array}}
\end{tabular}
\caption{\vanillaC\ semantic rules.}
\label{Fig: \vanillaC sem rules}
\end{figure*}

%% file: vanCdescr.tex
\subsection{\vanillaC\ Semantics} \label{van C descr}

%
%
In order to facilitate the correspondence between the \vanillaC\ and \piccoC\ semantics, we model our semantics using big-step evaluation judgements and define our C semantics with respect
to multiple {\em non interacting} parties that evaluate the same program. 
These judgements are defined over a six-tuple configuration $\Config = ((\pid, \gamma, \sigma, \DMap, \Acc, \stmt) \Mid \Config_1)$, where each rule is a reduction from one configuration to a subsequent. 
In the semantics, we denote the party identifier $\pid$;
the environment as $\gamma$; 
memory as $\sigma$;
a mapping structure for location-based tracking of changes at each level of nesting of private-conditioned branches $\DMap$;
the level of nesting of private-conditioned branches as $\Acc$; and
a big-step evaluation of a statement $\stmt$ to a value $\val$ using $\eval$.  
We use a $\hat{\ }$ and $\Veval$ to distinguish the \vanillaC\ semantics from those we use in the next section for \piccoC, as well as $\bsq$ in \vanillaC\ as a placeholder for $\Acc$ and $\DMap$ to maintain the same shape of configurations as that of \piccoC\ used in the next section. 
We annotate each evaluation with party-wise lists of the evaluation codes of all rules that were used during the execution of the rule (i.e., $\Veval_{\codeVLL}$) to facilitate reasoning over evaluation trees. 
We extend the list concatenation operator $::$ to also work over party-wise lists such as $\codeVLL$, defining its behavior as concatenating the lists within $\codeVLL$ by party (i.e., $[(1, \codeVL^1_1), (2, \codeVL^2_1)]::[(1,\codeVL^1_2)] = [(1, \codeVL^1_1::\codeVL^1_2), (2, \codeVL^2_1)]$). 
The assertions in each semantic rule can be read in sequential order, from left to right and top to bottom. 

We present a subset of  the \vanillaC\ semantics  rules in Figure~\ref{Fig: \vanillaC
  sem rules}, focusing on the most interesting rules. Specifically, we present rules for arrays and pointers
which we will compare with the rules for the
\piccoC\ semantics in Section~\ref{smc C descr}. Because the semantics
rules are mostly standard, we will describe one rule to familiarize 
the reader with our notation. 

Rule Multiparty Binary Operation is an example of a rule for the evaluation of a binary operation (comparison operations are handled separately). In \vanillaC, multiparty rules occur at the same time in all parties, but without any communication between the parties. 
We have the starting state $((\pidA,$ $\hgamma,$ $\hsigma,$ $\bsq,$ $\bsq,$ $\hExpr_1\ \binop\ \hExpr_2)$ $\Mid ...\Mid $ $(\pidZ,$ $\hgamma,$ $\hsigma,$ $\bsq,$ $\bsq,$ $\hExpr_1\ \binop\ \hExpr_2))$, 
with all parties at current environment $\hgamma$, current memory $\hsigma$, and the starting statement $\hExpr_1\ \binop\ \hExpr_2$.
First, all parties evaluate expression $\hExpr_1$, using the
current environment and memory states, 
resulting in environment $\hgamma$, memory
$\hsigma_1$, and number $\hn_1$.  We repeat this for
$\hExpr_2$. 
We then evaluate $\hn_1\ \binop\ \hn_2 = \hn_3$ here, and we will return $\hn_3$.  
The end state is then
$((\pidA, \hgamma, \hsigma_2, \bsq, \bsq, \hn_3)$ $\Mid ...\Mid$ $(\pidZ, \hgamma, \hsigma_2, \bsq, \bsq, \hn_3))$.

%% file: piccoCsemantics.tex
\begin{figure*}\footnotesize
\begin{tabular}{l}
Multiparty Binary Operation	\\ 
\inferrule{\begin{array}{l l}
	\begin{array}{l}
		\qquad((\pidA, \gamma^{\pidA}, \sigma^{\pidA}, \DMap^{\pidA}, \Acc, \Expr_{1}) \Mid ...\Mid (\pidZ, \gamma^{\pidZ}, \sigma^{\pidZ}, \DMap^{\pidZ}, \Acc, \Expr_{1})) 
			\crcr\Deval{\locLL_1}{\codeLL_1} 
			((\pidA, \gamma^{\pidA}, \sigma^{\pidA}_{1}, \DMap^{\pidA}_{1}, \Acc, \n^{\pidA}_{1}) \Mid ...\Mid (\pidZ, \gamma^{\pidZ}, \sigma^{\pidZ}_{1}, \DMap^{\pidZ}_{1}, \Acc, \n^{\pidZ}_{1}))
		\crcr
		\qquad((\pidA, \gamma^{\pidA}, \sigma^{\pidA}_{1}, \DMap^{\pidA}_{1}, \Acc, \Expr_{2}) \Mid ...\Mid (\pidZ, \gamma^{\pidZ}, \sigma^{\pidZ}_{1}, \DMap^{\pidZ}_{1}, \Acc, \Expr_{2})) 
			\crcr\Deval{\locLL_2}{\codeLL_2} 
			((\pidA, \gamma^{\pidA}, \sigma^{\pidA}_{2}, \DMap^{\pidA}_{2}, \Acc, \n^{\pidA}_{2})\Mid ...\Mid (\pidZ, \gamma^{\pidZ}, \sigma^{\pidZ}_{2}, \DMap^{\pidZ}_{2}, \Acc, \n^{\pidZ}_{2}))
		\crcr
		\MPC{b}(\binop, [\n^\pidA_1, ..., \n^\pidZ_1], [\n^\pidA_2, ..., \n^\pidZ_2]) = (\n^{\pidA}_{3}, ..., \n^{\pidZ}_{3})
	\end{array}
	&\begin{array}{l}
		\{(\Expr_{1}, \Expr_{2}) \isPriv \gamma^\pid\}^{\pidZ}_{\pid = \pidA}
		\crcr \binop\in\{\cdot, +, -, \div\}
	\end{array}\end{array}}
	{\begin{array}{l}
	((\pidA, \gamma^{\pidA}, \sigma^{\pidA}, \DMap^{\pidA}, \Acc, {\Expr_{1}\ \binop\ \Expr_{2}})\Mid ...\Mid (\pidZ, \gamma^{\pidZ}, \sigma^{\pidZ}, \DMap^{\pidZ},\Acc, {\Expr_{1}\ \binop\ \Expr_{2}})) 
		\Deval{\locLL_1 \addL \locLL_2}{\codeLL_1 \addC \codeLL_2 \addC \codeMP{mpb}} 
		\crcr((\pidA, \gamma^{\pidA}, \sigma^{\pidA}_{2}, \DMap^{\pidA}_{2}, \Acc, \n^{\pidA}_{3}) 
			\qquad\-\ \-\ \-\ \Mid ...\Mid 
		(\pidZ, \gamma^{\pidZ}, \sigma^{\pidZ}, \DMap^{\pidZ}_{2}, \Acc, \n^{\pidZ}_{3})) \end{array}}
\\ \\
Write Private Variable Public Value		\\
\inferrule{\begin{array}{l l}
		\begin{array}{l}
		(\Expr) \isPub \gamma \crcr
		\crcr \gamma(\x) = (\loc, {\Priv\ \btype})  
		\end{array}
		& \begin{array}{l}
		((\pid, \gamma, \sigma, {\DMap}, \Acc, \Expr) \Mid  \Config)  
			\Deval{\locLL_1}{\codeLL_1}  ((\pid, \gamma, \sigma{_1}, {\DMap_1}, \Acc, \n) \Mid  \Config_1) 
		\crcr 
		\Update(\sigma{_1}, \loc, \Encrypt(\n), \Acc, \Priv\ \btype) = \sigma{_2}
	\end{array}\end{array}}
	{((\pid, \gamma,\ \sigma,\ {\DMap},\ \Acc,\ {\x = \Expr}) \Mid  \Config)\ 
		\Deval{\locLL_1 \addL (\pid, [(\loc, 0)])}{\codeLL_1 \addC \codeSP{w2}}  
		((\pid, \gamma,\ \sigma{_2},\ {\DMap_1},\ \Acc,\ \Skip) \Mid  \Config_1)}
\\ \\
Private Malloc \\
\inferrule{\begin{array}{l}
		(\Expr) \isPub \gamma \qq
		\Acc = \AccZ \qq
		(\Type = {\Priv\ \btype*}) \lor (\Type = {\Priv\ \btype}) \crcr
		\loc = \phi()	\qq
		((\pid, \gamma,\ \sigma,\ {\DMap},\ \Acc,\ \Expr) \Mid  \Config)\ 
			\Deval{\locLL_1}{\codeLL_1}  ((\pid, \gamma,\ \sigma{_1},\ {\DMap},\ \Acc,\ {n}) \Mid  \Config_1) 
		\crcr \sigma_2 = \sigma{_1} \big[\loc \to \big(\Null,\ \Void*,\ {n\cdot\tau(\Type)},\ \PermL(\PermF, \Void*, \Priv, \n\cdot\tau(\Type))\big)\big]	
	\end{array}}
	{\begin{array}{l}
	((\pid, \gamma,\ \sigma,\ {\DMap},\ \Acc,\ {\PMalloc (\Expr,\ \Type)}) \Mid  \Config) 
		\Deval{\locLL_1 \addL (\pid, [(\loc, 0)])}{\codeLL_1 \addC \codeSP{malp}} 
		((\pid, \gamma,\ \sigma{_2},\ {\DMap},\ \Acc,\ (\loc, 0)) \Mid  \Config_1)
	\end{array} }
\\ \\
Multiparty Private Free	\\
\inferrule{ \begin{array}{l}
		\{\gamma^\pid(\x) = (\loc^\pid,\ {\Priv\ \btype*})\}^{\pidZ}_{\pid = \pidA}  
		\qq \Acc = \AccZ  \qq
		(\btype = \Int) \lor (\btype = \Float) \crcr
		\{\sigma^\pid(\loc^\pid) = (\byte^\pid, \Priv\ \btype*, \nl, \PtrPermL(\PermF, \Priv\ \btype*, \Priv, \nl))\}^{\pidZ}_{\pid = \pidA} \crcr
		\{[\nl,\ \locL^\pid\ \tagbL^\pid,\ \indir] = \DecodePtr({\Priv\ \btype*},\ \nl,\ \byte^\pid)\}^{\pidZ}_{\pid = \pidA}
		\qq \{\nl > 1\}^{\pidZ}_{\pid = \pidA}  \crcr
		\If(\indir>1) \{\Type = \Priv\ \btype* \}\ \Else\ \{ \Type = \Priv\ \btype \}
		\crcr
		\{\SelectFreeable(\gamma^\pid, \locL^\pid, \tagbL^\pid, \sigma^\pid) = 1\}^{\pidZ}_{\pid = \pidA} 
		\crcr
		\{\forall (\loc^\pid_m, 0) \in \locL.\quad \sigma^\pid(\loc^\pid_m) = (\byte^\pid_m, \Type, n, \PermL(\PermF, \Type, \Priv, n))\}^{\pidZ}_{\pid = \pidA}
		\crcr
		\PFree([[\byte^\pidA_0, ..., \byte^\pidA_{\nl-1}], ..., [\byte^\pidZ_0, ..., \byte^\pidZ_{\nl-1}]], [\tagbL^\pidA, ...\tagbL^\pidZ]) 
			\crcr\qq= ([[\byte'^\pidA_0, ..., \byte'^\pidA_{\nl-1}], ..., [\byte'^\pidZ_0, ..., \byte'^\pidZ_{\nl-1}]], [\tagbL'^\pidA, ..., \tagbL'^\pidZ])
		\crcr
		\{\UpdateBytesFree(\sigma^\pid, \locL^\pid, [\byte'^\pid_0, ..., \byte'^\pid_{\nl-1}]) = \sigma^\pid_1\}^{\pidZ}_{\pid = \pidA}
		\crcr
		\{\sigma^\pid_2 = \UpdatePtrLocs(\sigma^\pid_1, \locL^\pid[1:\nl-1], \tagbL^\pid[1:\nl-1], \locL^\pid[0], \tagbL^\pid[0])\}^{\pidZ}_{\pid = \pidA}
	\end{array}}
	{\begin{array}{l}((\pidA, \gamma^{\pidA}, \sigma^{\pidA}, \DMap^\pidA, \Acc, {\pfree (\x)})\Mid ...\Mid (\pidZ, \gamma^{\pidZ}, \sigma^{\pidZ}, \DMap^\pidZ, \Acc, {\pfree (\x)})) 
		\Deval{(\pidA, \locL^\pidA) \Mid ... \Mid (\pidZ, \locL^\pidZ)}{\codeMP{mpfre}} 
		\crcr((\pidA, \gamma^{\pidA}, \sigma^{\pidA}_2, \DMap^\pidA, \Acc, \Skip)\quad\-\ \ \Mid ...\Mid 
		(\pidZ, \gamma^{\pidZ}, \sigma^{\pidZ}_2, \DMap^\pidZ, \Acc, \Skip)) \end{array}}
\\ \\ 
Private Pointer Declaration \\ 
\inferrule{\begin{array}{l}\begin{array}{l l}
		\getIndirection(*) = \indir\qq
		&((\Type = {\btype *}) \lor (\Type = {\Priv\ \btype*})) \land ((\btype = \Int) \lor (\btype = \Float))
		\end{array}\crcr \begin{array}{l l l} 
		\loc = \phi() \qq
		& \gamma{_1} = \gamma[\x \to (\loc, \Priv\ \btype*)] 	
		\end{array}\crcr\begin{array}{l}
		\byte = \EncodePtr(\Priv\ \btype*, [1, [(\locDefault, 0)], [1], \indir])	\crcr
		\sigma{_1} = \sigma[\loc \to (\byte, \Priv\ \btype*, 1, \PtrPermL(\PermF, \Priv\ \btype*, \Priv, 1))]
	\end{array}\end{array}}					
	{((\pid, \gamma,\ \sigma,\ {\DMap},\ \Acc,\ {\Type\ \x}) \Mid  \Config)\ 
		\Deval{(\pid, [(\loc, 0)])}{\codeSP{dp1}}  ((\pid, \gamma{_1},\ \sigma{_1},\ {\DMap},\ \Acc,\ \Skip) \Mid  \Config)}
\\ \\
Multiparty Array Read Private Index \\
	\inferrule{\begin{array}{l}
		\qquad((\pidA, \gamma^{\pidA}, \sigma^{\pidA}, \DMap^\pidA, \Acc, \Expr)\ \Mid ...\Mid (\pidZ, \gamma^{\pidZ}, \sigma^{\pidZ}, \DMap^\pidZ, \Acc, \Expr)) 
		\qquad\ \
		\{(\Expr) \isPriv \gamma^\pid\}^{\pidZ}_{\pid = \pidA} 
		\crcr		
			\crcr\Deval{\locLL_1}{\codeLL_1} ((\pidA, \gamma^{\pidA}_{}, \sigma^{\pidA}_{1}, \DMap^{\pidA}_{1}, \Acc, \ind^{\pidA})\Mid ...\Mid (\pidZ, \gamma^{\pidZ}_{}, \sigma^{\pidZ}_{1}, \DMap^{\pidZ}_{1}, \Acc, \ind^{\pidZ}))
		\qquad\ \{\gamma^\pid(\x) = (\loc^\pid, \Const\ \llabel\ \btype*)\}^{\pidZ}_{\pid = \pidA}
		\crcr
		\{\sigma^\pid_1(\loc^\pid) = (\byte^\pid,\ {\llabel\ \Const\ \btype*}, 1, 
			\PtrPermL(\PermF, {\llabel\ \Const\ \btype*}, \llabel, 1))\}^{\pidZ}_{\pid = \pidA} 
		\crcr \{\DecodePtr({\llabel\ \Const\ \btype*},\ 1,\ \byte^\pid) 
			= [1,\ [({\loc^\pid_1}, 0)],\ [1],\ 1]\}^{\pidZ}_{\pid = \pidA}  
		\crcr \{\sigma^\pid_1({\loc^\pid_1}) = ({\byte^\pid_1}, {\llabel\ \btype}, {\nl}, 
									\ArrPermL(\PermF, \llabel\ \btype, \llabel, {\nl}))\}^{\pidZ}_{\pid = \pidA}
		\crcr 
			\{\forall \ind \in \{0...\nl-1\} \quad \DecodeArr({\llabel\ \btype}, \ind, {\byte^\pid_1}) =  \n^\pid_\ind\}^{\pidZ}_{\pid = \pidA}
		\crcr
		\MPC{ar}((\ind^\pidA, [\n^{\pidA}_{0}, ..., \n^{\pidA}_{\nl-1}]), ..., (\ind^\pidZ, [\n^{\pidZ}_{0}, ..., \n^{\pidZ}_{\nl-1}])) = (\n^{\pidA}, ..., \n^{\pidZ})
		\qq \{(\n^\pid) \isPriv \gamma^\pid\}^{\pidZ}_{\pid = \pidA} 
		\crcr \locLL_2 = (\pidA, [(\loc^\pidA, 0), (\loc^\pidA_1, 0), ..., (\loc^\pidA_1, \nl-1)])\Mid ... \Mid (\pidZ, [(\loc^\pidZ, 0), (\loc^\pidZ_1, 0), ..., (\loc^\pidZ_1, \nl-1)])
	\end{array}}
	{\begin{array}{l}
		((\pidA, \gamma^{\pidA}, \sigma^{\pidA}, \DMap^\pidA, \Acc, \x[\Expr])\Mid ...\Mid (\pidZ, \gamma^{\pidZ}, \sigma^{\pidZ}, \DMap^\pidZ, \Acc, \x[\Expr])) 
			\Deval{\locLL_1 \addL \locLL_2}{\codeLL_1\addC \codeMP{mpra}} 
			\crcr((\pidA, \gamma^{\pidA}_{}, \sigma^{\pidA}_{1}, \DMap^{\pidA}_{1}, \Acc, \n^{\pidA})\-\ \Mid ...\Mid (\pidZ, \gamma^{\pidZ}_{}, \sigma^{\pidZ}_{1}, \DMap^{\pidZ}_{1}, \Acc, \n^{\pidZ})) \end{array}}
\end{tabular}
\caption{\piccoC\ semantic rules.}
\label{Fig: \piccoC sem rules}
\end{figure*}

%% file: piccoCdescr.tex
\subsection{\piccoC\ Semantics} \label{smc C descr}

The \piccoC\ semantics are defined over multiple {\em interacting} parties.
The \piccoC\ semantics used to define the behavior of 
parties are mostly standard, with \emph{non-interactive} semantic rules identical to those of \vanillaC\ semantics aside from additional assertions over the privacy labels of data and properly managing the private data. 
A few notable exceptions are \emph{interactive} SMC operations (and in general operations over private values)
and the private-conditioned \Code{if else} statement, discussed in later in this section.    
To prevent leakage from within private-conditioned branches, we restrict all public side effects (i.e., the use of functions with public side effects, allocation and deallocation of memory, and any modifications to public variables). 
Additionally, in the case of pointer dereference write and array write statements, we have an additional check for when this occurs within a private-conditioned branch, as we need to perform additional analysis to ensure the location being written to is tracked properly due to the potential for the pointer's location being modified or an out-of-bounds array write. 
To enforce these restrictions, we use the assertion $\Acc = \AccZ$ within each restricted rule -- as the accumulator $\Acc$ is incremented at each level of nesting of a private-conditioned branch, this will result in a runtime failure. 
We annotate each evaluation with party-wise lists of the evaluation codes $\codeLL$ of all rules that were used during the execution of the rule (i.e., $\Deval{\locLL}{\codeLL}$) in order to keep an accurate evaluation tree, and party-wise lists of locations accessed $\locLL$ in order to show data-obliviousness (i.e., that given the same program and public data, we will always access the same set of locations).

%
%
\subsubsection{Basic Evaluation} \label{subsec: picco basic eval}
To better illustrate the correspondence between \piccoC\  and \vanillaC\  let us consider the Multiparty Binary Operation rule.
\piccoC\ rule Multiparty Binary Operation asserts that one of the given binary operators ($\cdot, +,-,\div$) is used and additionally that either expression contains private data with relation to the environment. 
We use the notation $(\Expr_{1}, \Expr_{2}) \isPriv \gamma$ to show this relation, and  
notation $\{...\}^{\pidZ}_{\pid = \pidA}$ to show that all parties will ensure that property holds locally.
We then use the multiparty protocol $\MPC{b}$, passing the given binary operator and the current values of $\n^\pid_1$ and $\n^\pid_2$ for each party $\pid$. 
This protocol will dictate how communication occurs and what data is exchanged between parties. 
We receive $\n^\pid_3$ as the result for each party, which we then return appropriately.   We assume that the protocol is implemented correctly (i.e. provided by the underlying SMC cryptographic library) and define this assumption formally, its impact on our noninterferences proof, and how to reason if a library adheres to our assumption later in Section~\ref{sec: noninterference}. 
Within the multiparty rules, each party maintains control of their own data, only sharing it with other parties in the ways dictated by the multiparty protocols. We choose to show the execution of the entire computational process here in order to emphasize what data is involved, and that each of the parties will take part in this computation. 

In rule Write Private Variable Public Value, we assert that $\Expr$ contains only public data,  
then evaluate $\Expr$ to $n$. We look up $\x$ in $\gamma$, asserting that it is a $\Priv\ \btype$ at location $\loc$. 
When we call $\Update$ to store this value to memory, we must pass $\Encrypt(n)$ as $\x$ is private and we are assigning it the public value $n$.

%
%
\subsubsection{Pointer Evaluations} \label{subsec: picco ptr eval}
In order to maintain data-oblivious execution, we need to allow storing multiple locations for pointers when they are modified within a private-conditioned branching statement. To achieve this, 
the structure of the data stored by pointers is as follows: the number $\nl$ of locations being pointed to; a list of $\nl$ locations being pointed to; a list of $\nl$ tags; and the level of indirection of the pointer. 
The privacy labels of the byte-wise permissions corresponding to the number $\nl$ of locations to which the pointer refers, the list of $\nl$ locations, and the level of indirection will always be public, as it is visible to an observer of memory the number and which locations are touched by a pointer, and the level of indirection is visible in the source program.
The privacy labels of the permissions corresponding to the tags of public pointers will always be public; for private pointers, they will be private, as these protect an observer of memory from being able to tell which location is the true location. 
%
In rule Private Pointer Declaration, we assert that the type of the variable being declared is an int or float pointer that is either declared as private or missing a privacy label (i.e., $\btype*$) and therefore is assumed to be private. 
As with the \vanillaC\ pointer declaration, we check the level of indirection and add the appropriate mappings to the environment and memory.

%
%
\subsubsection{Memory Allocation and Deallocation} \label{subsec: picco mem alloc/dealloc}
When allocating private memory, we provide the \TT{pmalloc} builtin function to internally handle obtaining the size of the private type; the programmer to only needs to know how many elements of the given type they desire to allocate.
In rule Private Malloc, we assert that the given type is either private int or private float, as this function only handles those types, and that the accumulator $\Acc$ is $\AccZ$ (i.e., we are not inside an \Code{if else} statement branching on private data, as this function causes public side effects). 
Then we evaluate $\Expr$ to $n$ and obtain the next open memory location $\loc$ from $\phi$. We add to $\sigma{_1}$ the new mapping from $\loc$ to the tuple of a $\Null$ set of bytes; the type $\Type$; the size $n$; and a list of $\Priv$, $\PermF$ permissions. 
As with public \Code{malloc}, we return the new location, $(\loc, 0)$. 

When deallocating private memory, we provide the \TT{pfree} builtin function to handle private pointers potentially having multiple locations. 
In the case of a single location, it behaves identically to Public Free; however, with multiple locations, we need to deterministically free a single location (which may or may not be the true location that was intended to be freed) to maintain data-obliviousness. We describe this case in more detail here.   
In rule Multiparty Private Free, we assert that $\x$ is a private pointer of type int or float, we are not inside a private-conditioned branch ($\Acc$ is $\AccZ$, as this rule causes public side effects), and that the number of locations the pointer refers to ($\nl$) is greater than 1 for all parties. 
We then assert that \emph{all} locations referred to by $\x$ are freeable (i.e., they are all memory blocks that were allocated via malloc) and proceed to retrieve the data that is stored for each of these locations.
This data and the tag lists are then passed to $\PFree$, as this is what we will need in order to privately free a location without revealing if it was the true location. 

To accomplish this, we must free one location based on publicly available information, regardless of the true location of the pointer. For that reason, and without loss of generality, we free the first location, $\loc_0$. 
Since $\loc_0$ may not be the true location and may be in use by other pointers, we need to do additional computation to maintain correctness without disclosing whether or not this was the true location. 
In particular, if $\loc_0$ is not the true location, we preserve the content of $\loc_0$ by obliviously copying it to the pointer's true location prior to freeing. 
This behavior is defined in function $\PFree$, and follows the strategy suggested in~\cite{Zhang18}.
$\PFree$ returns the modified bytes and tag lists. $\UpdateBytesFree$ then updates these in their corresponding locations in memory and marks the permissions of $\loc_0$ as $\PermN$ (i.e., this block has been freed). 
The remaining step is to update other pointers that stored $\loc_0$ on their lists to point to the updated location instead of $\loc_0$, which is accomplished by $\UpdatePtrLocs$.

%
%
\subsubsection{Array Evaluations} \label{subsec: picco arr eval}
With array evaluations, all evaluations that use a public index behave nearly identically to those over public data. 
The difference is that we have an additional check within array writes at a public index to see if we are in a private-conditioned branch; if so, we must ensure we properly track the modification made (this is because a public index that is not hard-coded could have lead to an out-of-bounds array write). 
We will discuss this further in the following section. 
When we have a private index, it is necessary to hide which location we a reading from or writing to to maintain data-obliviousness. One such semantic rule is Multiparty Array Read Private Index.
%
Here, we assert that the privacy label of $\Expr$ is private as this rule handles private indexing into a public array. 
We assert that $\x$ is a public array of type int or float, as we must return a private value when using a private index. 
Then we evaluate $\Expr$ to the private index $\ind$, and perform lookups to obtain the data of the array. 
Now, because we have a private index, we must obtain the value without revealing which location we are taking the value from. 
To do this, we use multiparty protocol $\MPC{ar}$, which returns a private number containing the value from the desired location.  
It is important to note here that even if the private index is beyond the bounds of the array, we do not access beyond the elements within the array, as that would reveal information about the private index. 
An example of how this protocol can be implemented is to iterate over all values stored in the array; at each value, we encrypt the current index number $m$, privately compare it to $\ind$, and perform a bitwise \Code{and} operation over this and the encrypted value $\n_m$ stored at index $m$. We perform a bitwise \Code{or} operation over each such value obtained from the array to attain our final encrypted value $\n$, which is returned. 
%

%% file: if-else-desc.tex
The public \Code{if else} rules 
are nearly identical to the \vanillaC\ rules, (e.g., \vanillaC\ rule If Else True shown in Figure~\ref{Fig: if else \vanillaC true}), 
with the added assertion that the guard of the conditional is public (i.e., does not contain private data): $(\Expr) \isPub \gamma$.
The private \TT{if else} rules, shown in Figures~\ref{Fig: iep vt} and~\ref{Fig: iep lt}, are more interesting. Our strategy for dealing with private-conditioned branches involves executing both branches as a sequence of statements (with some additional helper algorithms to aid in storing changes, restoration between branches, and resolution of true values). 
We chose to use big-step semantics to facilitate the comparison of the \piccoC\ semantics with the \vanillaC\ semantics, and for its proof of correctness that we will discuss in the next Section. 
We give also an example of \piccoC\ code in Figure~\ref{Fig: if else \piccoC code} and~\ref{Fig: if else \piccoC code dp}, and the corresponding execution in Figure~\ref{Fig: if else \piccoC expanded} and~\ref{Fig: if else piccoC expanded dp}. We use coloring throughout Figure~\ref{Fig: if else color} and~\ref{Fig: if else color dp} to highlight the corresponding sections of code and rule execution. 

The starting and ending states of the \piccoC\ Private If Else rules are essentially the same as the starting and ending states of the corresponding \vanillaC\ If Else rule; however, there are several additional assertions that guarantee that both of the private-conditioned branches are executed. 
The assertions of these semantic rules are listed sequentially, from top to bottom.
We have two different styles of tracking modifications within conditional code blocks that are used within these rules: variable tracking and location tracking. 
Variable tracking is used when there are only single-level changes within the private-conditioned branches, whereas location tracking is used when we have multi-level changes (i.e., a branch contains a pointer dereference write) or potential out-of-bounds changes (i.e., array write at a public index). 

The main idea of both styles is to first store the original value of each variable that is modified within either branch; execute the \TT{then} branch; save the resulting values from the \TT{then} branch and restore all modified variables to their original values; execute the \TT{else} branch; and finally, to securely resolve which values should be kept -- those from the \TT{then} branch or those from the \TT{else} branch. 
In the variable style of tracking, we utilize temporary variables to keep track of all modifications made during either branch -- initializing the \TT{else} temporary with the original value, storing the result of the \TT{then} branch in the \TT{then} temporary and using the \TT{else} temporary to restore the original value, and finally using the result of the private-conditional and what is stored in each variable at the end of the \TT{else} branch as well as it's corresponding \TT{then} temporary to securely resolve what values to continue evaluating the program with. 

This style of tracking is robust enough for many uses, however, there are two notable exceptions where we run into issues, both involving the potential of the location we track not being the location that is actually modified. 
The first exception involves pointer dereference writes -- these alone are not an issue, but when location the pointer refers to is modified and we also perform pointer dereference writes, it becomes clear that variable tracking cannot easily find and handle these cases. 
The second exception involves array writes at public indices -- these become problematic due to the potential for writing out-of-bounds. As most array indices are not hard-coded, it isn't obvious that the write will be within bounds until execution, and to ensure we catch all of these cases we must use a more robust style of tracking to catch out-of-bounds writes. We stress here that array writes at private indices do not fall within this exception, as this operation will securely update the array within its bounds (as updating beyond the bounds of the array would leak that this private value is larger than the size of the array), and as such we can simply track the entire array properly using variable tracking.  
It is possible to ensure that we find all of the locations that are modified in both of these cases by dynamically adding these types of modifications as they are evaluated, which is the goal of the location tracking. 
In the location style of tracking, we still follow a similar evaluation pattern as with variable tracking, storing the original values for locations we know will be modified first, then restoring between branches, and resolving at the end. As we evaluate each branch and come upon one of these special cases, we will check to see if we have already marked that location for tracking, and if not we add that location and its original value before the modification occurs. 
It is worthwhile to stress again the role of the accumulator here with respect to other statements. We increment it when we evaluate the \TT{then} and \TT{else} statements, so that if we attempt to evaluate a (sub)statement with public side effects or restricted operations, we have an (oblivious) runtime failure. It also facilitates scoping of temporary variables within nested private-conditioned \TT{if else} statements.
We proceed to further describe the different assertions and specifics of both styles next.


\input{fig-ifelse-color-VT}

\paragraph{Conditional Code Block Variable Tracking}
For this style of tracking, we first evaluate expression $\Expr$ over environment $\gamma$, memory $\sigma$ and accumulator $\Acc$ 
to obtain some number $n$; the same environment, 
and a potentially updated memory (e.g. in the case $\Expr = x++$). 
We then extract the non-local variables that are modified within either branch, and check whether multi-level modifications or array writes at a public index occur. This is achieved with Algorithm $\DynExtract$ by iterating through both statement $\stmt_1$ and $\stmt_2$ and storing the variable names in list $\vl_{\Acc+1}$, as well as updating and returning a tag to indicate whether we have found multi-level modifications (0 for false, 1 for true).
%
Next we call Algorithm $\Initialize$, which 
stores $n$ as the value of a temporary variable $\res_{\Acc+1}$, using $\Acc +1$ to denote the current level of nesting in the upcoming \Code{then} and \Code{else} statements. The variable $\res_{\Acc+1}$  is later used in the resolution phase, to select the result according to the branching condition. 
It then iterates through the list of variables, creating two temporary versions of each variable, named $\x\_{then\_\Acc}$ and $\x\_{else\_\Acc}$, and storing each in memory with the initial value of what $\x$ has in the memory $\sigma_1$. 
%
Next is the evaluation of the \TT{then} statement, and afterwards 
%
%
we must restore the original memory. 
To do this, we call $\Restore$, which iterates through each of the variables $\x$ contained within $\vl_{\Acc+1}$, storing their current value into their \Code{then} temporary (i.e., $\x\_\mathit{then}_{\Acc+1} = \x$) and restoring their original value from their \Code{else} temporary (i.e., $\x = \x\_\mathit{else}_{\Acc+1}$). 
Once we have completed this, 
%
 the evaluation of the \TT{else} statement can occur.

Finally, we need to perform the resolution of all changes made to variables in either branch. 
To enable this, we call Algorithm $\ResolveR$ to iterate through each of the variables within $\vl_{\Acc+1}$ and grab their values accordingly, as well as retrieving the result of the private condition (whose value we stored in $\res_{\Acc+1}$). 
We then use multiparty protocol $\MPC{resolve}$ to facilitate the resolution of the true values, as these computations require communication between parties. 
For variables that are not array or pointer variables (e.g., those in~\ref{Fig: if else \piccoC code}), we perform a series of binary operations over the byte values of the private variables as shown in~\ref{Fig: if else \piccoC expanded} (e.g., \Code{c=(res$\cdot$c\_t)+((1-res)$\cdot$c\_e)}). 
The process is similar for arrays, with some addition bookkeeping due to their structure as a const pointer referring to the location with the array data. 
For pointers, we must handle the different locations referred to from each branch, merging the two location lists and finding what the true location is. 
The resolved values are then returned, and Algorithm $\ResolveS$ stores all each back into memory for its respective variable. 
Notice that, in the conclusion, we revert to the original environment $\gamma$. In this way, all the temporary variables we used become out of scope and are discarded - in particular, this prevents reusing the same temporary variable mapping if we have multiple (not nested) private if else statements.

\input{fig-ifelse-color-LT}

\paragraph{Conditional Code Block Location Tracking}
\label{sec: priv if lt desc}
Here we track modifications during private-conditioned branches at the level of memory blocks and offsets, which ensures that we do not update any data in memory inaccurately, as is shown in Figure~\ref{Fig: simple pointer challenge ex} using variable tracking SMC techniques.  
To facilitate this, we use the mapping structure $\DMap$ to track changes to each location at each level of nesting. This structure maps locations to a four-tuple of the original value, the \TT{then} branch value, a tag to notate whether the \TT{then} branch value was updated during the restoration phase, and the type of value stored (i.e., $(\loc, \offset) \to$ $(\val_1, \val_2$, $\tagb$, $\Type$)). The tag is used to allow us to add to $\DMap$ as we encounter pointer dereference writes and array writes at public indices without needing to track which branch we are in. It is always initialized as 0, and updated to 1 when we enter the restoration phase and store a value into the \TT{then} position. This way, if a location was added in the \TT{else} branch (i.e., was not modified in the \TT{then} branch), we know to use the original value as the \TT{then} value when we resolve the true value of that location at the end. 

The overall structure of the location tracking rule is similar to the variable tracking rule. 
We first evaluate $\Expr$ to $\n$, then call $\DynExtract$ to find variables that are modified during the execution of either branch and that there are multi-level modifications within at least one branch. 
We then call $\DynInit$, which stores the result of the private conditional and uses the variables we found to create the initial mapping $\DMap$.  
Next, we proceed to evaluate the \TT{then} branch, and call $\DynRestore$ to update $\DMap$ with the ending \TT{then} values for all locations that are tracked and restore the original values back into memory. 
After, we evaluate the \TT{else} branch and, once complete, call $\DynResolveR$ to retrieve the result of the conditional and the \TT{then} and \TT{else} values for each location. 
As with variable tracking, we use multiparty protocol $\MPC{resolve}$ to obtain the true values, and then store them back into their respective locations using Algorithm $\DynResolveS$. 
It is important to note that when we evaluate a pointer dereference write or array write at a public index inside a branch, we check to see if the given location is in $\DMap[\Acc]$.  If it is not, we add a mapping to store the original data (i.e., $(\loc, \offset) \to$ (\TT{orig}, $\Null$, $0$, $\Type$)). Notice that the data can either be a regular value (i.e., for a memory block storing a private int) or a pointer data structure representing a private pointer (i.e., for a memory block storing a private int*). 

In Figure~\ref{Fig: if else piccoC expanded dp}, we show an approximation of the execution of the pointer challenge example shown in Figure~\ref{Fig: simple pointer challenge ex}. 
When we reach the private-conditioned branching statement, we first store the result of the condition \TT{a < b}. As we execute the \TT{then} branch, we add the entry for $\loc_\TT{a}$ to $\DMap$, as \TT{p} refers to \TT{a}. We restore between branches by resetting $\loc_\TT{a}$ to its original value stored in $\DMap[\loc_\TT{a}][0]$. As we execute the \TT{else} branch, we add the entry for $\loc_\TT{p}$ to $\DMap$, as we are modifying which location \TT{p} points to. Finally, we resolve the true values for each modified location in $\DMap$. This approach eliminates the issues shown in Figure~\ref{Fig: simple pointer challenge ex}, as we do not rely on the pointer's current location to appropriately resolve the true values.

%% file: fig-ifelse-color-VT.tex
\begin{figure*} \footnotesize
\begin{tabular}{l}
\hspace{0.3cm}
\begin{tabular}{l l}
\begin{subfigure}{.36\textwidth}
\begin{lstlisting}
private int a=3,b=7,c=0;		
if ($\Code{\ExprC{a<b}}$) $\Code{\sC{c=a;}}$
else $\Code{\ssC{c=b;}}$
\end{lstlisting}
	\caption{\piccoC\ code.}
	\label{Fig: if else \piccoC code}	
\end{subfigure}		
&
\begin{subfigure}{.57\textwidth}
\begin{lstlisting}
private int a=3,b=7,c=0,$\Code{\initC{res=}\ExprC{a<b},\initC{c\_t=c,c\_e=c};}$
$\Code{\sC{c=a;}}$ $\Code{\restC{c\_t=c; c=c\_e;}}$
$\Code{\ssC{c=b;}}$ $\resoC{\Code{c=(res}\cdot\Code{c\_t)+((1-res)}\cdot\Code{c);}}$ 	
\end{lstlisting}
	\caption{Variable-tracking execution.}	
	\label{Fig: if else \piccoC expanded}
\end{subfigure} 
\end{tabular}
\\ \\
\begin{subfigure}{\textwidth}
	\inferrule{\begin{array}{l} 
		\qquad
			\ExprC{((\pidA, \gamma^\pidA, \sigma^\pidA, \DMap^\pidA, \Acc, \Expr) \ \ \Mid ... \Mid
					 (\pidZ, \gamma^\pidZ, \sigma^\pidZ, \DMap^\pidZ, \Acc, \Expr))} 
			\crcr\ExprC{\Deval{\locLL_1}{\codeLL_1} 
			((\pidA, \gamma^\pidA, \sigma^\pidA_1, \DMap^\pidA_1, \Acc, n^\pidA) \Mid ... \Mid
			 (\pidZ, \gamma^\pidZ, \sigma^\pidZ_1, \DMap^\pidZ_1, \Acc, n^\pidZ))}
		\qq
			\{(\ExprC\Expr) \isPriv \gamma^\pid\}^\pidZ_{\pid = \pidA}
		\crcr 
			\extC{\{\DynExtract(}\sC{\stmt_1}\extC{,}\ \ssC{\stmt_2} \extC{, \gamma^\pid) = (\x_\vl, 0)\}^\pidZ_{\pid = \pidA}}
		\crcr
			\initC{\{\Initialize(}\extC{\x_\vl}\initC{,\ \gamma^\pid, \sigma^\pid_1, \ExprC{\n^\pid},\ \AccPP) = (\gamma^\pid_1, \sigma^\pid_2, \locL^\pid_2)\}^\pidZ_{\pid = \pidA}}
		\crcr\qquad
			\sC{((\pidA, \gamma^\pidA_1, \sigma^\pidA_2, \DMap^\pidA_1, \AccPP, \stmt_1) \ \ \ \Mid ... \Mid
			 	  (\pidZ, \gamma^\pidZ_1, \sigma^\pidZ_2, \DMap^\pidZ_1, \AccPP, \stmt_1))}
			\crcr\sC{\Deval{\locLL_3}{\codeLL_2} 
				((\pidA, \gamma^\pidA_2, \sigma^\pidA_3, \DMap^\pidA_2, \AccPP, \Skip) \Mid ... \Mid
				 (\pidZ, \gamma^\pidZ_2, \sigma^\pidZ_3, \DMap^\pidZ_2, \AccPP, \Skip))}
		\crcr
			\restC{\{\Restore(}\extC{\x_\vl}\restC{,\ \gamma^\pid_1,\ \sigma^\pid_3,\ \AccPP) = (\sigma^\pid_4, \locL^\pid_4)\}^\pidZ_{\pid = \pidA}}
		\crcr\qquad
			\ssC{((\pidA, \gamma^\pidA_1, \sigma^\pidA_4, \DMap^\pidA_2, \AccPP, \stmt_2) \ \ \ \Mid ... \Mid
				   (\pidZ, \gamma^\pidZ_1, \sigma^\pidZ_4, \DMap^\pidZ_2, \AccPP, \stmt_2))}
				\crcr\ssC{\Deval{\locLL_5}{\codeLL_3} 
				((\pidA, \gamma^\pidA_3, \sigma^\pidA_5, \DMap^\pidA_3, \AccPP, \Skip) \Mid  ... \Mid
				 (\pidZ, \gamma^\pidZ_3, \sigma^\pidZ_5, \DMap^\pidZ_3, \AccPP, \Skip))}
		\crcr
			\resoC{\{\Resolve\_\mathrm{Retrieve}(\extC{\x_\vl},\ \AccPP, \gamma^\pid_1, \sigma^\pid_5) 
				= ([(\val^\pid_{t1}, \val^\pid_{e1}), ..., (\val^\pid_{tm}, \val^\pid_{em})], 
					\ExprC{\n^\pid}, \locL^\pid_6)\}^\pidZ_{\pid = \pidA}}
		\crcr
			\resoC{\MPC{resolve}([\ExprC{\n^\pidA}, ..., \ExprC{\n^\pidZ}], 
				[[(\val^\pidA_{t1}, \val^\pidA_{e1}), ..., (\val^\pidA_{tm}, \val^\pidA_{em})]], ..., 
				 [(\val^\pidZ_{t1}, \val^\pidZ_{e1}), ..., (\val^\pidZ_{tm}, \val^\pidZ_{em})]])} 
				\crcr\qquad\resoC{= [[\val^\pidA_1, ..., \val^\pidA_m], ... [\val^\pidZ_1, ..., \val^\pidZ_m]]}
		\crcr
			\resoC{\{\Resolve\_\mathrm{Store}(\extC{\x_\vl},\ \AccPP, \gamma^\pid_1, \sigma^\pid_5, 
				[\val^\pid_1, ..., \val^\pid_m]) = (\sigma^\pid_6, \locL^\pid_7)\}^\pidZ_{\pid = \pidA}}
		\crcr
			\locLL_6 = \locLL_1 \addL (\pidA, \locL^\pidA_2) \Mid ... \Mid (\pidZ, \locL^\pidZ_2) \addL \locLL_3 
						\addL (\pidA, \locL^\pidA_4) \Mid ... \Mid (\pidZ, \locL^\pidZ_4) \addL \locLL_5
						\addL (\pidA, \locL^\pidA_6) \Mid ... \Mid (\pidZ, \locL^\pidZ_6) 
		\crcr 
			\locLL_7 = \locLL_6 \addL (\pidA, \locL^\pidA_7) \Mid ... \Mid (\pidZ, \locL^\pidZ_7)
		\qq 
			\codeLL_4 = \codeLL_1 \addC \codeLL_2 \addC \codeLL_3
	\end{array} }
	{\begin{array}{l} 
	((\pidA, \gamma^\pidA, \sigma^\pidA, \DMap^\pidA, \Acc, \If\ (\ExprC\Expr)\ \sC{\stmt_1}\ \Else\ \ssC{\stmt_2}) 
	 \Mid ... \Mid 
	 (\pidZ, \gamma^\pidZ, \sigma^\pidZ, \DMap^\pidZ, \Acc, \If\ (\ExprC\Expr)\ \sC{\stmt_1}\ \Else\ \ssC{\stmt_2})) 
		\Deval{\locLL_7}{\codeLL_4 \addC \codeSP{iep}} 
		\crcr ((\pidA, \gamma^\pidA, \resoC{\sigma^\pidA_{6}}, \resoC{\DMap^\pidA_{3}}, \Acc, \Skip) 
		\qq \-\ \-\ \Mid ... \Mid 
		 (\pidZ, \gamma^\pidZ, \resoC{\sigma^\pidZ_{6}}, \resoC{\DMap^\pidZ_{3}}, \Acc, \Skip) )
		\end{array}}
	\caption{\piccoC\ rule Private If Else - Variable Tracking.}
	\label{Fig: iep vt}
\end{subfigure}
\\ \\
\begin{subfigure}{\textwidth}
	\inferrule{\begin{array}{l}
		\ExprC{((\pid, \hgamma, \hsigma,\ \bsq, \bsq, \hExpr)\ \ \Mid \hConfig)\ \ }
		\ExprC{\Veval_{\codeVLL_1}\ }
			\ExprC{((\pid, \hgamma,\ \hsigma_1, \bsq, \bsq, \hat{n})\ \ \ \Mid \hConfig_1)}
		\qq \ExprC{\hat{n}} \neq \ExprC{0}
		\crcr\sC{((\pid, \hgamma, \hsigma_1, \bsq, \bsq, \hstmt_1) \Mid \hConfig_1)\ }  
			\sC{\Veval_{\codeVLL_2}((\pid, \hgamma_1, \hsigma_2, \bsq, \bsq, \Skip)\Mid \hConfig_2)}
	\end{array}}
	{\begin{array}{l}
	((\pid, \hgamma, \hsigma, \bsq, \bsq, \If (\ExprC{\hExpr})\ \sC{\hstmt_1}\ \Else\ \ssC{\hstmt_2}) \Mid \hConfig)
		\Veval_{\codeVLL_1\addC\codeVLL_2\addC[\codeVS{iet}]} 
		((\pid, \hgamma, \hsigma_2, \bsq, \bsq, \Skip) \Mid \hConfig_2)\end{array}}
	\caption{\vanillaC\ rule If Else True.}	
	\label{Fig: if else \vanillaC true}
\end{subfigure}	
\end{tabular}
\caption{\Code{if else} branching on private data example (\ref{Fig: if else \piccoC code}, \ref{Fig: if else \piccoC expanded}) matching to the \piccoC\ variable-tracking (\ref{Fig: iep vt}) and \vanillaC\ (\ref{Fig: if else \vanillaC true}) semantic rules. Coloring in the rules highlight the corresponding code and rule execution.}
\label{Fig: if else color}
\end{figure*}

%% file: fig-ifelse-color-LT.tex
\begin{figure*} \footnotesize
\begin{tabular}{l}
\begin{tabular}{l l}
\hspace{0.3cm}
\begin{subfigure}{.27\textwidth}
\begin{lstlisting}
private int a=3,
		b=7,c=5,*p=&a;
if ($\Code{\ExprC{a>b}}$)
	$\Code{\sC{*p=c;}}$ 
else 
	$\Code{\ssC{p=\&b;}}$
\end{lstlisting}
	\caption{\piccoC\ code.}
	\label{Fig: if else \piccoC code dp}	
\end{subfigure}		
&
\begin{subfigure}{0.69\textwidth}
\begin{lstlisting}[emph={[2]res, resolve}, emphstyle={[2]\color{blue}}]
private int a=3,b=7,c=5,*p=&a,$\Code{\initC{res=}\ExprC{a>b}}$;
$\initC{\DMap[\Acc][\loc_\TT{p}]=([1,[(\loc_\TT{a}, 0)],[1],1],[],0, \Code{private int*});}$  
$\sC{\DMap[\Acc][\loc_\TT{a}]=(3,0,0,\Code{private int});}$   $\sC{\Code{*p=c;}}$ 
$\restC{\DMap[\Acc][\loc_\TT{p}]=([1,[(\loc_\TT{a}, 0)],[1],1],[1,[(\loc_\TT{a}, 0)],[1],1],1, \Code{private int*});}$
$\restC{\DMap[\Acc][\loc_\TT{a}]=(3,5,1, \Code{private int});}$ $\restC{\loc_\TT{p}=\DMap[\Acc][\loc_\TT{p}][0];}$  $\restC{\loc_\TT{a}=\DMap[\Acc][\loc_\TT{a}][0];}$
$\Code{\ssC{p=\&b;}}$
$\resoC{\loc_\TT{p}=\Code{resolve(res,}\DMap[\Acc][\loc_\TT{p}],\loc_\TT{p});}$   $\resoC{\loc_\TT{a}=\Code{resolve(res,}\DMap[\Acc][\loc_\TT{a}],\loc_\TT{a});}$
\end{lstlisting}
	\caption{Location-tracking execution.}	
	\label{Fig: if else piccoC expanded dp}
\end{subfigure} 
\end{tabular}
\\ \\
\begin{subfigure}{\textwidth}
\inferrule{\begin{array}{l} 
		\qquad
			\ExprC{((\pidA, \gamma^\pidA, \sigma^\pidA, \DMap^\pidA, \Acc, \Expr) \ \ \Mid ... \Mid
				 	 (\pidZ, \gamma^\pidZ, \sigma^\pidZ, \DMap^\pidZ, \Acc, \Expr))} 
			\crcr\ExprC{\Deval{\locLL_1}{\codeLL_1} 
			((\pidA, \gamma^\pidA, \sigma^\pidA_1, \DMap^\pidA_1, \Acc, n^\pidA) \Mid ... \Mid
			  (\pidZ, \gamma^\pidZ, \sigma^\pidZ_1, \DMap^\pidZ_1, \Acc, n^\pidZ))}
		\qq 
			\{(\ExprC\Expr) \isPriv \gamma^\pid\}^\pidZ_{\pid = \pidA}
		\crcr 
			\extC{\{\DynExtract(}\sC{\stmt_1}\extC{,}\ \ssC{\stmt_2} \extC{, \gamma^\pid) = (\x_\vl, 1)\}^\pidZ_{\pid = \pidA}}
		\crcr
			\initC{\{\DynInit(\DMap^\pid_1,}\ \extC{\x_\vl}\initC{,\ \gamma^\pid,\ \sigma^\pid_1,\ \ExprC{\n^\pid}, \AccPP) = (\gamma^\pid_1, \sigma^\pid_2, \DMap^\pid_2, \locL^\pid_2)\}^\pidZ_{\pid = \pidA}}
		\crcr\qquad 
			\sC{((\pidA, \gamma^\pidA_1, \sigma^\pidA_2, \DMap^\pidA_2, \AccPP, \stmt_1) \ \ \ \Mid  ... \Mid
				  (\pidZ, \gamma^\pidZ_1, \sigma^\pidZ_2, \DMap^\pidZ_2, \AccPP, \stmt_1))}
				\crcr\sC{\Deval{\locLL_3}{\codeLL_2} 
				((\pidA, \gamma^\pidA_2, \sigma^\pidA_3, \DMap^\pidA_3, \AccPP, \Skip) \Mid  ... \Mid
				 (\pidZ, \gamma^\pidZ_2, \sigma^\pidZ_3, \DMap^\pidZ_3, \AccPP, \Skip))}
		\crcr
			\restC{\{\DynRestore(\sigma^\pid_3, \DMap^\pid_3, \AccPP) = (\sigma^\pid_4, \DMap^\pid_4, \locL^\pid_4)\}^\pidZ_{\pid = \pidA}}
		\crcr\qquad
			\ssC{((\pidA, \gamma^\pidA_1, \sigma^\pidA_4, \DMap^\pidA_4, \AccPP, \stmt_2) \ \ \ \Mid  ... \Mid
				    (\pidZ, \gamma^\pidZ_1, \sigma^\pidZ_4, \DMap^\pidZ_4, \AccPP, \stmt_2))}
				\crcr\ssC{\Deval{\locLL_5}{\codeLL_3} 
				((\pidA, \gamma^\pidA_3, \sigma^\pidA_5, \DMap^\pidA_5, \AccPP, \Skip) \Mid  ... \Mid
				  (\pidZ, \gamma^\pidZ_3, \sigma^\pidZ_5, \DMap^\pidZ_5, \AccPP, \Skip))}
		\crcr
			\resoC{\{\DynResolve\_\mathrm{Retrieve}(\gamma^\pid_1, \sigma^\pid_5, \DMap^\pid_5, \AccPP) 
				= ([(\val^\pid_{t1}, \val^\pid_{e1}), ..., (\val^\pid_{tm}, \val^\pid_{em})], 
					\ExprC{\n^\pid}, \locL^\pid_6)\}^\pidZ_{\pid = \pidA}}
		\crcr
			\resoC{\MPC{resolve}([\ExprC{\n^\pidA}, ..., \ExprC{\n^\pidZ}], 
				[[(\val^\pidA_{t1}, \val^\pidA_{e1}), ..., (\val^\pidA_{tm}, \val^\pidA_{em})]], ..., 
				 [(\val^\pidZ_{t1}, \val^\pidZ_{e1}), ..., (\val^\pidZ_{tm}, \val^\pidZ_{em})]])} 
				\crcr\qquad\resoC{= [[\val^\pidA_1, ..., \val^\pidA_m], ... [\val^\pidZ_1, ..., \val^\pidZ_m]]}
		\crcr
			\resoC{\{\DynResolve\_\mathrm{Store}(\DMap^\pid_5, \sigma^\pid_5, \AccPP, [\val^\pid_1, ..., \val^\pid_m]) 
				= (\sigma^\pid_6, \DMap^\pid_6, \locL^\pid_7)\}^\pidZ_{\pid = \pidA}}
		\crcr
			\locLL_6 = \locLL_1 \addL (\pidA, \locL^\pidA_2) \Mid ... \Mid (\pidZ, \locL^\pidZ_2) \addL \locLL_3 
						\addL (\pidA, \locL^\pidA_4) \Mid ... \Mid (\pidZ, \locL^\pidZ_4) \addL \locLL_5
						\addL (\pidA, \locL^\pidA_6) \Mid ... \Mid (\pidZ, \locL^\pidZ_6) 
		\crcr 
			\locLL_7 = \locLL_6 \addL (\pidA, \locL^\pidA_7) \Mid ... \Mid (\pidZ, \locL^\pidZ_7)
		\qq 
			\codeLL_4 = \codeLL_1 \addC \codeLL_2 \addC \codeLL_3
	\end{array} }
	{\begin{array}{l} 
	((\pidA, \gamma^\pidA, \sigma^\pidA, \DMap^\pidA, \Acc, \If\ (\ExprC\Expr)\ \sC{\stmt_1}\ \Else\ \ssC{\stmt_2}) 
	 \Mid ... \Mid 
	 (\pidZ, \gamma^\pidZ, \sigma^\pidZ, \DMap^\pidZ, \Acc, \If\ (\ExprC\Expr)\ \sC{\stmt_1}\ \Else\ \ssC{\stmt_2})) 
		 \Deval{\locLL_7}{\codeLL_4 \addC \codeMP{iepd}} 
		 \crcr ((\pidA, \gamma^\pidA, \resoC{\sigma^\pidA_{6}}, \resoC{\DMap^\pidA_{6}}, \Acc, \Skip) 
		 	 \qq \-\ \-\ \Mid ... \Mid
		  	(\pidZ, \gamma^\pidZ, \resoC{\sigma^\pidZ_{6}}, \resoC{\DMap^\pidZ_{6}}, \Acc, \Skip))
		\end{array}}
	\caption{\piccoC\ rule Private If Else - Location Tracking.}
	\label{Fig: iep lt}
\end{subfigure}
\\ \\
\begin{subfigure}{\textwidth}
Multiparty If Else False  		\\
\inferrule{\begin{array}{l}
	\qquad 
		\ExprC{((\pidA, \hgamma,\ \hsigma,\ \bsq, \bsq, \hExpr)\ \ \ \Mid ... \Mid
		  (\pidZ, \hgamma,\ \hsigma,\ \bsq, \bsq, \hExpr))} 
	 \crcr \ExprC{\Veval_{\codeVLL_1} 
			((\pidA, \hgamma,\ \hsigma_1, \bsq, \bsq, \hat{n})\ \ \ \Mid ... \Mid
			  (\pidZ, \hgamma,\ \hsigma_1, \bsq, \bsq, \hat{n}))}
	\qq \ExprC\hn = 0
	\crcr\qquad \sC{((\pidA, \hgamma,\ \hsigma_1, \bsq, \bsq, \hstmt_1)\ \ \Mid ... \Mid
			   (\pidZ, \hgamma,\ \hsigma_1, \bsq, \bsq, \hstmt_1))} 
	\crcr \sC{\Veval_{\codeVLL_2} 
			((\pidA, \hgamma_1, \hsigma_2, \bsq, \bsq, \Skip) \Mid ... \Mid
			 (\pidZ, \hgamma_1, \hsigma_2, \bsq, \bsq, \Skip))}
	\crcr\qquad \ssC{((\pidA, \hgamma,\ \hsigma_1, \bsq, \bsq, \hstmt_2)\ \ \Mid ... \Mid
			   (\pidZ, \hgamma,\ \hsigma_1, \bsq, \bsq, \hstmt_2))} 
	\crcr \ssC{\Veval_{\codeVLL_3} 
			((\pidA, \hgamma_2, \hsigma_3, \bsq, \bsq, \Skip) \Mid ... \Mid
			 (\pidZ, \hgamma_2, \hsigma_3, \bsq, \bsq, \Skip))}
	\end{array}}
	{\begin{array}{l}
	((\pidA, \hgamma, \hsigma,\ \bsq, \bsq, \If (\hExpr)\ \hstmt_1\ \Else\ \hstmt_2)\Mid ... \Mid
	  (\pidZ, \hgamma, \hsigma,\ \bsq, \bsq, \If (\hExpr)\ \hstmt_1\ \Else\ \hstmt_2)) 
		\Veval_{\codeVLL_1\addC\codeVLL_2\addC\codeVLL_3\addC[\codeVS{mpief}]} 
		\crcr
		((\pidA, \hgamma, \ssC{\hsigma_3}, \bsq, \bsq, \Skip)
		\qq\ \Mid ... \Mid
		 (\pidZ, \hgamma, \ssC{\hsigma_3}, \bsq, \bsq, \Skip))
		 \end{array}}
	\caption{\vanillaC\ rule Multiparty If Else False.}
	\label{Fig: van mpief}
\end{subfigure}
\end{tabular}
\caption{\Code{if else} branching on private data example (\ref{Fig: if else \piccoC code dp}, \ref{Fig: if else piccoC expanded dp}), \piccoC\ location-tracking (\ref{Fig: iep lt}), and \vanillaC\ Multiparty If Else False (\ref{Fig: van mpief}) rules. Coloring in the rules highlight the corresponding code and rule execution.}
\label{Fig: if else color dp}
\end{figure*}

%% file: selected-proofs.tex
In this section we present the main methatheoretic results, 
with proof sketches and important definitions given in Appendix~\ref{app: metatheory}.
We will begin by discussing how we leverage multiparty protocols, then proceed to discuss 
the most challenging result, which is correctness. 
Once correctness is proven, noninterference follows from a standard argument, with some adaptations needed to deal with the fact that private data is encrypted and 
that we want to show indistinguishability of evaluation traces.

\subsection{Multiparty Protocols}
In our semantics, we leverage multiparty protocols to compartmentalize
the complexity of handling private data. In the formal treatment this
corresponds to using Axioms in our proofs to reason about
protocols. These Axioms allow us to guarantee the desired properties
of correctness and noninterference for the overall model, to provide
easy integration with new, more efficient protocols as they become
available, and to avoid re-proving the formal guarantees for the
entire model when new protocols are added.  Proving that these Axioms
hold is a responsibility of the library implementor in order to have
the system fully encompassed by our formal model.  Secure
multiparty computation protocols that already come with guarantees
of correctness and security are the only ones worth considering, so the implementor would only need to
ensure that these guarantees match our definitions of correctness and
noninterference.

For example, if private values are represented using Shamir secret sharing~\cite{Shamir79}, Algorithm~\ref{algo: mpc mult}, $\MPC{mult}$, represents
a simple multiparty protocol for multiplying
private values from~\cite{Gennaro98}.
In Algorithm~\ref{algo: mpc mult}, lines 2 and
3 define the protocol, while lines 1, 4, and 5 relate the
protocol to our semantic representation.

\begin{algorithm*}\footnotesize
\caption{$\n^\pid_3 \gets \MPC{mult}(\n^\pid_1, \n^\pid_2)$}
\label{algo: mpc mult}
\begin{algorithmic}[1]
	\STATE Let $f_a(\pid) = \n^\pid_1$ and $f_b(\pid) = \n^\pid_1$.
	\STATE Party $\pid$ computes the value $f_a(\pid) \cdot f_b(\pid)$ and creates its shares by choosing a random polynomial $h_\pid(x)$ of degree $t$, such that $h_\pid(0)=f_a(\pid) \cdot f_b(\pid)$. Party $\pid$ sends to each party $i$ the value $h_\pid(i)$. 
	\STATE After receiving shares from all other parties, party $\pid$ computes their share of $a \cdot b$ as the linear combination $H(\pid) = \sum^{\pidZ}_{i=1} \lambda_i h_i(\pid)$.
	\STATE Let $n^\pid_3 = H(\pid)$
	\RETURN $n^\pid_3$
\end{algorithmic}
\end{algorithm*}

When computation is performed by $q$ parties, at most $t$ of whom may collude ($t < q/2$), Shamir secret sharing encodes a private integer $a$ by choosing a polynomial $f(x)$ of degree $t$ with random coefficients such that $f(0) = a$ (all computation takes place over a finite field). Each participant obtains evaluation of $f$ on a unique non-zero point as their representation of private $a$; for example, party $\pid$ obtains $f(\pid)$. This representation has the property that combining $t$ or fewer shares reveals no information about $a$ as all values of $a$ are equally likely; however, possession of $t+1$ or more shares permits recovering of $f(x)$ via polynomial interpolation and thus learning $f(0) = a$. 

Multiplication in Algorithm~\ref{algo: mpc mult} corresponds to each party locally multiplying shares of inputs $a$ and $b$, which computes the product, but raises the polynomial degree to $2t$. The parties consequently re-share their private intermediate results to lower the polynomial degree to $t$ and re-randomize the shares. Values $\lambda_\pid$ refer to interpolation coefficients which are derived from the computation setup and party $\pid$ index.

In order to preserve the correctness and noninterference guarantees of our
model when such an algorithm is added, a library developer will need
to guarantee that the implementation of this algorithm is correct, meaning that it has the expected input output behavior, and it guarantees noninterference on what is observable.

\subsection{Correctness} \label{sec: erasure} 
We first show the correctness of the \piccoC\ semantics with respect
to the \vanillaC\ semantics. As usual we will do this by establishing
a simulation relation between a \piccoC\ program and a corresponding
\vanillaC\ program. To do so we face two main challenges.

First, we need to guarantee that
the private operations in a \piccoC\ program are reflected in the
corresponding \vanillaC\ program and that the evaluation steps between the two programs correspond. 
To address the former issue, we define an \emph{erasure function} $\bm{\erasure}$ which
translates a \piccoC\ program into a \vanillaC\ program by erasing all
labels and replacing all functions specific to \piccoC\ with their public equivalents. This function also translates memory.
As an example, let us consider
pmalloc; in this case, we have
$\bm{\erasure}({\PMalloc(\Expr,\ \Type)} 
= {(\Malloc(\bm{\erasure}(\Expr) \cdot \sizeof(\bm{\erasure}(\Type))))})$.
That is, pmalloc is rewritten to use malloc, and since the given private type is now public we can use the sizeof function to find the size we will need to allocate. 
To address the latter issue, we have defined our operational semantics in terms of big-step evaluation judgments which allow the evaluation trees of the two programs to have a corresponding structure. In particular, notice how we designed
the Private If Else rule to perform multiple operations in one step, guaranteeing that we have similar ``synchronization points'' in the two evaluation trees. 

Second, we need to guarantee that at each evaluation step the memory
used by a \piccoC\ program corresponds to the one used by the
\vanillaC\ program. 
Given that we simulate multiparty execution over $\pidZ$ parties in \piccoC, we will also use $\pidZ$ parties in \vanillaC. This allows us to easily reason about both local and global semantic rules, as each \piccoC\ party has a corresponding \vanillaC\ party at an identical position in the evaluation trace.
Unfortunately, just applying the function $\bm{\erasure}$ to the \piccoC\ memories in the evaluation trace is not enough. 
In our setting, with
explicit memory management, manipulations of pointers, and array overshooting, guaranteeing a correspondence between the memories becomes particularly
challenging. To better understand the issue here, let us consider
the rule Private Free, discussed in
Section~\ref{subsec: picco mem alloc/dealloc}. 
Remember that our semantic model associates a pointer
with a list of locations, and the Private Free rule frees the
first location in the list, and relocates the content of that location
if it is not the true location. 
Essentially, this rule may swap the content of two locations if the first location
in the list is not the location intended to be freed and 
make the \piccoC\ memory and the \vanillaC\ memory look quite
different. To address this challenge in the proof of correctness, we use a \emph{map}, denoted $\psi$,
to track the swaps that happen when the rule
Private Free is used. The simulation uses and modifies this map to
guarantee that the two memories correspond.
Another related challenge comes from array overshooting. If, by
overshooting an array, a program goes over or into memory blocks of
different types, we may end up in a situation where the locations in
the \piccoC\ memory are significantly different from the ones in the
\vanillaC\ memory. This is mostly due to the size of private types
being larger than their public counterpart. One option to address this
problem would be to keep a more complex map between the two
memories. However, this can result in a much more complex proof, for
capturing a behavior that is faulty, in principle. Instead, we prefer
to focus on situations where overshooting arrays are \emph{well-aligned}, in the sense that they access only memory locations and blocks of the right type and size. 
An illustration of this is given in the Appendix, Figure~\ref{fig: overshooting alignment}. 

Before stating our correctness, we need to introduce some notation.  We
use party-wise lists of codes $\codeLL = (\pidA, [\code_1,\ldots,\code_n])\Mid ... \Mid(\pidZ, [\code_1,\ldots,\code_n]),$ $\codeVLL = (\pidA, [\codeV_1,\ldots,\codeV_m])\Mid$ $...$ $\Mid(\pidZ, [\codeV_1,\ldots,\codeV_m])$ in evaluations (i.e.,
$\eval_{\codeLL})$ to describe the rules of the semantics that
are applied in order to derive the result.  We write
$\codeLL\cong \codeVLL$ to state that the \piccoC\
codes are in correspondence with the \vanillaC\
codes, $\codeLL^\pid$ to denote the list of codes for a specific party $\pid$, and $\codeLL_1::\codeLL_2$ to denote concatenation of the party-wise evaluation code lists. 
We write $\{...\}^\pidZ_{\pid =1}$ to show that an assertion holds for all parties. 
Almost every \piccoC\ rule is in one-to-one
correspondence with a single \vanillaC\ rule within an execution trace (exceptions being private-conditioned branches, \TT{pmalloc}, and multiparty comparison operations).

 We write ${\stmt}\cong \hstmt$ to state that the \vanillaC\ configuration statement $\hstmt$ can be obtained by applying the erasure function to the \piccoC\ statement ${\stmt}$. Similarly, we can extend this notation to configuration by also using the map $\psi$. That is, we write
$(\pid,$ $\gamma{},$ $\sigma{},$ $\DMap$, $\Acc,$ ${\stmt})$ $\cong_\psi$ $(\pid,$ $\hgamma{},$ $\hsigma{},$ $\bsq,$ $\bsq,$ $\hstmt)$ to state that the \vanillaC\ configuration  $(\pid,$ $\hgamma{},$ $\hsigma{},$ $\bsq,$ $\bsq,$ $\hstmt)$ can be obtained by applying the erasure function to the \piccoC\ configuration  $(\pid,$ $\gamma{},$ $\sigma{},$ $\DMap$, $\Acc,$ ${\stmt})$, and memory $\hsigma$ can be obtained from $\sigma{}$ by using the map $\psi$.

We state correctness in terms of  evaluation trees, since we will use evaluation trees to prove a strong form of noninterference  in the next subsection. We use capital Greek letters $\Pi, \Sigma$ to denote evaluation trees.  In the \piccoC\ semantics, we write $\Pi \deriv ((\pidA, \gamma^{{\pidA}}_{},$ $\sigma^{{\pidA}}_{},$ $\DMap^{{\pidA}}_{}$, $\Acc^{{\pidA}}_{},$ $\stmt^{{\pidA}})\ \Mid ...\Mid$ 
	$(\pidZ, \gamma^{{\pidZ}}_{},$ $\sigma^{{\pidZ}}_{},$ $\DMap^{{\pidZ}}_{}$, $\Acc^{{\pidZ}}_{},$ $\stmt^{{\pidZ}}))$ 
	$\Deval{\locLL}{\codeLL}$ 
	$((\pidA, {\gamma^{{\pidA}}_{1}},$ ${\sigma^{{\pidA}}_{1}},$ $\DMap^{{\pidA}}_{1}$, $\Acc^{{\pidA}}_{1},$ ${\val^{{\pidA}}_{}}) \Mid ... \Mid$ 
	$(\pidZ, {\gamma^{{\pidZ}}_{1}},$ ${\sigma^{{\pidZ}}_{1}},$ $\DMap^{{\pidZ}}_{1}$, $\Acc^{{\pidZ}}_{1},$ $\val^{{\pidZ}}))$, to stress that the evaluation tree $\Pi$ proves as conclusion that, for each party $\pid$, configuration  $(\pid,$ $\gamma^\pid,$ $\sigma^\pid,$ $\DMap^\pid$, $\Acc^\pid,$ $\stmt^\pid)$ evaluates to configuration
$(\pid,$ $\gamma^\pid_1,$ $\sigma^\pid_1,$ $\DMap^\pid_1$, $\Acc^\pid_1,$ $\val^\pid)$ by means of the codes in $\codeLL^\pid$. Similarly, for the \vanillaC\ semantics. We then write $\Pi\cong_\psi \Sigma$ for the extension to evaluation trees of the congruence relation with map $\psi$.

In order to properly reason about global multiparty rules, we must assert that all parties are executing from the same original program with corresponding start states and input. 
To do this, we first show that the non-determinism of the semantics will always bring all parties to the same outcome: given $\pidZ$ parties with corresponding start states, if we reach intermediate states that are not corresponding for one or more parties, then there exists a set of steps that will bring all parties to corresponding states again.

\begin{theorem}[Confluence]
\label{Thm: confluence}
Given ${\Config^\pidA} \Mid ... \Mid {\Config^\pidZ}$ such that $\{{\Config^\pidA} \sim {\Config^\pid}\}^{\pidZ}_{\pid = \pidA}$ 
\\
if $({\Config^\pidA} \Mid ... \Mid {\Config^\pidZ})$ $\Deval{{\locLL_1}}{{\codeLL_1}}$ $({\Config^\pidA_1} \Mid ... \Mid {\Config^\pidZ_1})$ such that $\exists \pid\in\{\pidA...\pidZ\} {\Config^\pidA_1} \not\sim {\Config^\pid_1}$, 
\\
then $\exists$ $({\Config^\pidA_1} \Mid ... \Mid {\Config^\pidZ_1})$ $\Deval{{\locLL_2}}{{\codeLL_2}}$ $({\Config^\pidA_2} \Mid ... \Mid {\Config^\pidZ_2})$
\\ 
such that $\{{\Config^\pidA_2} \sim {\Config^\pid_2}\}^{\pidZ}_{\pid = \pidA}$, 
$\{({\locLL^\pidA_1}\addL{\locLL^\pidA_2}) = ({\locLL^\pid_1}\addL{\locLL^\pid_2})\}^{\pidZ}_{\pid = \pidA}$, 
and $\{({\codeLL^\pidA_1}\addC{\codeLL^\pidA_2}) = ({\codeLL^\pid_1}\addC{\codeLL^\pid_2})\}^{\pidZ}_{\pid = \pidA}$.
\end{theorem}

We can now state our correctness result showing that if an \piccoC\ program $\stmt$ can be evaluated  to a value $\val$, and the evaluation is well-aligned (it is an evaluation where all the overshooting of arrays are well-aligned), then the \vanillaC\  program $\hat{\stmt}$ obtained by applying the erasure function to $\stmt$, i.e., $\stmt\cong\hat{\stmt}$, can be evaluated to $\hat{\val}$ where $\val\cong\hat{\val}$. This property can be formalized in terms of congruence: 
\begin{theorem}[Correctness]
\label{Thm: erasure}
For every configuration $\{(\pid,\ \gamma^{{\pid}}_{},$ $\sigma^{{\pid}}_{},$ $\DMap^{{\pid}}_{}$, $\Acc^{{\pid}}_{},$ $\stmt^{\pid})\}^{\pidZ}_{\pid = \pidA}$, 
\\ $\{(\pid,$ $\hgamma^\pid,$ $\hsigma^\pid,$ $\bsq,$ $\bsq,$ $\hstmt^\pid)\}^{\pidZ}_{\pid = \pidA}$ and \LocMap\ $\psi$ 
\\ such that $\{(\pid, \gamma^{{\pid}}_{},$ $\sigma^{{\pid}}_{},$ $\DMap^{{\pid}}_{}$, $\Acc^{{\pid}}_{},$ $\stmt^{{\pid}})$ $\Pcong$ 
$(\pid,$ $\hgamma^\pid,$ $\hsigma^\pid,$ $\bsq,$ $\bsq,$ $\hstmt^\pid)\}^{\pidZ}_{\pid = \pidA}$, 
\\ 
if $\Pi \deriv ((\pidA, \gamma^{{\pidA}}_{},$ $\sigma^{{\pidA}}_{},$ $\DMap^{{\pidA}}_{}$, $\Acc^{{\pidA}}_{},$ $\stmt^{{\pidA}})\ \Mid ...\Mid$ 
	$(\pidZ, \gamma^{{\pidZ}}_{},$ $\sigma^{{\pidZ}}_{},$ $\DMap^{{\pidZ}}_{}$, $\Acc^{{\pidZ}}_{},$ $\stmt^{{\pidZ}}))$ 
	\\ \-\ \-\ \-\ $\Deval{\locLL}{\codeLL}$ 
	$((\pidA, {\gamma^{{\pidA}}_{1}},$ ${\sigma^{{\pidA}}_{1}},$ $\DMap^{{\pidA}}_{1}$, $\Acc^{{\pidA}}_{1},$ ${\val^{{\pidA}}_{}}) \Mid ... \Mid$ 
	$(\pidZ, {\gamma^{{\pidZ}}_{1}},$ ${\sigma^{{\pidZ}}_{1}},$ $\DMap^{{\pidZ}}_{1}$, $\Acc^{{\pidZ}}_{1},$ $\val^{{\pidZ}}))$ 
\\ for codes $\codeLL \in {\piccoCodes}$,
then there exists a derivation 
\\ 
$\Sigma \deriv ((\pidA,$ $\hgamma^\pidA,$ $\hsigma^\pidA,$ $\bsq,$ $\bsq,$ $\hstmt^\pidA)\Mid ...\Mid $
	$(\pidZ,$ $\hgamma^\pidZ,$ $\hsigma^\pidZ,$ $\bsq,$ $\bsq,$ $\hstmt^\pidZ))$ 
	\\ $\Deval{}{\codeVLL}$ 
	$((\pidA,$ $\hgamma^\pidA_1,$ $\hsigma^\pidA_1,$ $\bsq,$ $\bsq,$ $\hval^\pidA)\Mid...\Mid$
	$(\pidZ,$ $\hgamma^\pidZ_1,$ $\hsigma^\pidZ_1,$ $\bsq,$ $\bsq,$ $\hval^\pidZ))$ 
\\ for codes $\codeVLL \in \vanillaCodes$ 
and 
a \LocMap\ $\psi_1$ 
such that 
\\ 
$\codeLL \cong \codeVLL$, 
$\{(\pid, {\gamma^{{\pid}}_{1}},$ ${\sigma^{{\pid}}_{1}},$ $\DMap^{{\pid}}_{1}$, $\Acc^{{\pid}}_{1},$ $\val^{{\pid}}_{})$ $\cong_{\psi_1}$ 
$(\pid,$ $\hgamma^\pid_1,$ $\hsigma^\pid_1,$ $\bsq,$ $\bsq,$ $\hval^\pid)\}^{\pidZ}_{\pid = \pidA}$, 
and $\Pi \cong_{\psi_1} \Sigma$.
\end{theorem}
The proof 
proceeds by induction on the evaluation tree $\Pi$ and the challenges are the ones 
discussed above related to memory management, and proving that the control flow of the erased program (and the corresponding memory) is correct with respect to the one of the \piccoC\ program.
\subsection{Noninterference} \label{sec: noninterference}
\piccoC\ satisfies a strong form of noninterference guaranteeing that two execution traces are indistinguishable up to differences in private values. This stronger version entails data-obliviousness. Instead of using execution traces, we will work directly with evaluation trees in the \piccoC\ semantics -- equivalence of evaluation trees up to private values implies equivalence of execution traces based on the \piccoC\ semantics. This guarantee is provided at the semantics level, we do not consider here compiler optimizations.   

For noninterference, it is convenient to introduce a notion of equivalence requiring that the two memories agree on publicly observable values. Because we assume that private data in memories are encrypted, and so their encrypted value is publicly observable, it is sufficient to consider syntactic equality of memories. Notice that if $\sigma_1=\sigma_2$ we can still have $\sigma_1\ell \neq \sigma_2\ell$, i.e., two executions starting from the same configuration can actually differ with respect to private data.

We can now state our main noninterference result. 
\begin{theorem}[Noninterference over evaluation trees]
\label{Thm: strong noninterference}
For every environment $\{\gamma^{\pid}_{},$ $\gamma^{\pid}_{1},$ $\gamma'^{\pid}_{1}\}^{\pidZ}_{\pid = \pidA}$; 
memory $\{\sigma^{\pid}_{}$, $\sigma^{\pid}_{1}$, $\sigma'^{\pid}_{1} \}^{\pidZ}_{\pid = \pidA}\in\Mem$; 
\changeMap $\{\DMap^{\pid}_{}$, $\DMap^{\pid}_{1}$, $\DMap'^{\pid}_{1}\}^{\pidZ}_{\pid = \pidA}$;
accumulator $\{\Acc^{\pid}_{}$, $\Acc^{\pid}_{1}$, $\Acc'^{\pid}_{1}\}^{\pidZ}_{\pid = \pidA}\in\N$; 
statement $\stmt$, values $\{\val^{\pid}_{}$, $\val'^{\pid}_{}\}^{\pidZ}_{\pid = \pidA}$; 
step evaluation code lists $\codeLL,\codeLL'$ and their corresponding lists of locations accessed $\locLL,\locLL'$; 
\\
if 
$\Pi \deriv\ ((\pidA, \gamma^{\pidA}_{},$ $\sigma^{\pidA}_{},$ $\DMap^{\pidA}_{}$, $\Acc^{\pidA}_{},$ $\stmt)\ \ \Mid ...\Mid (\pidZ, \gamma^{\pidZ}_{},$ $\sigma^{\pidZ}_{},$ $\DMap^{\pidZ}_{}$, $\Acc^{\pidZ}_{},$ $\stmt))$ 
\\ \-\ \quad $\Deval{\locLL}{\codeLL}$ $((\pidA, \gamma^{\pidA}_{1},$ $\sigma^{\pidA}_{1},$ $\DMap^{\pidA}_{1}$, $\Acc^{\pidA}_{1},$ $\val^{\pidA}_{})\Mid ...\Mid (\pidZ, \gamma^{\pidZ}_{1},$ $\sigma^{\pidZ}_{1},$ $\DMap^{\pidZ}_{1}$, $\Acc^{\pidZ}_{1},$ $\val^{\pidZ}_{}))$ 
\\ and   
$\Sigma \deriv ((\pidA, \gamma^{\pidA}_{},$ $\sigma^{\pidA}_{},\ $ $\DMap^{\pidA}_{},\ $ $\Acc^{\pidA}_{},$ $\stmt)\ \ \Mid ...\Mid (\pidZ, \gamma^{\pidZ}_{},\ $ $\sigma^{\pidZ}_{},\ $ $\DMap^{\pidZ}_{}$, $\Acc^{\pidZ}_{},$ $\stmt))$ 
\\ $\-\ \quad \Deval{\locLL'}{\codeLL'}$ $((\pidA, \gamma'^{\pidA}_{1},$ $\sigma'^{\pidA}_{1},$ $\DMap'^{\pidA}_{1}$, $\Acc'^{\pidA}_{1},$ $\val'^{\pidA}_{})\Mid ...\Mid (\pidZ, \gamma'^{\pidZ}_{1},$ $\sigma'^{\pidZ}_{1},$ $\DMap'^{\pidZ}_{1}$, $\Acc'^{\pidZ}_{1},$ $\val'^{\pidZ}_{}))$
\\ then $\{\gamma^{\pid}_{1}=\gamma'^{\pid}_{1}\}^{\pidZ}_{\pid = \pidA}$, 
$\{\sigma^{\pid}_{1}=\sigma'^{\pid}_{1}\}^{\pidZ}_{\pid = \pidA}$, 
$\{\DMap^{\pid}_{1} =\DMap'^{\pid}_{1}\}^{\pidZ}_{\pid = \pidA}$, 
$\{\Acc^{\pid}_{1}=\Acc'^{\pid}_{1}\}^{\pidZ}_{\pid = \pidA}$, 
$\{\val^{\pid}_{}=\val'^{\pid}_{}\}^{\pidZ}_{\pid = \pidA}$, 
$\codeLL=\codeLL'$, 
$\locLL = \locLL'$, 
$\Pi \loweq \Sigma$.
\end{theorem} 
Notice that low-equivalence of evaluation trees already implies the equivalence of the resulting configurations. We repeated them to make the meaning of the theorem clearer. 
This also proves data-obliviousness over memory accesses: in any two executions of the same program with the same public data, which locations in memory are accessed will be based on public data, and therefore identical between the two executions. 
As we proceed to prove Theorem~\ref{Thm: strong noninterference}, we leverage our Axioms reasoning about noninterference of the multiparty protocols.

%% file: implementation.tex
We implement our semantics in the PICCO compiler.
Our implementation is available at \url{https://github.com/applied-crypto-lab/formal-picco} and we plan to merge it into the main PICCO branch as a submitted patch. 
The main modifications to PICCO revolve around a {\em conditional code block tracking} scheme and resolution mechanism for private variables, support for pointer operations and computation within private-conditioned branches, and the addition of tracking structures for pointers based on our semantics in Section~\ref{Sec: Semantics}.
PICCO uses single-statement resolution to manage modifications to private variables made within a private-conditioned branch, which can prove to be more costly when a single variable is modified multiple times within a private-conditioned branch, as shown in the example in Figure~\ref{Fig: resolution cost ex}. 
Here, we show how PICCO conditionally updates the true value of a variable after each statement, resulting in a more costly operation for each statement, as obtaining the true value requires communication between the various computational parties (i.e. this is a distributed secure computation). 
In \piccoC, we provide a conditional code block tracking scheme, where we store the original value for each modified variable before the execution of the \TT{then} branch, execute the \TT{then} branch as normal, store the updated values, restore the original values, execute the \TT{else} branch as normal, and perform resolution of all modified variables once at the end of both branches.
The cost of executing each statement remains the same, however, in \piccoC, the communication overhead is greatly reduced due to only needing to communicate between computational parties to resolve a single variable once at the end of both branches. 
In Figures~\ref{Fig: resolution cost ex PICCO} and~\ref{Fig: resolution cost ex \piccoC}, this amounts to six fewer times the program needs to communicate between computational parties to resolve.

\begin{figure*}\footnotesize
\hspace{0.3cm} \begin{tabular}{l l l}
\begin{subfigure}{0.21\textwidth}
\begin{lstlisting}
private int c, 
		a=1,b=2;
if(a < b) {		
	c = a; 
	a = a + b;
	c = c * b; 
	a = c + a; 
} else {						
	c = b; 
	a = a - b;
	c = c * a; 
	a = c - a; 
}
\end{lstlisting}
\caption{Example code.}
\label{Fig: resolution cost ex code}
\end{subfigure} &
\-\ 
\begin{subfigure}{0.36\textwidth}
\begin{lstlisting}[emph={[2]res,c_e,c_t},emphstyle={[2]\color{blue}}]
private int c,a=1,b=2;
private int res=a<b;
{
c=(res$\cdot$a)$+$((1-res)$\cdot$c); 
a=(res$\cdot$(a+b))$+$((1-res)$\cdot$a);
c=(res$\cdot$(c*b))$+$((1-res)$\cdot$c); 
a=(res$\cdot$(c+a))$+$((1-res)$\cdot$a);	
}	
{c=((1-res)$\cdot$b)$+$(res$\cdot$c);
a=((1-res)$\cdot$(a-b))$+$(res$\cdot$a);
c=((1-res)$\cdot$(c*a))$+$(res$\cdot$c);
a=((1-res)$\cdot$(c-a))$+$(res$\cdot$a);
}
\end{lstlisting}
\caption{How PICCO executes part a.}
\label{Fig: resolution cost ex PICCO}
\end{subfigure} &
\-\ 
\begin{subfigure}{0.37\textwidth}
\begin{lstlisting}[emph={[2]res,c_e,c_t,a_e,a_t},emphstyle={[2]\color{blue}}]
private int c,a=1,b=2,c_t,
	res=a<b,c_e=c,a_t,a_e=a; 
{c = a; 
 a = a + b;
 c = c * b; 
 a = c + a;}
c_t=c;c=c_e;a_t=a;a=a_e;
{c = b; 
 a = a - b;
 c = c * a; 
 a = c - a;}
c=(res$\cdot$c_t)$+$((1-res)$\cdot$c); 
a=(res$\cdot$a_t)$+$((1-res)$\cdot$a); 
\end{lstlisting}
\caption{How \piccoC\ executes part a.}
\label{Fig: resolution cost ex \piccoC}
\end{subfigure}
\end{tabular}
\caption{Example illustrating why single-statement resolution is more costly when modifying variables multiple times in both branches.}
\label{Fig: resolution cost ex}
\end{figure*}

\subsection{Conditional Code Block Tracking}

To implement data-oblivious execution of branches on private data, SMC implementations typically execute both branches, but privately apply the effects of only the relevant branch. Figure~\ref{Fig: simple correct ex} shows the standard handling of private-conditioned branches. 
The top left shows the original C-code block, with annotations for private data;  the top right shows how compilers would flatten the branch to hide the private data used in the condition. 
Here, the only variable modified within either branch is \texttt{c}, so we insert temporary variables to assist in tracking the results of both branches.
Notation $l_\texttt{x}$ in the table indicates the (private) value stored at the location of variable \texttt{x}. For every variable we list the value in the \texttt{initial} state, before executing the conditional; the one after executing the \texttt{then} branch; the one after the initial state has been restored; the one after the execution of the \texttt{else} branch; and the one after the true value of \texttt{c} has been resolved.

To guarantee data-oblivious executions we also need to guarantee that when we branch on a private condition, the two branches do not have different public side effects. A way to address this is to reject programs that contain public side effects in the body of private-conditioned branching statements. This has an impact on other language constructs such as functions. That is, to be able to call a function from the body of such a branching statement, it must have no public side effects. To address this challenge, in our formalism we evaluate each declared function for public side effects and mark it with a flag that indicates whether it is allowed to be called from within a private-conditioned branch.

PICCO is the only compiler we are aware of that treats pointers to private values, or private pointers.  For private pointers, the location being pointed to might be private as well. That is, if a pointer is assigned a new value inside a private-conditioned branch, we cannot reveal whether the original or new location is to be used. For that reason, a private pointer is associated with multiple locations and a private tag that indicates which location is the true location. In particular, a private pointer is represented as a data structure which tracks the number of locations $\nl$ being pointed to; a list of these $\nl$ (known) locations; a list of $\nl$ (private) tags, one of which is set to 1 and all others set to 0; and the level of pointer indirection. Because locations associated with pointers can now be private, there might be additional limitations on what programs can contain within private-conditioned branches to guarantee data-oblivious execution. 
To understand why multiple locations may potentially be stored for a pointer consider Figure~\ref{Fig: simple correct pointer ex}, whose code will result in the pointer \texttt{p} storing two locations, $l_\texttt{a}$ and $l_\texttt{b}$, with $l_\texttt{a}$ being the true location.

\begin{figure*}  \small
\begin{subfigure}{\textwidth} \small
\begin{minipage}{.35\textwidth}
\begin{lstlisting}[emph={[2]if,else}, emphstyle={[2]\color{blue}}]
private int a=3,b=7,c=0;
if(a<b) {c=a;}
else {c=b;}
\end{lstlisting}
\end{minipage}
\hfill
\begin{minipage}{.6\textwidth}
\begin{lstlisting}[emph={[2]res,c_e,c_t},emphstyle={[2]\color{blue}}]
private int a=3,b=7,c=0,res=a<b,c_e=c,c_t;				
{c=a;} c_t=c; c=c_e; 						
{c=b;} c=(res$\cdot$c_t)$+$((1-res)$\cdot$c);
\end{lstlisting}
\end{minipage}
    \\
    \begin{tabular}[t]{| c | c | c | c | c | c |} \hline
        	\TT{location}		&\TT{init} 	& {\TT{then}} &\TT{restore} 	& {\TT{else}}	&\TT{resolve}
        \\ \hline
        $l_\texttt{c}$ 		& {0} 		&	\teal{3}	&	\teal{0}		&	\teal{7}		&\teal{3}
        \\ \hline 
        $l_\texttt{c\_t}$ 	& {} 		&	{}	&	\teal{3}		&	{3}		&{3}
        \\ \hline
    \end{tabular}
    \hfill
   \begin{tabular}[t]{| c | c |} \hline
   		\TT{location} 	& \TT{value} 	\\ \hline
    		$l_\TT{res}$ 	& 1 			\\ \hline
		$l_\TT{c\_e}$ 	& 0				\\ \hline
    \end{tabular}
    \caption{Regular variable handling within private-conditioned branches.}
    \label{Fig: simple correct ex}
\end{subfigure}
\\ \\ \\
{
\begin{subfigure}{\textwidth}
\begin{minipage}{0.35\textwidth}
\begin{lstlisting}[emph={[2]if,else}, emphstyle={[2]\color{blue}}]
private int a=3,b=7,*p;		
if(a$<$b) {p=&a;}	
else {p=&b;}
\end{lstlisting}	
\end{minipage}
\hfill
\begin{minipage}{0.6\textwidth}
\begin{lstlisting}[emph={[2]res,p_t,p_e}, emphstyle={[2]\color{blue}}]
private int a=3,b=7,*p,res=a$<$b,*p_e=p,*p_t;	
{p=&a;} p_t=p; p=p_e;		
{p=&b;} p=$\purple{\TT{resolve}}$(res,p_t,p);	
\end{lstlisting}	
\end{minipage}\small
    \\
    \begin{tabular}[t]{| c | c | c | c | c | c |} \hline
    	\TT{location}		&\TT{init} 	& {\TT{then}} 	& 	\TT{restore} 	& {\TT{else}} 	&\TT{resolve}
    \\ \hline
    $l_\texttt{p}$ 		& {(\ ),\ (\ )}	& {($l_\texttt{a}$),\ (1)}	& 	\teal{(\ ),\ (\ )}	& \teal{($l_\texttt{b}$),\ (1)}	&\teal{($l_\texttt{a}, l_\texttt{b}$),}
    \teal{(1, 0)}
    \\ \hline 
    $l_\texttt{p\_t}$ 	& {}	& {}	& \teal{($l_\texttt{a}$),\ (1)}	& {($l_\texttt{a}$),\ (1)}	&{($l_\texttt{a}$),\ (1)}
    \\ \hline
    \end{tabular}
    \hfill
   \begin{tabular}[t]{| c | c |} \hline
   		\TT{location} 	& \TT{value} 	\\ \hline
		$l_\TT{res}$ 	& 1				\\ \hline
		$l_\TT{p\_e}$ 	& (\ ),\ (\ )		\\ \hline
    \end{tabular}
    \caption{Pointer handling within private-conditioned branches.}
    \label{Fig: simple correct pointer ex}
\end{subfigure}}
\\ \\ \\
\begin{subfigure}{\textwidth}
\begin{minipage}{.32\textwidth}
\begin{lstlisting}[emph={[2]if,else}, emphstyle={[2]\color{blue}}]
private int a=3,
	b=7,c=5,*p=&a;
if (a$<$b) {
	*p=c; }
else {
	p=&b; }
\end{lstlisting}
\end{minipage}
\hfill
\begin{minipage}{.66\textwidth}
\begin{lstlisting}[emph={[2]res,dp_t,p_e,p_t,dp_e}, emphstyle={[2]\color{blue}}]
private int a=3,b=7,c=5,*p=&a,res=a$<$b,
	dp_e=*p,dp_t,*p_e=p,*p_t; 
{*p=c;} 
dp_t=*p; p_t=p; *p=dp_e; p=p_e; 
{p=&b;} 
*p=(res$\cdot$dp_t)$+$((1-res)$\cdot$*p); p=$\purple{\TT{resolve}}$(res,p_t,p);	
\end{lstlisting}
\end{minipage}\small
    \\
    \begin{tabular}[t]{| c | c | c | c | c | c |} \hline
   \TT{location}	& \TT{initial} 		&  \TT{then} 	& \TT{restore} 		& \TT{else} 	& \TT{resolve} 				\\ \hline
    $l_\TT{a}$ 	& {3}       		& \teal{5}   		& \teal{3}			& {3}			& \red{3}					\\ \hline
    $l_\TT{b}$ 	& {7} 	     		& {7}			& {7}				& {7}			& \red{5}					\\ \hline
    $l_\TT{p}$ 	& {($l_\TT{a}$), (1)}	& {($l_\TT{a}$), (1)}	& \teal{($l_\TT{a}$), (1)}	& {($l_\TT{b}$), (1)}	& \teal{($l_\TT{a}$, $l_\TT{b}$),\ (1, 0)}\\ \hline
    $l_\TT{dp\_t}$& {}			&  				& \teal{5} 			& {5}			& {5}					\\ \hline
    $l_\TT{p\_t}$	& {}			& {} & \teal{($l_\TT{a}$), (1)} 	& {($l_\TT{a}$), (1)} & {($l_\TT{a}$), (1)}		\\ \hline
    \end{tabular}
    \hfill
   \begin{tabular}[t]{| c | c |} \hline
   		\TT{location} 	& \TT{value} 		\\ \hline
		$l_\TT{res}$ 	& 1					\\ \hline
		$l_\TT{dp\_e}$ 	& 3					\\ \hline
		$l_\TT{p\_e}$ 	& ($l_\TT{a}$), (1)	\\ \hline
    \end{tabular}
    \caption{Challenges of pointer manipulations within private-conditioned branches.}
    \label{Fig: simple pointer challenge ex}
\end{subfigure}
\\ \\ \\
\begin{subfigure}{\textwidth}
\begin{minipage}[t]{0.32\textwidth}
\begin{lstlisting}[emph={[2]if,else}, emphstyle={[2]\color{blue}}]
public int i=1,j=2;
private int a[j]={0,0},
	b=7,c=3,d=4;
if (c$<$d) { 
	a[i]=c; }
else {
	a[j]=d; }
\end{lstlisting}
\end{minipage}
\hfill
\begin{minipage}[t]{0.66\textwidth}
\begin{lstlisting}[emph={[2]a_t,a_e,b_t,b_e,res}, emphstyle={[2]\color{blue}}]
public int i=1,j=2;
private int a[j]={0,0},b=7,c=3,d=4;
private int res=c$<$d,a_t,a_e=a;
a[i]=c; 
a_t=a; a=a_e; 
a[j]=d;
a=(res$\cdot$a_t)$+$((1-res)$\cdot$a);
\end{lstlisting}
\end{minipage}
\\
    \begin{tabular}[t]{| c | c | c | c | c | c |} \hline
    \TT{location}	& \TT{initial} &  \TT{then} 	& \TT{restore} 		& \TT{else} 	& \TT{resolve} 		\\ \hline
    $l_\TT{a}$ 	& {0, 0}     	& \teal{0, 3}   	& \teal{0, 0}			& {0, 0}			& \teal{0, 3}		\\ \hline
    $l_\TT{b}$ 	& {7} 	     	& {7}			& {7}				& \red{4}		& \red{4}		\\ \hline
    $l_\TT{a\_t}$& {}			&  				& \teal{0, 3} 		& {0, 3}			& {0, 3}			\\ \hline
    \end{tabular}
    \hfill
   \begin{tabular}[t]{| c | c |} \hline
   		\TT{location} 	& \TT{value} 		\\ \hline
		$l_\TT{res}$ 	& 1					\\ \hline
		$l_\TT{a\_e}$ 	& 0, 0				\\ \hline
    \end{tabular}
\caption{Challenges of writing at a public index in a private array within private-conditioned branches.}
\label{Fig: array challenge}
\end{subfigure}
\caption{Private-conditioned branching examples. Simple examples are shown in~\ref{Fig: simple correct ex} and~\ref{Fig: simple correct pointer ex}, and challenging examples are shown in~\ref{Fig: simple pointer challenge ex} and ~\ref{Fig: array challenge}. We show values in memory that change in the table to the left, and values for temporary variables that do not change in the table to the right. We indicate correct updates in memory in \teal{green}, and problematic values in memory in \red{red}.}
\label{Fig: challenge ex}
\end{figure*}

Unfortunately, directly adopting this approach exhibits problematic behavior when more complex operations are considered.
The first level of locations that the pointer refers to is managed, as shown in Figure~\ref{Fig: simple correct pointer ex}, but dereferencing the pointer and modifying the value stored at the location that is pointed to can result in incorrect program behavior and ultimately information leakage. 
Any SMC system that supports both pointers and branches, using PICCO style pointers and classic branch resolution cannot handle these cases.
This occurs because the approach we discussed above relies on assignment statements and single-level, constant location changes to properly restore and resolve the changes made inside the private-conditioned branch. 
One can try to modify this approach to support tracking changes made using pointers at a higher level of indirection  (i.e., tracking \texttt{*p=c} using temporary variables \TT{dp\_t} and \TT{dp\_e}, as shown in Figure~\ref{Fig: simple pointer challenge ex}). 
However, this modification can lead to the incorrect resolution of data when multiple levels of indirection of a pointer are modified within the private-conditioned branches. 
An example of this is shown in Figure~\ref{Fig: simple pointer challenge ex}, where we modify \texttt{*p} in the \TT{then} branch, and then change the location that is pointed to by \texttt{p} in the \TT{else} branch. 
The \TT{then} branch, restoration, and the \TT{else} branch will execute correctly, however, resolving the variables after the completion of the \TT{else} branch will not. 
Given that we have modified the location pointed to by \TT{p}, when we attempt to resolve the modification we made using \texttt{*p}, we will read from and modify the value in $l_\texttt{b}$ (where \TT{p} currently points) instead of the value in $l_\texttt{a}$ (where \TT{p} pointed and wrote to in the \TT{then} branch).

We encounter a similar issue if we write to a public index within a private array during the execution of a private-conditioned branch, and that index happens to be out of bounds. 
We give an example of this in Figure~\ref{Fig: array challenge}. Here, we assume that \TT{b} is assigned the location directly after the array data's location, thus giving us a \emph{well-aligned} out-of-bounds write to illustrate why simple variable tracking is not enough here. 
In any given implementation, an out-of-bounds access is not guaranteed to be \emph{well-aligned}, and therefore unpredictable behavior can occur. 
In this example, we have a zero-indexed array \TT{a} of two elements. In the \TT{then} branch, we modify index 1 (or the second element), then store this updated array and return the array back to it's initial state. 
In the \TT{else} branch, we modify index 2 (out-of-bounds of the array), which updates the value stored for \TT{b}. Given that we were not tracking \TT{b}, this value does not get resolved, and any further uses of \TT{b} will result in incorrect results. 
By using location tracking, we would catch that the location $\loc_\TT{b}$ was modified, and in turn properly restore it to it's original value (as long as the out-of-bounds access is \emph{well-aligned}).

To be able to handle pointers and arrays correctly, 
we must use a location based tracking instead of a variable based tracking.  
However, as our semantics illustrated, this requires additional tracking structures and dynamically checking to ensure that the locations modified during pointer dereference write statements are tracked.  
We, thus, propose a small optimization to full location tracking, which analyzes each top-level private-conditioned branching statement to see if it contains a pointer dereference writes or array writes at public indices in the \TT{then} or \TT{else} clauses, as well as in any nested branching statements present in those clauses.
If no such writes occur, we are able to use simple variable tracking, as shown in rule Private If Else Variable Tracking in Figure~\ref{Fig: iep vt}.

When any such write statement occurs in either branch, we switch to tracking by location, as shown in rule Private If Else Location Tracking in Figure~\ref{Fig: iep lt}. 
For example, consider a program using a pointer to iterate through and modify elements of an array. Allowing pointer dereference writes enables us to perform a different operation on the array depending on whether a private condition holds. 
For location based conditional code block tracking, we create a map to store the \TT{original} and \TT{then} values for each location that is modified within each branch, as well as a tag to indicate whether this location has a value stored for the \TT{then} branch. 
This tag ensures that even when a location is modified for the first time in the \TT{else} branch, we are still able to properly resolve the value for that location by using the original value stored at that location.  This corresponds to the rules and explanation given in subsection~\ref{sec: priv if lt desc}.

As with the {\em conditional code block} variable  tracking scheme, first we find a list of all modified variables, excluding those only used for pointer dereference writes or array writes at a public index. 
We exclude pointer dereference writes as we will grab the location that is pointed to dynamically to ensure we are tracking the modification at the correct location.  
We use this list to store the original values at the location referred to by each of these variables before the execution of the \TT{then} branch. 
Between branches, our restoration is similar to that formalized for \piccoC, just by location. 
We iterate through our map, storing the current value of each location as the \TT{then} value, and restoring the value at the location as the original value. 
We set the tag associated with each location to be 1, as we have added \TT{then} values for each of these locations. 
When executing the \TT{then} and \TT{else} branches, we check before the execution of a pointer dereference write to see if the location we will modify is already being tracked. 
If it is not, we store the current value as the original value for that location, and then continue to execute as normal; if it is already tracked, we proceed as we do not need to store anything additional. 
We set the tag for each new location to 0, as we do not currently have a \TT{then} value stored for those locations.
After the execution of the \TT{else} branch, we proceed to resolve similarly to \piccoC, just by location. 
For each location in the map, we securely compute whether to keep the \TT{then} value (or \TT{original} value, if the tag is 0) stored in the map or the current value at that location based on the private condition.

%% file: evaluation.tex
To highlight the feasibility of our approach we provide preliminary performances numbers over both microbenchmarks and real-world SMC programs.  All experiments were run in a local and distributed manner.  We leverage local runs, where all participants in the SMC program execute on the same machine, to analyze overheads and benefits of our approach. We also  provide distributed deployment of the same benchmarks to illustrate real-world performance. In the distributed configuration, each participant in the SMC program is executed on a separate machine.
We ran our experiments using single-threaded execution on three 2.10GHz machines running CentOS-7. The machines were connected via 1Gbps Ethernet.

Benchmark pay-gap is the program shown in Figure~\ref{Fig: salary vs gender}, where the average female salary and average male salary is computed.  
Benchmark LR-parser is a program that will parse and execute a mathematical function over private data. This program is given as input the function as a list of tokens and the private data to execute the function over.
Benchmark h\_analysis is a program that takes two sets of input data and calculates the percent difference between each element of the two data sets.
All of the "pb" benchmarks are programs designed to emphasize the difference in executing using single-statement resolution, as PICCO does, and when using the conditional code block tracking, as in \piccoC. These programs are run through a large number of iterations of branches where variables are modified several times in each branch, as in the example in Figure~\ref{Fig: resolution cost ex}.

When executing locally, we ran all three computational parties on a single machine. 
When executing distributed, one computational party is run on each machine. 
All programs were run for 50 times both distributed and locally. 
LR-parser has the shortest execution time of our benchmarks (seconds), and its timing is more influenced by small differences in communication overhead during the computation, resulting in a greater standard deviation. 
Program pb-reuse has the longest runtime, executing in about 5 and 3.5 minutes locally and 15 and 8 minutes distributed for PICCO and \piccoC, respectively. 
In our computations, we first took the average time of each of the three computational parties, then the total average and standard deviation of all runs. We calculated the percent speedup in Figure~\ref{Fig: percent diff} using PICCO's timings as a baseline.   We can see from our micro-benchmarks which stress private-conditioned branches that our approach provides significant runtime improvement.  This is not surprising as our resolution mechanism requires less
communication between parties.  However, this benefit is reduced in real-world SMC programs and is proportional to the number of private-conditioned branches as well as their complexity (i.e. the number of private variables they use and modify). Similarly we observe that performance improvement or reduction is dilated when communication overheads increase (e.g. we execute in a distributed setting).  This is also not surprising as the communication cost as a percentage of total runtime increases in a distributed execution.

\begin{figure*}
\begin{subfigure}{0.48\textwidth}\footnotesize
\centering{\includegraphics[width=\textwidth]{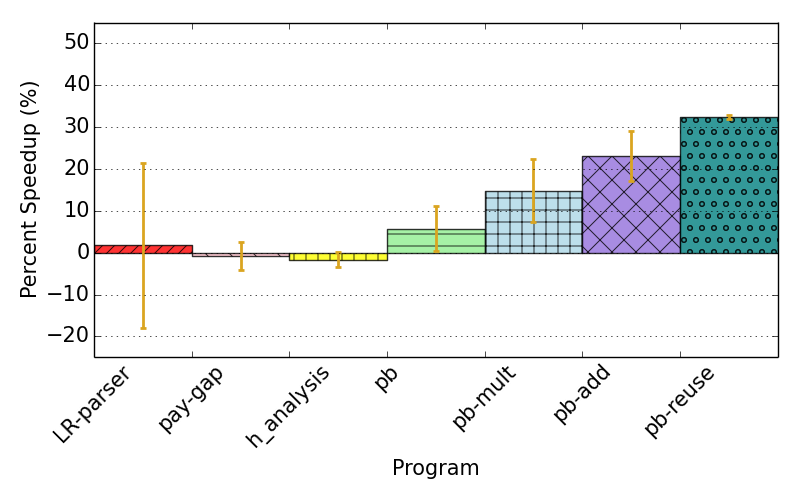}}
\caption{Local}
\label{Fig: local percent diff}
\end{subfigure}
\quad
\begin{subfigure}{0.48\textwidth}\footnotesize
\centering{\includegraphics[width=\textwidth]{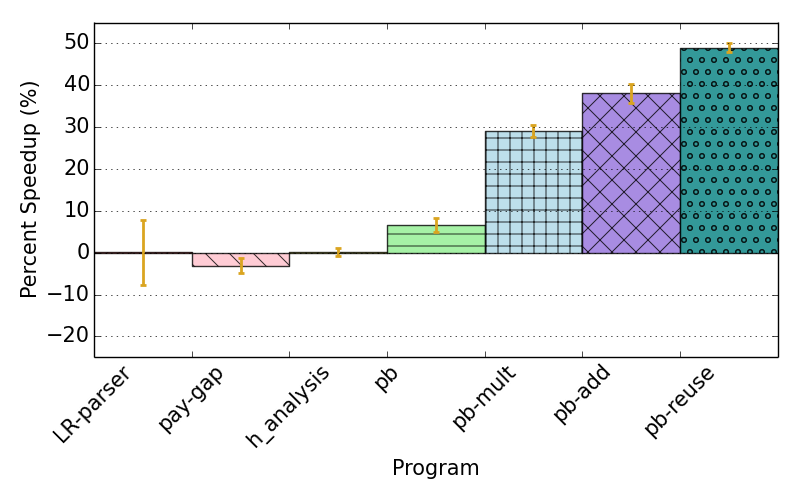}}
\caption{Distributed}
\label{Fig:dist percent diff}
\end{subfigure}
\caption{Percentage Runtime Differences}
\label{Fig: percent diff}
\vspace{-0.3cm}
\end{figure*}

%% file: appendix.tex
\section{Multiparty Protocols}
\label{app: mpc protocols}

Algorithm~\ref{algo: mpc binop}, $\MPC{b}$, is a selection control algorithm that directs the evaluation to the relevant multiparty computation algorithm based on the given binary operation $\binop\in\{\cdot,\div,+,-\}$. 

\begin{algorithm}[H]\footnotesize
\caption{$(\n^\pidA_3, ..., \n^\pidZ_3) \gets \MPC{b}(\binop, [\n^\pidA_1, ..., \n^\pidZ_1], [\n^\pidA_2, ..., \n^\pidZ_2])$}
\label{algo: mpc binop}
\begin{algorithmic}
 	\IF{$(\binop = \cdot)$}
		\FORALL{$\pid\in\{\pidA...\pidZ\}$} 
			\STATE $\n^\pid_3 = \MPC{mult}(\n^\pid_1, \n^\pid_2)$
		\ENDFOR
		\RETURN $(\n^\pidA_3, ..., \n^\pidZ_3)$
	\ELSIF{$(\binop = \div)$}
		\FORALL{$\pid\in\{\pidA...\pidZ\}$} 
			\STATE $\n^\pid_3 = \MPC{div}(\n^\pid_1, \n^\pid_2)$
		\ENDFOR
		\RETURN $(\n^\pidA_3, ..., \n^\pidZ_3)$
	\ELSIF{$(\binop = -)$}
		\FORALL{$\pid\in\{\pidA...\pidZ\}$} 
			\STATE $\n^\pid_3 = \MPC{sub}(\n^\pid_1, \n^\pid_2)$
		\ENDFOR
		\RETURN $(\n^\pidA_3, ..., \n^\pidZ_3)$
	\ELSIF{$(\binop = +)$}
		\FORALL{$\pid\in\{\pidA...\pidZ\}$} 
			\STATE $\n^\pid_3 = \MPC{add}(\n^\pid_1, \n^\pid_2)$
		\ENDFOR
		\RETURN $(\n^\pidA_3, ..., \n^\pidZ_3)$
	\ENDIF
\end{algorithmic}
\end{algorithm}

Each of the given multiparty protocols in Algorithm~\ref{algo: mpc binop} (i.e., $\MPC{mult},$ $\MPC{sub},$ $\MPC{add},$ $\MPC{div}$) must be defined using protocols that have been proven to uphold the desired properties within our proofs (i.e., correctness and noninterference). 
We give an example definition for $\MPC{mult}$ in Algorithm~\ref{algo: mpc mult}, but this definition can be swapped out with any protocol for the secure multiparty computation of multiplication that maintains the properties of correctness and noninterference. 
We defer the definition of all other SMC binary operations, rely on assertions that the protocols chosen to be used with this model will maintain both correctness and noninterference in our proofs. 
We chose this strategy as SMC implementations of such protocols will be proven to hold our desired properties on their own, and this allows us to not only leverage those proofs, but to also improve the versatility of our model by allowing such algorithms to be easily swapped out as newer, improved versions become available.

\section{Semantics}
\label{app: semantics}

\input{app-rules}

We show selected additional semantic rules in Figures~\ref{Fig: sem app} and~\ref{Fig: sem app arr} to give a more encompassing view of our semantic model.
In particular, we show a larger subset of array and pointer rules, as well as input/output and public allocation and deallocation. 
We use several algorithms within the rules in order to increase the readability of the rules and compartmentalize common functionalities between rules; we will discuss a few further here. 

The variations of $\PermL$ are all used to map the information given as arguments into the appropriate byte-wise permission tuples. This allows us to show the important information about the byte-wise permissions without the repetition of showing the list of permission tuples and getting into the more intricate details, particularly for pointer data structures which have a mix of public and private data and therefore a mix of permissions. 
Algorithms such as $\InputVal$ and $\OutputArr$ are, respectively, handlers for reading data in and writing data to a file. 
Algorithm $\SelectFreeable$ is used to evaluate whether the locations a pointer refers to are indeed freeable (i.e., all locations were allocated through the use of \TT{malloc} or \TT{pmalloc} and not the default location or a location allocated during a variable declaration. 
Algorithm $\Free$ is used to modify the permissions of the location that is being freed to be $\PermN$. 
All algorithms such as $\Update$, $\UpdateArr$, and $\UpdatePtr$ contain the intricacies of updating the specific type of data within memory. 

Algorithm $\Retrieve$ will go through memory and pull all values that are referred to at the $\nl$ locations of that pointer. This is a helper function to allow up to show the exact information that will be used within the multiparty protocol. 
Algorithm $\DynUpdate$ ensures we are tracking every location that is modified within a private-conditioned branch by checking if the current location we are modifying with the pointer dereference write is already tracked in $\DMap$, and if it is not, adding it and the original value for that location to $\DMap$. 
Algorithm $\GetLoc$ handles incrementing the given location and offset by the appropriate size in bytes. It returns the new location and offset that the pointer will refer to, without complicating the rule with all the redundancies of handling locations of different sizes and finding the appropriate position. 
Algorithm $\DerefPtrPub$ will take the location and offset that the pointer refers to as well as the expected type, and give back the value that is obtained from reading from that position. This handles all the intricacies of pointers with non-zero offsets, and ensures that we will always read data from memory as the expected type. 

We also handle pointers of higher levels of indirection; we chose not to show those rules here as the main concepts behind the rules are fairly similar to those shown. 
The full semantic model for \piccoC\ is available at \cite{amys-dissertation} (Chapter 5), including the algorithms used within the semantics and the full \vanillaC\ semantics.

\subsection{Scoping}
In our semantics, we implement standard C scoping through our use of the environment $\gamma$. 
To illustrate this, Figure~\ref{Fig: scoping} contains a few rules with $\gamma$ and its additions highlighted in \red{red}. 
First, we show the Private Array Declaration rule, where we add a new mapping for the array variable and its newly allotted location to the environment and return the updated environment. 
Next, we have the standard Statement Sequencing rule, where we pass along all additions to the environment that were made within each statement. 
Finally we show the Public If Else True rule. Here, we have an updated environment $\RT{\gamma_1}$ returned from the evaluation of the \TT{then} branch $\stmt_1$; any new mappings introduced within $\stmt_1$ become out of scope once we exit this rule, so we return the original environment $\RT\gamma$. This way, any further references to local variables declared within $\stmt_1$ will no longer be found in the environment and cannot execute, as is expected.

\begin{algorithm*}\footnotesize
\caption{$(\x_\mathit{mod}, \tagb) \gets \DynExtract(\stmt_1,\ \stmt_2, \gamma)$}
\label{algo: dyn extract}
\begin{algorithmic}
	\STATE $\tagb = 0$
 	\STATE $\x_\mathit{local} = [\ ] $
	\STATE $\x_\mathit{mod} = [\ ] $		
       \FORALL {$\stmt \in \{\stmt_1;\ \stmt_2\}$}												
		 \IF{$((\stmt ==\ \Type\ \x) \lor (\stmt ==\ \Type\ \x[\Expr]))$}
            			\STATE $\x_\mathit{local}.\mathit{append}(\x)$ 								
		\ELSIF{$((\stmt ==\ \x = \Expr) \land (\lnot\x_\mathit{local}.\mathit{contains}(\x)))$}			
            			\STATE $\x_\mathit{mod} = \x_\mathit{mod} \cup [\x]$
				\FORALL{$\Expr_1 \in \Expr$}
					\IF{$((\Expr_1 ==\ \plpl\x_1) \land (\lnot\x_\mathit{local}.\mathit{contains}(\x_1)))$}
						\STATE $\x_\mathit{mod} = \x_\mathit{mod} \cup [\x_1]$
					\ENDIF
				\ENDFOR	
		\ELSIF{$((\stmt ==\ \x[\Expr_1] = \Expr_2) \land (\lnot\x_\mathit{local}.\mathit{contains}(\x)))$}
				\IF{$(\Expr_1)\isPriv \gamma$}
					\STATE $\x_\mathit{mod} = \x_\mathit{mod} \cup [\x]$
				\ELSE
					\STATE $\tagb = 1$
				\ENDIF
				\FORALL{$\Expr \in \{\Expr_1, \Expr_2\}$}
					\IF{$((\Expr ==\ \plpl\x_1) \land (\lnot\x_\mathit{local}.\mathit{contains}(\x_1)))$}
						\STATE $\x_\mathit{mod} = \x_\mathit{mod} \cup [\x_1]$
					\ENDIF
				\ENDFOR
		\ELSIF{$((\stmt ==\ \plpl\x) \land (\lnot\x_\mathit{local}.\mathit{contains}(\x)))$}
				\STATE $\x_\mathit{mod} = \x_\mathit{mod} \cup [\x]$
		\ELSIF{$(\stmt == *\x = \Expr)$}
				\STATE $\tagb = 1$
				\FORALL{$\Expr_1 \in \Expr$}
					\IF{$((\Expr_1 ==\ \plpl\x_1) \land (\lnot\x_\mathit{local}.\mathit{contains}(\x_1)))$}
						\STATE $\x_\mathit{mod} = \x_\mathit{mod} \cup [\x_1]$
					\ENDIF
				\ENDFOR
		\ENDIF
	\ENDFOR
         \RETURN $(\x_\mathit{mod}, \tagb)$
\end{algorithmic}
\end{algorithm*}

\begin{algorithm*}\footnotesize
\caption{$(\gamma_1, \sigma_1, \locL) \gets \Initialize(\x_{\vl},\ \gamma, \sigma, n, \Acc)$}
\label{algo: initialize}
\begin{algorithmic}
	\STATE $\loc_\res = \phi(\mathit{temp})$
	\STATE $\gamma_1 = \gamma[\res\_\Acc \to (\loc_\res, \Priv\ \Int)]$
	\STATE $\byte_\res = \Encode(\Priv\ \Int, n)$
	\STATE $\sigma_1 = \sigma[\loc_\res \to (\byte_\res, \Priv\ \Int, 1, \VarPermL(\PermF, \Priv\ \Int, \Priv, 1))]$
	\STATE $\locL = [(\loc_\res, 0)]$
	\FORALL{$\x \in \x_{\vl}$}
		\STATE $(\loc_\x, \Type) = \gamma(\x)$	
		\STATE $\loc_t = \phi(\mathit{temp})$
		\STATE $\loc_e = \phi(\mathit{temp})$
		\STATE $\locL = \locL \addL[(\loc_\x, 0), (\loc_t, 0), (\loc_e, 0)]$
		\STATE $\gamma_1 = \gamma_1[\x\_t\_\Acc \to (\loc_t, \Type)][\x\_e\_\Acc \to (\loc_e, \Type)]$
		\STATE $(\byte_\x, \Type, \nl, \VarPermL(\PermF, \Type, \Priv, \nl)) = \sigma_1(\loc_\x)$
		\IF{$(\Type = \Priv\ \Const\ \btype*)$} 
			\STATE $\loc_{ta} = \phi(\mathit{temp})$
			\STATE $\loc_{ea} = \phi(\mathit{temp})$
			\STATE $[1, [(\loc_{xa}, 0)], [1], 1] = \DecodePtr(\Type, 1, \byte_{x})$
			\STATE $(\byte_{xa}, \Priv\ \btype, \nl, \PtrPermL(\PermF, \Priv\ \btype, \Priv, \nl)) = \sigma_1(\loc_{xa})$
			\STATE $\sigma_1 = \sigma_1[\loc_{ta} \to (\byte_{xa}, \Priv\ \btype, \nl, \VarPermL(\PermF, \Priv\ \btype, \Priv, \nl))]$
			\STATE $\sigma_1 = \sigma_1[\loc_{ea} \to (\byte_{xa}, \Priv\ \btype, \nl, \VarPermL(\PermF, \Priv\ \btype, \Priv, \nl))]$
			\STATE $\byte_{t} = \EncodePtr(\Type, [1, [(\loc_{t}, 0)], [1], 1])$
			\STATE $\byte_{e} = \EncodePtr(\Type, [1, [(\loc_{e}, 0)], [1], 1])$
			\STATE $\sigma_1 = \sigma_1[\loc_t \to (\byte_{t}, \Type, 1, \PtrPermL(\PermF, \Type, \Priv, 1))]$
			\STATE $\sigma_1 = \sigma_1[\loc_e \to (\byte_{e}, \Type, 1, \PtrPermL(\PermF, \Type, \Priv, 1))]$
			\FORALL{$i \in \{0...\nl-1\}$}
				\STATE $\locL = \locL \addL [(\loc_{xa}, i), (\loc_{ta}, i), (\loc_{ea}, i)]$
			\ENDFOR
		\ELSE
			\STATE $\sigma_1 = \sigma_1[\loc_t \to (\byte_\x, \Type, \nl, \VarPermL(\PermF, \Type, \Priv, \nl))]$
			\STATE $\sigma_1 = \sigma_1[\loc_e \to (\byte_\x, \Type, \nl, \VarPermL(\PermF, \Type, \Priv, \nl))]$
		\ENDIF
	\ENDFOR									
	\RETURN $(\gamma_1, \sigma_1, \locL)$
\end{algorithmic}
\end{algorithm*}

\begin{algorithm*}\footnotesize
\caption{$(\sigma_4, \locL) \gets \Restore(\x_{\vl}, \gamma, \sigma, \Acc)$}
\label{algo: restore}
\begin{algorithmic}
	\STATE $\locL = [\ ]$
	\FORALL{$\x \in \x_{\vl}$}	 
		\STATE $(\loc_\x, \Type) = \gamma(\x)$
		\STATE $(\loc_t, \Type) = \gamma(\x\_t\_\Acc)$
		\STATE $(\loc_e, \Type) = \gamma(\x\_e\_\Acc)$
		\STATE $\locL = \locL \addL [(\loc_\x, 0), (\loc_t, 0), (\loc_e, 0)]$
		\IF{$(\Type = \Priv\ \Const\ \btype*)$}
			\STATE $(\byte_{xa}, \Type, 1, \VarPermL(\PermF, \Type, \Priv, 1)) = \sigma(\loc_\x)$
			\STATE $(\byte_{ta}, \Type, 1, \VarPermL(\PermF, \Type, \Priv, 1)) = \sigma(\loc_t)$
			\STATE $(\byte_{ea}, \Type, 1, \VarPermL(\PermF, \Type, \Priv, 1)) = \sigma(\loc_e)$
			\STATE $[1, [(\loc_{xa}, 0)], [1], 1] = \DecodePtr(\Type, 1, \byte_{xa})$
			\STATE $[1, [(\loc_{ta}, 0)], [1], 1] = \DecodePtr(\Type, 1, \byte_{ta})$
			\STATE $[1, [(\loc_{ea}, 0)], [1], 1] = \DecodePtr(\Type, 1, \byte_{ea})$
			\STATE $\sigma_1[\loc_{xa} \to (\byte_t, \Type, \nl, \VarPermL(\PermF, \Type, \Priv, \nl))] = \sigma$
			\STATE $\sigma_2[\loc_{ta} \to (\byte_\x, \Type, \nl, \VarPermL(\PermF, \Type, \Priv, \nl))] = \sigma_1$
			\STATE $\sigma_3 = \sigma_2[\loc_{ta} \to (\byte_t, \Type, \nl, \VarPermL(\PermF, \Type, \Priv, \nl)]$
			\STATE $(\byte_\x, \Type, \nl, \VarPermL(\PermF, \Type, \Priv, \nl)) = \sigma_3(\loc_{ea})$
			\STATE $\sigma_4 = \sigma_3[\loc_{xa} \to (\byte_\x, \Type, \nl, \VarPermL(\PermF, \Type, \Priv, \nl)]$
			\FORALL{$i \in \{0...\nl-1\}$}
				\STATE $\locL = \locL \addL [(\loc_{xa}, i), (\loc_{ta}, i), (\loc_{ea}, i)]$
			\ENDFOR
		\ELSE
			\STATE $\sigma_1[\loc_\x \to (\byte_t, \Type, \nl, \VarPermL(\PermF, \Type, \Priv, \nl))] = \sigma$
			\STATE $\sigma_2[\loc_t \to (\byte_\x, \Type, \nl, \VarPermL(\PermF, \Type, \Priv, \nl)] = \sigma_1$
			\STATE $\sigma_3 = \sigma_2[\loc_t \to (\byte_t, \Type, \nl, \VarPermL(\PermF, \Type, \Priv, \nl)]$
			\STATE $(\byte_\x, \Type, \nl, \VarPermL(\PermF, \Type, \Priv, \nl)) = \sigma_3(\loc_e)$
			\STATE $\sigma_4 = \sigma_3[\loc_\x \to (\byte_\x, \Type, \nl, \VarPermL(\PermF, \Type, \Priv, \nl)]$
		\ENDIF
		\STATE $\sigma = \sigma_4$
	\ENDFOR							
	\RETURN $(\sigma_4, \locL)$
\end{algorithmic}
\end{algorithm*}

\begin{algorithm*}\footnotesize
\caption{$(\valL, \n_{\res}, \locL) \gets \ResolveR(\x_{\vl}, \Acc, \gamma, \sigma)$}
\label{algo: resolve R}
\begin{algorithmic}
	\STATE $\valL = [\ ]$
	\STATE $(\loc_{\res}, \Priv\ \Int) = \gamma(\res\_\Acc)$
	\STATE $(\byte_{\res}, \Priv\ \Int, 1, \VarPermL(\PermF, \Priv\ \Int, \Priv, 1)) = \sigma(\loc_{\res})$
	\STATE $\n_{\res} = \Decode(\Priv\ \Int, \byte_{\res})$
	\STATE $\locL = [(\loc_{\res}, 0)]$
	\FORALL{ $\x \in \x_{\vl}$ }	
		\STATE $(\loc_\x, \Type) = \gamma(\x)$
		\STATE $(\loc_t, \Type) = \gamma(\x_t)$	
		\STATE $(\byte_{x}, \Type, \nl, \VarPermL(\PermF, \Type, \Priv, \nl)) = \sigma(\loc_\x)$
		\STATE $(\byte_{t}, \Type, \nl, \VarPermL(\PermF, \Type, \Priv, \nl)) = \sigma(\loc_t)$
		\STATE $\locL = \locL \addL [(\loc_\x, 0), (\loc_t, 0)]$
		\IF{ $(\Type = \Priv\ \btype)$ }	
			\STATE $\val_{x} = \Decode(\Priv\ \btype, \byte_{x})$
			\STATE $\val_{t} = \Decode(\Priv\ \btype, \byte_{t})$
			\STATE $\valL = \valL.\mathit{append}((\val_{t}, \val_{x}))$
		\ELSIF{ $(\Type = \Priv\ \Const\ \btype*)$ }	
			\STATE $[1, [(\loc_{xa}, 0)], [1], 1] = \DecodePtr(\Type, 1, \byte_{x})$
			\STATE $[1, [(\loc_{ta}, 0)], [1], 1] = \DecodePtr(\Type, 1, \byte_{t})$
			\STATE $(\byte_{xa}, \Priv\ \btype, \nl, \VarPermL(\PermF, \Priv\ \btype, \Priv, \nl)) = \sigma(\loc_{xa})$
			\STATE $(\byte_{ta}, \Priv\ \btype, \nl, \VarPermL(\PermF, \Priv\ \btype, \Priv, \nl)) = \sigma(\loc_{ta})$
			\FORALL{$i \in \{0 ... \nl-1\}$}
				\STATE $\val_{xi} = \DecodeArr(\Priv\ \btype, i, \byte_{xa})$
				\STATE $\val_{ti} = \DecodeArr(\Priv\ \btype, i, \byte_{ta})$
				\STATE $\valL = \valL.\mathit{append}((\val_{ti}, \val_{xi}))$
				\STATE $\locL = \locL \addL [(\loc_{xa}, i), (\loc_{ta}, i)]$
			\ENDFOR
		\ELSIF{$(\Type = \Priv\ \btype*)$}	
			\STATE $[\nl, \locL_\x, \tagbL_\x, \indir] = \DecodePtr(\Type, \nl, \byte_{x})$
			\STATE $[\nl, \locL_t, \tagbL_t, \indir] = \DecodePtr(\Type, \nl, \byte_{t})$
			\STATE $\valL = \valL.\mathit{append}(([\nl, \locL_t, \tagbL_t, \indir], [\nl, \locL_\x, \tagbL_\x, \indir]))$
		\ENDIF	
	\ENDFOR
	\RETURN $(\valL, \n_\res, \locL)$	
\end{algorithmic}
\end{algorithm*}

\begin{algorithm*}\footnotesize
\caption{$(\sigma_1, \locL) \gets \ResolveS(\x_{\vl}, \Acc, \gamma, \sigma, \valL)$}
\label{algo: resolve S}
\begin{algorithmic}
	\STATE $\locL = [\ ]$
	\STATE $\sigma_1 = \sigma$
	\FORALL{ $i \in \{0...|\valL| -1\}$ }			
		\STATE $\x = \x_{\vl}[i]$
		\STATE $\val_\x = \valL[i]$
		\STATE $(\loc_\x, \Type) = \gamma(\x)$
		\STATE $\locL = \locL.\mathit{append}((\loc_\x, 0))$
		\IF{ $(\Type = \Priv\ \btype)$ }	
			\STATE $\sigma_2 = \Update(\sigma_1,\ \loc_\x,\ \val_\x,\ \Type)$
			\STATE $\sigma_1 = \sigma_2$
		\ELSIF{ $(\Type = \Priv\ \Const\ \btype*)$ }	
			\STATE $[1, [(\loc_{xa}, 0)], [1], 1] = \DecodePtr(\Type, 1, \byte_{x})$
			\FORALL{$\offset \in \{0 ... \nl-1\}$}
				\STATE $\val_\offset = \val_\x[\offset]$
				\STATE $\sigma_{2+\offset} = \UpdateArr(\sigma_{1+\offset},\ (\loc_{xa},\offset),\ \val_\offset,\ \Type)$
				\STATE $\locL = \locL.\mathit{append}((\loc_{xa}, \offset))$
			\ENDFOR
			\STATE $\sigma_1 = \sigma_{2+\offset}$
		\ELSIF{$(\Type = \Priv\ \btype*)$}	
			\STATE $\sigma_2 = \UpdatePtr(\sigma_1,\ (\loc_\x, 0),\ \val_\x,\ \Type)$ 
			\STATE $\sigma_1 = \sigma_2$
		\ENDIF	
	\ENDFOR
	\RETURN $(\sigma_1, \locL)$	
\end{algorithmic}
\end{algorithm*}

\subsection{Array Overshooting}
\label{app: array oob}

\begin{minipage}{\textwidth}
\centering{\includegraphics[width=0.75\textwidth]{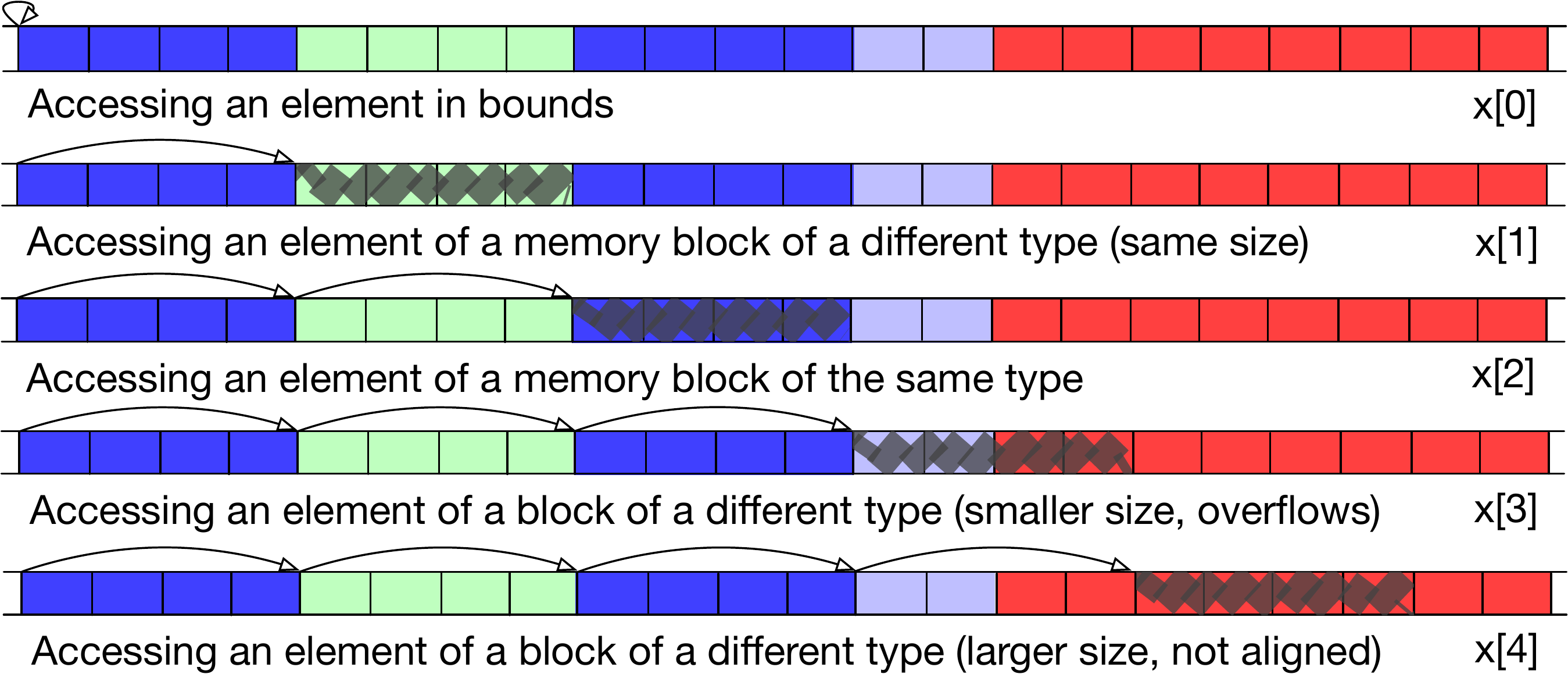}}
\captionof{figure}{Types of overshooting array accesses.}
\label{Fig: arr oob read vanC}
\end{minipage}
\input{overshooting}

\section{Metatheory}
\label{app: metatheory}
In this section, we will provide our main definitions and metatheory, including proof sketches of the most crucial Theorems and Lemmas, to give the reader a more complete understanding and some intuition behind how our proofs work. 

\subsection{Correctness}
In our semantics, we give each evaluation an identifying code as a shorthand way to refer to that specific evaluation, as well as to allow us to quickly reason about the \vanillaC\ and \piccoC\ evaluations that are congruent to each other (i.e., a \vanillaC\ rule and an identical one handling only public data in \piccoC).

The list of \vanillaC\ codes are as follows: 
$\vanillaCodes$ = 
	[$\mathit{mpb}$, $\mathit{mpcmpt}$, $\mathit{mpcmpf}$, $\mathit{mppin}$, 
	$\mathit{mpra}$, $\mathit{mpwe}$, $\mathit{mpfre}$, $\mathit{mpiet}$, $\mathit{mpief}$,
	$\mathit{mprdp}$, $\mathit{mprdp1}$, $\mathit{mpwdp}$, $\mathit{mpwdp1}$, 
	$\mathit{fls}$, $\mathit{ss}$, $\mathit{sb}$, $\mathit{ep}$, $\mathit{cv}$, $\mathit{cl}$, 
	$\mathit{r}$, $\mathit{w}$, $\mathit{ds}$, $\mathit{dv}$, $\mathit{dp}$, $\mathit{da}$, 
	$\mathit{wle}$, $\mathit{wlc}$, $\mathit{bp}$, $\mathit{bs}$, $\mathit{bm}$, $\mathit{bd}$, 
	$\mathit{ltf}$, $\mathit{ltt}$, $\mathit{eqf}$, $\mathit{eqt}$, $\mathit{nef}$, $\mathit{net}$, 
	$\mathit{mal}$, $\mathit{fre}$, $\mathit{wp}$, $\mathit{wdp}$, $\mathit{wdp1}$, $\mathit{rp}$, 
	$\mathit{rdp}$, $\mathit{rdp1}$, $\mathit{ra}$, $\mathit{wa}$, $\mathit{rao}$, $\mathit{wao}$, 
	$\mathit{rae}$, $\mathit{wae}$, $\mathit{loc}$, $\mathit{iet}$, $\mathit{ief}$,  
	$\mathit{inp}$, $\mathit{inp1}$, $\mathit{out}$, $\mathit{out1}$, $\mathit{df}$, $\mathit{ty}$, 
	$\mathit{fd}$, $\mathit{fpd}$, $\mathit{fc}$, $\mathit{pin}$, $\mathit{pin1}$, $\mathit{pin2}$].

The list of \piccoC\ codes are as follows: 
$\piccoCodes$ = 
	[$\mathit{mpb}$, $\mathit{mpcmp}$, 
	$\mathit{mpra}$, $\mathit{mpwa}$, $\mathit{mppin}$, $\mathit{mpdp}$, $\mathit{mpdph}$, 
	$\mathit{mpfre}$, 
	$\mathit{mprdp}$, $\mathit{mprdp1}$, 
	$\mathit{mpwdp}$, $\mathit{mpwdp1}$, $\mathit{mpwdp2}$, $\mathit{mpwdp3}$, 
	$\mathit{iet}$, $\mathit{ief}$, $\mathit{iep}$, $\mathit{iepd}$, $\mathit{wle}$, $\mathit{wlc}$, 
	$\mathit{dp}$, $\mathit{dp1}$, 
	$\mathit{rp}$, $\mathit{rp1}$, $\mathit{rdp}$, $\mathit{rdp1}$, $\mathit{rdp2}$, 
	$\mathit{wp}$, $\mathit{wp1}$, $\mathit{wp2}$, 
	$\mathit{wdp}$, $\mathit{wdp1}$, $\mathit{wdp2}$, $\mathit{wdp3}$, $\mathit{wdp4}$, $\mathit{wdp5}$, 
	$\mathit{da}$, $\mathit{da1}$, $\mathit{das}$, 
	$\mathit{ra}$, $\mathit{ra1}$, $\mathit{rea}$, 
	$\mathit{wa}$, $\mathit{wa1}$, $\mathit{wa2}$, 
	$\mathit{wea}$, $\mathit{wea1}$, $\mathit{wea2}$, 
	$\mathit{rao}$, $\mathit{rao1}$, 
	$\mathit{wao}$, $\mathit{wao1}$, $\mathit{wao2}$, 
	$\mathit{pin}$, $\mathit{pin1}$, $\mathit{pin2}$, $\mathit{pin3}$, $\mathit{pin4}$, 
	$\mathit{pin5}$, $\mathit{pin6}$, $\mathit{pin7}$, 
	$\mathit{mal}$, $\mathit{malp}$, $\mathit{fre}$, $\mathit{pfre}$, 
	$\mathit{cv}$, $\mathit{cv1}$, $\mathit{cl}$, $\mathit{cl1}$, 
	$\mathit{loc}$, $\mathit{ty}$, 
	$\mathit{df}$, $\mathit{fd}$, $\mathit{fpd}$, $\mathit{fc}$, $\mathit{fc1}$, 
	$\mathit{bp}$, $\mathit{bs}$, $\mathit{bm}$, $\mathit{bd}$, 
	$\mathit{ltf}$, $\mathit{ltt}$, $\mathit{eqf}$, $\mathit{eqt}$, $\mathit{nef}$, $\mathit{net}$, 
	$\mathit{dv}$, $\mathit{d1}$, $\mathit{r}$, $\mathit{r1}$, $\mathit{w}$, $\mathit{w1}$, $\mathit{w2}$, 
	$\mathit{ds}$, $\mathit{ss}$, $\mathit{sb}$, $\mathit{ep}$, 
	$\mathit{inp}$, $\mathit{inp1}$, $\mathit{inp2}$, $\mathit{inp3}$, 
	$\mathit{out}$, $\mathit{out1}$, $\mathit{out2}$, $\mathit{out}$].

\begin{figure*}[h]
\begin{subfigure}{0.55\textwidth}
\includegraphics[width=\textwidth]{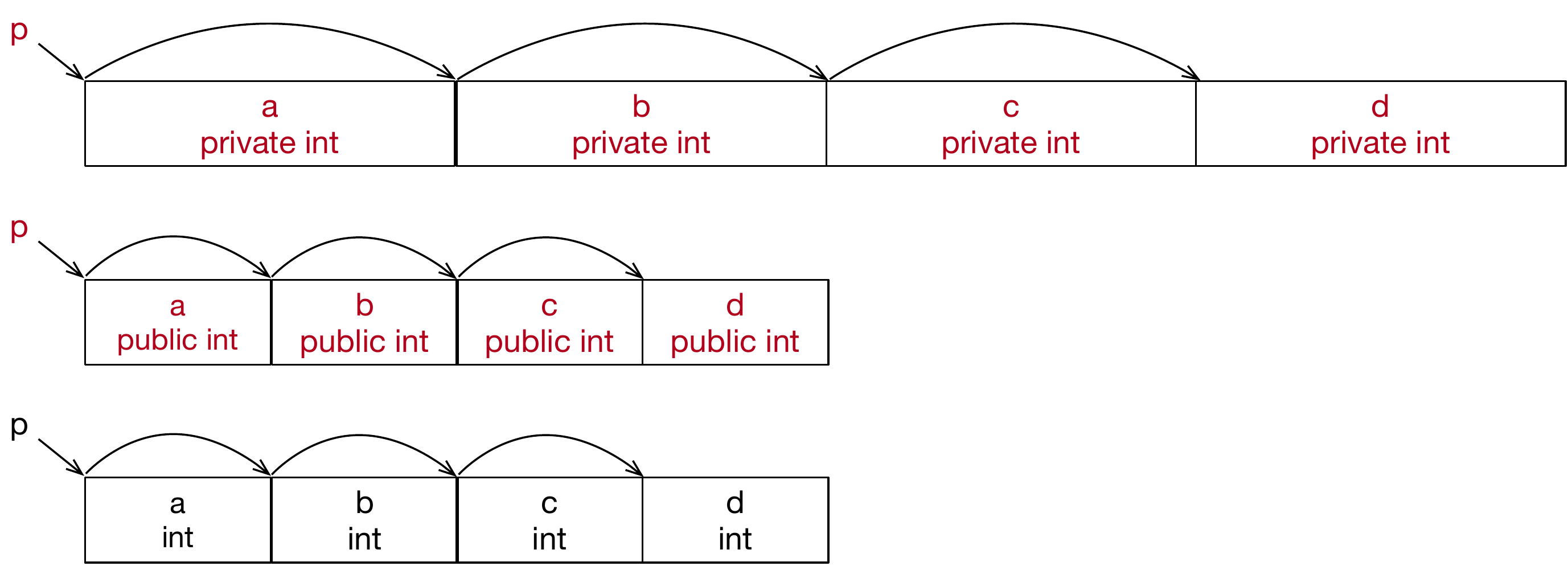}
\caption{Well-aligned accesses}
\label{fig: well-aligned}
\end{subfigure}
\quad
\begin{subfigure}{0.35\textwidth}
\includegraphics[width=\textwidth]{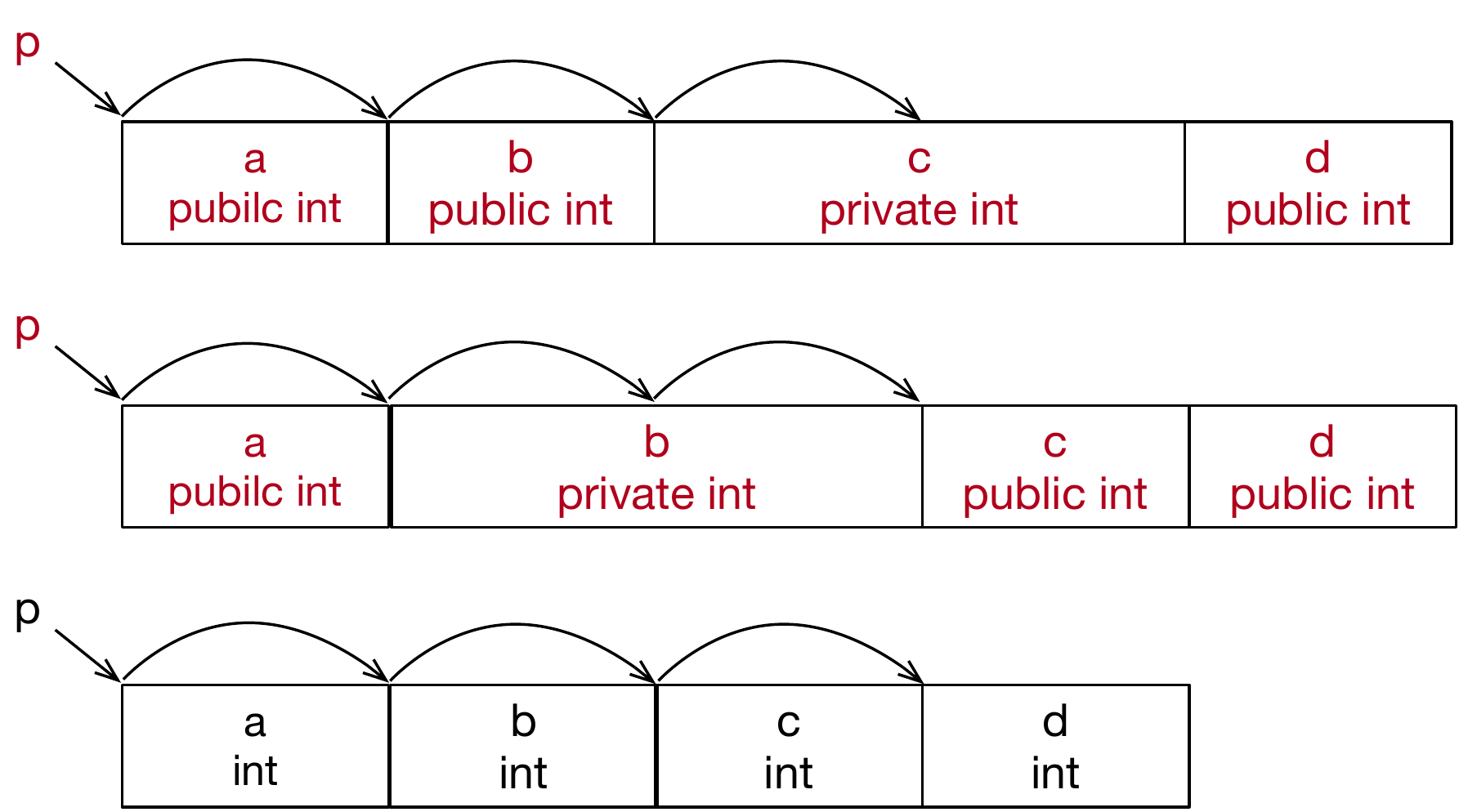}
\caption{Not well-aligned accesses}
\label{fig: not well-aligned}
\end{subfigure}
\caption{Examples of alignment between \piccoC\ and \vanillaC\ in overshooting accesses by incrementing pointer \TT{p} three times.}
\label{fig: overshooting alignment}
\end{figure*}

\subsubsection{Erasure Function}

\input{erasure-ef}

\subsubsection{Selected Metatheory}
\input{short-correctness}

\subsection{Noninterference}
\label{app: noninterference}

\input{short-ni}

%% file: app-rules.tex
\begin{figure*}[h]\footnotesize
\begin{tabular}{l}
\phantomsection\label{rule: da1}
Private Array Declaration \\
	\inferrule{\begin{array}{l}\begin{array}{l l}
		({\Expr}) \isPub \rrgamma \qq\qquad&
		((\Type = {\Priv\ \btype}) \lor (\Type = {\btype})) \land
			((\btype = \Int) \lor (\btype = \Float))
		\crcr
		&((\pid, \rrgamma,\ \sigma,\ {\DMap},\ \Acc,\ \Expr) \Mid  \Config) 
			\Deval{\locLL_1}{\codeLL_1} 
			((\pid, \rrgamma,\ \sigma{_1},\ {\DMap},\ \Acc,\ \nl) \Mid  \Config_1)
		\crcr \nl > 0
		& \byte = \EncodePtr({\Priv\ \Const\ \btype *},\ [1,\ [({\loc_1}, 0)],\ [1],\ 1]) 
		\crcr {\loc} = \phi() 
		& \byte_1 = \EncodeArr(\Priv\ \btype, \nl, \Null)   
		\crcr  {\loc_1} = \phi()
		&\RT{\gamma{_1} = \gamma[\x\ \to\ (\loc,\ {\Priv\ \Const\ \btype*})]} 
		\end{array}\crcr \begin{array}{l}
		\sigma{_2} = \sigma{_1}[{\loc} \to (\byte, {\Priv\ \Const\ \btype *}, 1, 
			\PtrPermL(\PermF, {\Priv\ \Const\ \btype*}, \Priv, 1))]
		\crcr 
		\sigma{_3} = \sigma{_2}[{\loc_1} \to (\byte_1,\ {\Priv\ \btype},\ \nl,\ 
			\ArrPermL(\PermF, \Priv\ \btype, \Priv, \nl))]
	\end{array}\end{array}}					
	{((\pid, \rrgamma,\ \sigma,\ {\DMap},\ \Acc,\ {\Type\ \x[\Expr]}) \Mid  \Config)\ 
		\Deval{\locLL_1 \addL (\pid, [(\loc, 0), ({\loc_1}, 0)])}{\codeLL_1 \addC \codeSP{da1}}  
		((\pid, \rgamma{_1},\ \sigma{_3},\ {\DMap},\ \Acc,\ \Skip) \Mid  \Config_1)}
\\ \\ 
\phantomsection\label{rule: ss}
Statement Sequencing \\
  	\inferrule{\begin{array}{l}
		((\pid,\ \rrgamma,\ \sigma,\ {\DMap},\ \Acc, {\stmt_1}) \Mid  \Config)\ 
			\Deval{\locLL_1}{\codeLL_1}  ((\pid, \rgamma{_1}, \sigma{_1}, {\DMap_1}, \Acc, \val_1) \Mid  \Config_1)
		\crcr
		((\pid, \rgamma{_1}, \sigma{_1}, {\DMap_1}, \Acc, {\stmt_2}) \Mid  \Config_1) 
			\Deval{\locLL_2}{\codeLL_2} ((\pid, \rgamma{_2}, \sigma{_2}, {\DMap_2}, \Acc, \val_2) \Mid  \Config_2)
	\end{array}}
	{((\pid, \rrgamma,\ \sigma,\ {\DMap},\ \Acc,\ {\stmt_1;\ \stmt_2}) \Mid  \Config)
		\Deval{\locLL_1 \addL \locLL_2}{\codeLL_1 \addC \codeLL_2 \addC \codeSP{ss}} 
		((\pid, \RT{\gamma{_2}},\ \sigma{_2},\ {\DMap_2},\ \Acc,\ {\val}) \Mid  \Config_2)}
\\ \\
\phantomsection\label{rule: iet}
Public If Else True \\
  	\inferrule{\begin{array}{l l}
		(\Expr) \isPub \rrgamma \qq &
		((\pid, \rrgamma,\ \sigma,\ {\DMap},\ \Acc, \Expr)\ \ \Mid  \Config)\ 
			\Deval{\locLL_1}{\codeLL_1}  ((\pid, \rrgamma,\ \sigma{_1}, {\DMap_1}, \Acc, \n)\quad\ \Mid  \Config_1) 
		\crcr \n\ \neq\ 0 
		& ((\pid, \rrgamma, \sigma{_1}, {\DMap_1}, \Acc, {\stmt_1}) \Mid  \Config_1) 
			\Deval{\locLL_2}{\codeLL_2} ((\pid, \rgamma{_1}, \sigma{_2}, {\DMap_2}, \Acc, \Skip) \Mid  \Config_2)
	\end{array}}
	{((\pid, \rrgamma,\ \sigma,\ {\DMap},\ \Acc,\ {\If\ (\Expr)\ \stmt_1\ \Else\ \stmt_2}) \Mid  \Config)\ 
		\Deval{\locLL_1\addL\locLL_2}{\codeLL_1 \addC \codeLL_2 \addC \codeSP{iet}}  ((\pid, \rrgamma,\ \sigma{_2},\ {\DMap_2},\ \Acc,\ \Skip) \Mid  \Config_2)}
\end{tabular}
\caption{An illustration of scoping within \DynamicPicco\ rules. We highlight the environment $\RT\gamma$ and its modifications in \red{red}.}
\label{Fig: scoping}
\end{figure*}

\begin{figure*}\footnotesize
\begin{tabular}{l}
\phantomsection\label{rule: bm}
Public Multiplication \\
\inferrule{\begin{array}{l}
		({\Expr_1}, {\Expr_2}) \isPub \gamma \qq\quad
		((\pid, \gamma,\ \sigma,\ {\DMap},\ \Acc,\ {\Expr_1}) \Mid  \Config) \ 
			\Deval{\locLL_1}{\codeLL_1}  ((\pid, \gamma,\ \sigma{_1},\ {\DMap_1},\ \Acc,\ {n_1}) \Mid  \Config_1) 
		\crcr 
		((\pid, \gamma,\ \sigma{_1},\ {\DMap_1},\ \Acc,\ {\Expr_2}) \Mid  \Config_1) \ 
			\Deval{\locLL_2}{\codeLL_2}  ((\pid, \gamma,\ \sigma{_2},\ {\DMap_2},\ \Acc,\ {n_2}) \Mid  \Config_2) 
		\qq {n_1} \cdot {n_2} = {n_3}
	\end{array}}
	{((\pid, \gamma,\ \sigma,\ {\DMap},\ \Acc,\ {\Expr_1 \cdot \Expr_2}) \Mid  \Config)\ 
		\Deval{\locLL_1\addL\locLL_2}{\codeLL_1\addC \codeLL_2 \addC \codeSP{bm}} ((\pid, \gamma,\ \sigma{_2},\ {\DMap_2},\ \Acc,\ {n_3}) \Mid  \Config_2) }
\\ \\
\phantomsection\label{rule: r1}
Read Private Variable \\
  	\inferrule{\begin{array}{l l} 
		\gamma(\x) = (\loc,\ \Priv\ \btype)\ \qquad 
		&\sigma(\loc) = (\byte,\ \Priv\ \btype,\ 1,\ \VarPermL(\PermF, \Priv\ \btype, \Priv, 1)) 
		\crcr &\Decode(\Priv\ \btype,\ \byte) = \n
	\end{array}}
	{((\pid, \gamma,\ \sigma,\ {\DMap},\ \Acc,\ \x) \Mid  \Config)\ 
		\Deval{(\pid, [(\loc, 0)])}{\codeSP{r1}}  
		((\pid, \gamma,\ \sigma,\ {\DMap},\ \Acc,\ \n) \Mid  \Config)}		
\\ \\
\phantomsection\label{rule: mal}
Public Malloc \\
	\inferrule{\begin{array}{l l l}
		\Acc = \AccZ\qq
		& (\Expr) \isPub \gamma\qq
		& ((\pid, \gamma,\ \sigma,\ {\DMap},\ \Acc,\ \Expr) \Mid  \Config)\ 
			\Deval{\locLL_1}{\codeLL_1}  ((\pid, \gamma,\ \sigma{_1},\ {\DMap},\ \Acc,\ \n) \Mid  \Config_1)
		\crcr \loc = \phi()
		&& \sigma_2 = \sigma_1\big[\loc \to \big(\Null, \Void*, \n, \PermL(\PermF, \Void*, \Pub, \n)\big)\big]
	\end{array}}
	{\begin{array}{l}
		((\pid, \gamma,\ \sigma,\ {\DMap},\ \Acc,\ \Malloc (\Expr)) \Mid  \Config) 
			\Deval{\locLL_1 \addL (\pid, [(\loc, 0)])}{\codeLL_1 \addC \codeSP{mal}} 
			((\pid, \gamma,\ \sigma{_2},\ {\DMap},\ \Acc,\ (\loc, 0)) \Mid  \Config_1) 
	\end{array}}
\\ \\
\phantomsection\label{rule: fre}
Public Free \\ 
  	\inferrule{\begin{array}{l}\begin{array}{l l}	
		\gamma(\x) = (\loc,\ {\Pub\ \btype*})\quad\
		&	\sigma(\loc) = (\byte, \Pub\ \btype*, 1, \PermL(\PermF, \Pub\ \btype*, \Pub, 1)) 
		\crcr \Acc = \AccZ \qq 
		& 	\DecodePtr({\Pub\ \btype*},\ 1,\ \byte) = [1,\ [({\loc_1}, 0)],\ [1],\ 1]
		\end{array}\crcr \begin{array}{l l}
		\SelectFreeable(\gamma, [({\loc_1}, 0)], [1], \sigma) = 1\qq
		&	 \Free(\sigma,\ \loc_1) = (\sigma{_1}, (\loc_1, 0))
	\end{array}\end{array}}
	{((\pid, \gamma,\ \sigma,\ {\DMap},\ \Acc,\ {\free (\x)}) \Mid  \Config)\ 
		\Deval{(\pid, [(\loc, 0), (\loc_1, 0)])}{\codeSP{fre}}  
		((\pid, \gamma,\ \sigma{_1},\ {\DMap},\ \Acc,\ \Skip) \Mid  \Config)}
\\ \\
\phantomsection\label{rule: inp}
SMC Input Public Value \\
	\inferrule{ \begin{array}{l l}
		({\Expr}) \isPub \gamma \qq
		&((\pid, \gamma, \sigma, {\DMap}, \Acc, {\Expr} ) \Mid  \Config)\ 
			\Deval{\locLL_1}{\codeLL_1} ((\pid, \gamma, \sigma{_1}, {\DMap_1}, \Acc, {n})\quad \Mid  \Config_1) 	\qquad
		\crcr
		\gamma(\x) = (\loc, \Pub\ \btype)
		& \InputVal(\x, {n}) = {n_1}	\qq
		\crcr 
		\Acc =\AccZ \qq
		&((\pid, \gamma, \sigma{_1}, {\DMap_1}, \Acc, {\x = n_1} ) \Mid  \Config_1) 
			\Deval{\locLL_2}{\codeLL_2} ((\pid, \gamma, \sigma{_2}, {\DMap_2}, \Acc, \Skip) \Mid  \Config_2)
	\end{array}}
	{((\pid, \gamma,\ \sigma,\ {\DMap},\ \Acc,\ {\smcinput(\x,\ \Expr)}) \Mid  \Config) 
		\Deval{\locLL_1 \addL \locLL_2}{\codeLL_1 \addC \codeLL_2 \addC \codeSP{inp}}  
		((\pid, \gamma,\ \sigma{_2},\ {\DMap_2},\ \Acc,\ \Skip) \Mid  \Config_2)}
\\ \\ 
\phantomsection\label{rule: out3}
SMC Output Private Array \\
	\inferrule{ \begin{array}{l}\begin{array}{l l}
		({\Expr_1}, {\Expr_2}) \isPub \gamma \qquad
		&	((\pid, \gamma,\ \sigma,\ {\DMap},\ \Acc, {\Expr_1} ) \Mid  \Config)\ 
				\Deval{\locLL_1}{\codeLL_1} ((\pid, \gamma, \sigma{_1}, {\DMap_1}, \Acc, {n}) \Mid  \Config_1) 
		\crcr \gamma(\x) = ({\loc}, {\Priv\ \Const\ \btype*}) \qquad
		&	((\pid, \gamma, \sigma{_1}, {\DMap_1}, \Acc, {\Expr_2} ) \Mid  \Config_1) 
				\Deval{\locLL_2}{\codeLL_2} ((\pid, \gamma, \sigma{_2}, {\DMap_2}, \Acc, {\nl}) \Mid  \Config_2)
		\end{array}\crcr\begin{array}{l}
			\sigma{_2}({\loc}) = ({\byte}, {\Priv\ \Const\ \btype*}, 1, 
					\PtrPermL(\PermF, {\Priv\ \Const\ \btype*}, \Priv, 1))
		\crcr \DecodePtr({\Priv\ \Const\ \btype*},\ 1,\ {\byte}) = [1,\ [({\loc_1}, 0)],\ [1],\ \Priv\ \btype,\ 1]
		\crcr \sigma{_2}({\loc_1}) = ({\byte_1},\ {\Priv\ \btype},\ {\nl},\ \ArrPermL(\PermF, \Priv\ \btype, \Priv, {\nl}))
		\crcr \forall \ind \in \{0, ..., \nl-1\} \qquad \DecodeArr({\Priv\ \btype},\ {\ind},\ {\byte_1})  = {m_\ind}
		\crcr \OutputArr(\x,\ {n},\ [{m_0},\ {...},\ {m_{\nl-1}}])
	\end{array}\end{array}}
	{\begin{array}{l}
	((\pid, \gamma,\ \sigma,\ {\DMap},\ \Acc,\ {\smcoutput(\x,\ \Expr_1,\ \Expr_2)}) \Mid  \Config) 
		\crcr\Deval{\locLL_1 \addL \locLL_2 \addL (\pid, [(\loc, 0), (\loc_1, 0), ..., (\loc_1, n_1-1)])}{\codeLL_1 \addC \codeLL_2 \addC \codeSP{out3}}  
		((\pid, \gamma,\ \sigma{_2},\ {\DMap_2},\ \Acc,\ \Skip) \Mid  \Config_2)
		\end{array}}
\\ \\ 
\phantomsection\label{rule: pin}
Pre-Increment Public Variable \\
  	\inferrule{\begin{array}{l l}
		\gamma (\x) = (\loc,\ \Pub\ \btype) \qquad
		&	\sigma(\loc) = (\byte,\ \Pub\ \btype,\ 1,\ \VarPermL(\PermF, \Pub\ \btype, \Pub, 1))
		\crcr \Acc =\AccZ \qquad
		&	\Decode(\Pub\ \btype,\ \byte) = {\n}	
		\crcr {\n_1} = {\n} + 1 
		&	\Update(\sigma,\ \loc,\ {\n_1},\ {\DMap},\ \Acc, \Pub\ \btype) = (\sigma{_1},\ {\DMap})
	\end{array}} 
	{((\pid, \gamma,\ \sigma,\ {\DMap},\ \Acc,\ {\plpl\x}) \Mid  \Config)\ 
		\Deval{(\pid, [(\loc, 0)])}{\codeSP{pin}}  
		((\pid, \gamma,\ \sigma{_1},\ {\DMap},\ \Acc,\ {\n_1}) \Mid  \Config)}
\\ \\
\phantomsection\label{rule: mprdp}
Multiparty Private Pointer Dereference Single Level Indirection\\ 
	\inferrule{\begin{array}{l}
		\{(\x) \isPriv \gamma^\pid\}^{\pidZ}_{\pid = \pidA} 
		\qq \{ \gamma^\pid(\x) = (\loc^\pid, \Priv\ \btype*)\}^{\pidZ}_{\pid = \pidA}
		\qq  \nl > 1  \crcr
		\{\sigma^\pid(\loc^\pid) = (\byte^\pid,\ {\Priv\ \btype*},\ \nl,\ \PtrPermL(\PermF, {\Priv\ \btype*}, \Priv, \nl))\}^{\pidZ}_{\pid = \pidA}
		\crcr 
		\{\DecodePtr({\Priv\ \btype*},\ \nl,\ \byte^\pid) = [\nl,\ \locL^\pid,\ \tagbL^\pid, 1] \}^{\pidZ}_{\pid = \pidA}
		\crcr 
		\{\Retrieve(\nl, \locL^\pid, \Priv\ \btype, \sigma^\pid) = ([\n^\pid_0, ...\n^\pid_{\nl-1}], 1)\}^{\pidZ}_{\pid = \pidA}
		\crcr \MPC{dv}([[\n^{\pidA}_{0}, ..., \n^{\pidA}_{\nl-1}], ..., [\n^{\pidZ}_{0}, ..., \n^{\pidZ}_{\nl-1}]], [\tagbL^\pidA, ..., \tagbL^\pidZ]) = (\n^{\pidA}, ..., \n^{\pidZ})
	\end{array}}
	{\begin{array}{l}
	((\pidA, \gamma^{\pidA}, \sigma^{\pidA}, \DMap^\pidA, \Acc, {* \x})\Mid ...\Mid 
	(\pidZ, \gamma^{\pidZ}, \sigma^{\pidZ}, \DMap^\pidZ, \Acc, {* \x}))
		\Deval{(\pidA, (\loc^\pidA, 0)\addL\locL^{\pidA}) \Mid ... \Mid (\pidZ, (\loc^\pidZ, 0)\addL\locL^{\pidZ})}{\codeMP{mprdp}}  
		\crcr((\pidA, \gamma^{\pidA}, \sigma^{\pidA}, \DMap^\pidA, \Acc, \n^{\pidA}_{})\Mid ...\Mid 
		(\pidZ, \gamma^{\pidZ}, \sigma^{\pidZ}, \DMap^\pidZ, \Acc, \n^{\pidZ}_{}))
		\end{array}}
\end{tabular}
\caption{Additional \DynamicPicco\ semantic rules.}
\label{Fig: sem app}
\end{figure*}

\begin{figure*}\footnotesize
\begin{tabular}{l}
\phantomsection\label{rule: ra1}
Private Array Read Public Index \\
  	\inferrule{
	\begin{array}{l}
	(\Expr) \isPub \gamma \qq
	((\pid, \gamma,\ \sigma,\ {\DMap},\ \Acc,\ \Expr) \Mid  \Config)\ 
		\Deval{\locLL_1}{\codeLL_1}  ((\pid, \gamma,\ \sigma{_1},\ {\DMap_1},\ \Acc,\ \ind) \Mid  \Config_1) \crcr
	\gamma(\x) = (\loc,\ {\Priv\ \Const\ \btype*}) \crcr
	\sigma{_1}(\loc) = (\byte,\ {\Priv\ \Const\ \btype*}, 1, 
		\PtrPermL(\PermF, {\Priv\ \Const\ \btype*}, \Priv, 1))
	\crcr \DecodePtr({\Priv\ \Const\ \btype*},\ 1,\ \byte) = [1,\ [({\loc_1}, 0)],\ [1],\ 1] 
	\qquad 0 \leq \ind \leq {\nl} -1 
	\crcr \sigma{_1}({\loc_1}) = ({\byte_1}, {\Priv\ \btype}, {\nl}, 
		\ArrPermL(\PermF, \Priv\ \btype, \Priv, {\nl})) 
	\crcr \DecodeArr({\Priv\ \btype},\ \ind,\ {\byte_1}) = \n_\ind 
	\end{array}}
	{ ((\pid, \gamma,\ \sigma,\ {\DMap},\ \Acc,\ {\x[\Expr]}) \Mid  \Config)\ 
		\Deval{\locLL_1 \addL (\pid, [(\loc, 0), ({\loc_1}, \ind)])}{\codeLL_1 \addC \codeSP{ra1}}  ((\pid, \gamma,\ \sigma{_1},\ {\DMap_1},\ \Acc,\ {\n_\ind}) \Mid  \Config_1)}
\\ \\ 
\phantomsection\label{rule: wa}
Public Array Write Public Value Public Index \\
  	\inferrule{\begin{array}{l}
		({\Expr_1}, {\Expr_2}) \isPub \gamma \qquad\quad
		((\pid, \gamma,\ \sigma,\ {\DMap},\ \Acc, {\Expr_1}) \Mid  \Config)\ 
			\Deval{\locLL_1}{\codeLL_1}  ((\pid, \gamma,\ \sigma{_1},\ {\DMap_1},\ \Acc,\ \ind) \Mid  \Config_1) \crcr
		\Acc = \AccZ \qq\ \
		((\pid, \gamma, \sigma{_1}, {\DMap_1}, \Acc, {\Expr_2}) \Mid  \Config_1) 
			\Deval{\locLL_2}{\codeLL_2}  ((\pid, \gamma,\ \sigma{_2},\ {\DMap_2},\ \Acc,\ \n) \Mid  \Config_2) \crcr
		\gamma(\x) = (\loc,\ {\Pub\ \Const\ \btype*})	\crcr
		\sigma{_2}(\loc) = (\byte, {\Pub\ \Const\ \btype*}, 1, \PtrPermL(\PermF, {\Pub\ \Const\ \btype*}, \Pub, 1)) 
		\crcr \DecodePtr({\Pub\ \Const\ \btype*}, 1, \byte) = [1,\ [({\loc_1, 0})],\ [1],\ 1]
		\crcr \sigma{_2}({\loc_1}) = ({\byte_1}, \Pub\ \btype, \nl, \ArrPermL(\PermF, \Pub\ \btype, \Pub, \nl))
		\qquad 0 \leq {\ind} \leq {\nl-1} 
		\crcr 
		 \UpdateArr(\sigma{_2},\ ({\loc_1}, \ind),\ \n,\ \Pub\ \btype) = \sigma{_3}
	\end{array}}
	{((\pid, \gamma,\ \sigma,\ {\DMap},\ \Acc,\ {\x[\Expr_1]\ = \Expr_2}) \Mid  \Config)\ 
		\Deval{\locLL_1 \addL \locLL_2 \addL (\pid, [(\loc, 0), (\loc_1, \ind)])}{\codeLL_1 \addC \codeLL_2 \addC \codeSP{wa}}  
		((\pid, \gamma,\ \sigma{_3},\ {\DMap_2},\ \Acc,\ \Skip) \Mid  \Config_2)}
\\ \\ 
\phantomsection\label{rule: rp}
Pointer Read Single Location \\
	\inferrule{\begin{array}{l l}
		\gamma(\x) = (\loc,\ {\llabel\ \btype*}) \qquad 
		& \sigma(\loc) = (\byte,\ {\llabel\ \btype*},\ 1,\ 
			\PtrPermL(\PermF, {\llabel\ \btype*}, \llabel, 1))
		\crcr &\DecodePtr({\llabel\ \btype*},\ 1,\ \byte) = [1,\ [({\loc_1}, \offset_1)],\ [1],\ \indir] 
	\end{array}}
	{((\pid, \gamma,\ \sigma,\ {\DMap},\ \Acc,\ \x) \Mid  \Config)\ 
		\Deval{(\pid, [(\loc, 0)])}{\codeSP{rp}}  ((\pid, \gamma,\ \sigma,\ {\DMap},\ \Acc,\ {(\loc_1, \offset_1)}) \Mid  \Config)}
\\ \\
\phantomsection\label{rule: wp1}
Private Pointer Write \\
  	\inferrule{\begin{array}{l}
		\gamma(\x) = ({\loc}, {\Priv\ \btype*}) \qquad\
		((\pid, \gamma, \sigma, {\DMap}, \Acc, \Expr) \Mid  \Config) 
			\Deval{\locLL_1}{\codeLL_1} ((\pid, \gamma, \sigma{_1}, {\DMap_1}, \Acc, ({\loc_e}, \offset_e)) \Mid  \Config_1) \crcr
		(\Expr) \isPub \gamma \qq
		\sigma{_1}({\loc}) = (\byte,\ {\Priv\ \btype*},\ \nl, \PtrPermL(\PermF, {\Priv\ \btype*}, \Priv, \nl))
		\crcr \DecodePtr({\Priv\ \btype*},\ \nl,\ \byte) = [\nl,\ \locL,\ \tagbL,\ \indir]
		\crcr \UpdatePtr(\sigma{_1},\ {(\loc, 0)},\ [1,\ [({\loc_e}, \offset_e)],\ [1],\ \indir],\ 
			{\DMap_1},\ \Acc, \Priv\ \btype*) = (\sigma{_2},\ {\DMap_2}, 1) 
	\end{array}} 
	{((\pid, \gamma,\ \sigma,\ {\DMap},\ \Acc,\ {\x = \Expr}) \Mid  \Config)\ 
		\Deval{\locLL_1 \addL (\pid, [(\loc, 0)])}{\codeLL_1\addC\codeSP{wp1}}  ((\pid, \gamma,\ \sigma{_2},\ {\DMap_2},\ \Acc,\ \Skip) \Mid  \Config_1)}
\\ \\ 
\phantomsection\label{rule: wdp3}
Private Pointer Dereference Write Single Location Private Value \\ 
	\inferrule{
	\begin{array}{l}
		(\Expr) \isPriv \gamma	\qq
		((\pid, \gamma,\ \sigma,\ {\DMap},\ \Acc,\ \Expr) \Mid  \Config) 
			\Deval{\locLL_1}{\codeLL_1} ((\pid, \gamma,\ \sigma{_1},\ {\DMap_1},\ \Acc,\ \n) \Mid  \Config_1) 
		\crcr \gamma(\x) = (\loc,\ {\Priv\ \btype*}) 
		\qq (\btype = \Int) \lor (\btype = \Float) 
		\crcr \sigma{_1}(\loc) = (\byte,\ {\Priv\ \btype*},\ 1,\ \PtrPermL(\PermF, {\Priv\ \btype*}, \Priv, 1)) 
		\crcr \DecodePtr({\Priv\ \btype*},\ 1,\ \byte) = [1,\ [(\loc_1, \offset_1)],\ [1],\ 1]
		\crcr \DynUpdate(\DMap_1, \sigma, [(\loc_1, \offset_1)],\ \Acc, \Priv\ \btype) = \DMap_2
		\crcr \UpdateOffset(\sigma{_1},\ ({\loc_1}, \offset_1),\ \n, \Priv\ \btype) = (\sigma{_2}, 1)
	\end{array}}
	{((\pid, \gamma,\ \sigma,\ {\DMap},\ \Acc,\ {* \x = \Expr}) \Mid  \Config) 
		\Deval{\locLL_1 \addL (\pid, [(\loc, 0), ({\loc_1}, \offset_1)])}{\codeLL_1\addC\codeSP{wdp3}}  
		((\pid, \gamma,\ \sigma{_2},\ {\DMap_2},\ \Acc,\ \Skip) \Mid  \Config_1) }
\\ \\
\phantomsection\label{rule: pin1}
Pre-Increment Public Pointer Single Location \\
  	\inferrule{\begin{array}{l}
		\gamma (\x) = (\loc,\ {\Pub\ \btype*}) \crcr
		\sigma(\loc) = (\byte,\ {\Pub\ \btype*},\ 1,\ \PtrPermL(\PermF, {\Pub\ \btype*}, \Pub, 1))
		\crcr \DecodePtr({\Pub\ \btype*},\ 1,\ \byte) = [1,\ [({\loc_1}, \offset_1)],\ [1],\ 1]
		\crcr (({\loc_2}, \offset_2), 1) = \GetLoc(({\loc_1}, \offset_1), \tau({\Pub\ \btype}), \sigma)
		\crcr \UpdatePtr(\sigma,\ (\loc, 0),\ [1,\ [({\loc_2}, \offset_2)],\ [1],\ 1],\ {\DMap},\ \Acc, 
			\Pub\ \btype*) = (\sigma{_1},\ {\DMap_1}, 1)
	\end{array}} 
	{((\pid, \gamma,\ \sigma,\ {\DMap},\ \Acc,\ {\plpl\x}) \Mid  \Config)\ 
		\Deval{(\pid, [(\loc, 0)])}{\codeSP{pin1}}  
		((\pid, \gamma,\ \sigma{_1},\ {\DMap_1},\ \Acc,\ ({\loc_2}, \offset_2)) \Mid  \Config)}
\\ \\
\phantomsection\label{rule: rdp}
Pointer Dereference Single Location  \\ 
	\inferrule{\begin{array}{l}
		\gamma(\x) = (\loc,\ {\llabel\ \btype*}) 
		\qq  \sigma({\loc}) = ({\byte},\ {\llabel\ \btype*},\ 1,\ \PtrPermL(\PermF, {\llabel\ \btype*}, \llabel, 1)) 
		\crcr \DecodePtr({\llabel\ \btype*},\ 1,\ {\byte}) = [1,\ [({\loc_1}, \offset_1)],\ [1],\ 1] 
		\qquad\quad \DerefPtrPub(\sigma, \llabel\ \btype, ({\loc_1}, \offset_1)) = (\n, 1)
	\end{array}}
	{((\pid, \gamma,\ \sigma,\ {\DMap},\ \Acc,\  {* \x}) \Mid  \Config)\ 
		\Deval{(\pid, [(\loc, 0), ({\loc_1}, \offset_1)])}{\codeSP{rdp}}  ((\pid, \gamma,\ \sigma,\ {\DMap},\ \Acc,\ \n) \Mid  \Config)}
\end{tabular}
\caption{Additional \DynamicPicco\ semantic rules for arrays and pointers.}
\label{Fig: sem app arr}
\end{figure*}

%% file: overshooting.tex
It is possible to overshoot memory bounds in both \vanillaC\ and \piccoC. 
Figure~\ref{Fig: arr oob read vanC} shows an example of an array read that overshoots the bounds of the array \TT{x} (for simplicity, \TT{x} is of size 1).  
The first access shown  is an in-bounds access -- this is the default behavior of a  correct
program.  The second access is an out-of-bounds access of a memory block of a different type, but the same size. This data would be read as if it was the type of the array, and may not be meaningful.  This corresponds to an access where implicit conversions between types is possible, but not always correct.
The third out-of-bounds access corresponds to reading out of a memory block of the same type. This data would be meaningful from a type perspective, but the specific value read may not be semantically meaningful to the program. 
 The fourth out-of-bounds access is of a memory block of a different type of a smaller size. This read would grab the data from the smaller memory block, then grab data from the next memory block(s) to obtain the correct amount for the expected type.  In this situation a value, which may not be meaningful, is constructed from two, or more, values in memory.
 The last out-of-bounds access is of a larger memory block, not aligned. This read would obtain a portion of the data of the larger memory block, and read it as the type of the array, thereby reading a partial value from memory.

With \piccoC, when dealing with array overshooting, we have the added complexity of private data, which has a different representation and is of a larger size than the corresponding C representation of the type. 
Additionally, we need to ensure that no leakage can occur, so we must consider all possible combinations of bytes from public and private data with either public or private variables. 
Consider reading a value from an overshot array and storing it into a variable. If both the data read and variable are private or both are public, no leakage can occur as these are the default cases.
Next, consider reading public data and storing into a private variable. The public data will be grabbed at the byte-level, and interpreted as though it is private (no encryption will occur), so no leakage occurs.
Third, we consider reading private data and storing in a public variable. The private data will be grabbed at the byte-level, and interpreted as though it is public. No decryption will occur, so no leakage can occur. This is similar in nature to reading a partial value in Figure~\ref{Fig: arr oob read vanC}. 
Fourth, consider if the data read is a mix of public and private data and stored in a public variable. Given that the private data will not be decrypted, this read will not result in any leakage, but a value is constructed from a mix
of private (encrypted) data and public data. 
Lastly, consider reading a mix of public and private data and storing into a private variable. Like before, the byte-level data will be merged and read as the expected type.

Writes that occur out-of-bounds of an array have situations similar to out-of-bound read accesses (and can be illustrated as with the reads shown in Figure~\ref{Fig: arr oob read vanC}).
Writing private data out-of-bounds to a private location results in the data still residing in a private memory block, so no leakage will occur. 
Writing public data out-of-bounds to a public location is safe, as the data is already public. 
When writing private data out-of-bounds to a public location, the data will be written as-is -- no decryption will occur when the data is written to or later read back from that location -- therefore, there is no leakage.
Writing public data out-of-bounds to a private location or a mix of public and private locations is safe, as the data was already public; no encryption will occur. 
Lastly, writing private data out-of-bounds to a mix of locations will result in the data being written to the locations as-is. No decryption will occur when the data is written to any location or later read back, therefore, there is no leakage.

In \piccoC, we ensure this behavior, 
using algorithms $\ReadOOB$ and $\WriteOOB$. 
In particular, $\ReadOOB$ ensures that no matter what mix of byte-wise data we grab from memory, we will decode it as a value of the type of data in the array, ignoring it's true type. 
Similarly, $\WriteOOB$ ensures that we will write to memory the byte-wise encoding of the given value as the type for the array, without taking into consideration the type of the memory block(s) and without modifying any of the metadata within the memory block(s) we write to. 
In proving the correctness of \piccoC\ with relation to \vanillaC, the various possible alignments for reading and writing out-of-bounds poses complications due to the different sizes of private and public data (an example of this is shown in the Appendix, Figure~\ref{fig: overshooting alignment}). Therefore, we can only prove correctness over well-aligned accesses (i.e., those that iterate only over aligned elements of the same type). 
When proving noninterference, we must prove that these cases (particularly those involving private data) cannot leak any information about the private data that is affected. We discuss this in more detail in
the following section.

%% file: erasure-ef.tex
Here, we show the full erasure function in Figure~\ref{Fig: app erasure}. 
This function is intended to take a \piccoC\ program or configuration and remove all private privacy labels, decrypt any private data, and clear any additional tracking features that are specific to \piccoC; this process will result in a \vanillaC\ program or configuration.

\begin{figure*}
\footnotesize
\begin{minipage}{0.5\textwidth}
\begin{subfigure}{\textwidth}
$\begin{array}{l}
\bm{\erasure(\rstmt)} =  \\ 
	\mid\ \RT{\x[\Expr]}\ =>\ \bm{\x[\erasure(\rExpr)]}  \\ 
	\mid\ \RT{\x(\Elist)}\ =>\ \bm{\x(\erasure(\RT{\Elist}))} \\ 
	\mid\ \RT{\Expr_1\ \binop\ \Expr_2}\ =>\ \bm{\erasure(\RT{\Expr_1})\ \binop\ \erasure(\RT{\Expr_2})}  \\ 
	\mid\ \RT{\unop\ \x}\ =>\ \bm{\unop\ \x}  \\ 
	\mid\ \RT{( \Expr )}\ =>\ \bm{(\erasure(\rExpr))}  \\ 
	\mid\ \RT{(\Type)\ \Expr}\ =>\ \bm{\erasure(\rType))\ \erasure(\rExpr)} \\ 
	\mid\ \RT{\var = \Expr}\ =>\ \bm{\erasure(\RT\var) = \erasure(\rExpr)}  \\ 
	\mid\ \RT{*\x = \Expr}\ =>\ \bm{*\x = \erasure(\rExpr)}  \\ 
	\mid\ \RT{\stmt_1;\ \stmt_2}\ =>\ \bm{\erasure(\RT{\stmt_1});\ \erasure(\RT{\stmt_2})}  \\ 
	\mid\ \RT{\{ \stmt \}}\ =>\ \bm{\{ \erasure(\rstmt) \}}  \\ 
	\mid\ \RT{\free(\Expr)}\ =>\ \bm{\free(\erasure(\rExpr))}  \\ 
	\mid\ \RT{\pfree(\Expr)}\ =>\ \bm{\free(\erasure(\rExpr))}  \\ 
	\mid\ \RT{\sizeof(\Type)}\ =>\ \bm{\sizeof(\erasure(\rType))}  \\ 
	\mid\ \RT{\Malloc(\Expr)}\ =>\ \bm{\Malloc(\erasure(\rExpr))}  \\ 
	\mid\ \RT{\PMalloc(\Expr,\ \Type)}\ =>\ \\ \ \quad \bm{\Malloc(\sizeof(\erasure(\rType)) \cdot \erasure(\rExpr))}  \\ 
	\mid\ \RT{\smcinput(\Elist)}\ =>\ \\ \ \quad \bm{\inputFun(\erasure(\RT\Elist))}  \\ 
	\mid\ \RT{\smcoutput(\Elist)}\ =>\ \\ \ \quad \bm{\outputFun(\erasure(\RT\Elist))} \\ 
	\mid\ \RT{[\val_0, ..., \val_n]} =>\ \\ \ \quad  \bm{[\erasure(\RT{\val_0}),\ \erasure(\RT{...}),\ \erasure(\RT{\val_n})]} \\ 
	\mid\ \RT{\Type\ \var}\ =>\ \\ \ \quad  \bm{\erasure(\rType)\ \erasure(\RT{\var})}  \\ 
	\mid\ \RT{\Type\ \var = \Expr}\ =>\ \\ \ \quad  \bm{\erasure(\rType)\ \erasure(\RT{\var}) = \erasure(\rExpr)}  \\ 
	\mid\ \RT{\Type\ \x(\plist)}\ =>\ \\ \ \quad  \bm{\erasure(\rType)\ \x(\erasure(\RT{\plist}))}  \\ 
	\mid\ \RT{\Type\ \x (\plist)\ \{ \stmt\}}\ =>\ \\ \ \quad  \bm{\erasure(\RT{\Type\ \x (\plist)})\ \{ \erasure(\rstmt) \} }  \\ 
	\mid\ \RT{\If (\Expr)\ \stmt_1\ \Else\ \stmt_2}\ =>\ \\ \ \quad \bm{\If (\erasure(\rExpr))\ \erasure(\RT{\stmt_1})\ \Else\ \erasure(\RT{\stmt_2})}  \\ 
	\mid\ \RT{\While\ (\Expr)\ \stmt}\ \\ \ \ \ =>\ \bm{\While\ (\erasure(\rExpr))\ \erasure(\rstmt)}  \\ 
	\mid\ \RT\_\ =>\ \bm{\stmt}
\end{array}$
\caption{Erasure function over statements} 	\label{Fig: erasure stmt}
\end{subfigure}
\end{minipage}
\begin{minipage}{0.4\textwidth}
\begin{subfigure}{\textwidth}
$\begin{array}{l}
\bm{\erasure(\RT\Config)} = 
	\\ \mid \RT{\Config_1} \Mid \RT{\Config_2}\ =>\ \bm{\erasure(\RT{\Config_1})} \Mid \bm{\erasure(\RT{\Config_2})}
	\\ \mid (\RT\pid, \rrgamma,\ \rrsigma,\ \RT{\DMap}, \rAcc,\ \rstmt) \ =>\
		\\ \ \ \ \bm{(\pid, \erasure(\rrgamma, \rrsigma, [\ ], [\ ]), \bsq, \bsq, \erasure(\rstmt))}
\end{array}$
\caption{Erasure function over configurations} 	\label{Fig: erasure config}
\end{subfigure}
\\ \\ \\
\begin{subfigure}{\textwidth}
$\begin{array}{l}
\bm{\erasure(\rType)} =  \\  	
	\mid\ \rlabel\ \rbtype\ =>\ \bm{\btype}  \\ 
	\mid\ \RT{\rlabel\ \rbtype\ *}\ =>\ \bm{\btype*}  \\ 
	\mid\ \RT{\Tlist \to \rType}\ =>\ \\ \ \quad \bm{\erasure(\RT{\Tlist}) \to \erasure(\rType))}  \\ 
	\mid\ \RT\_\ => \bm{\Type} \\ \\
\bm{\erasure(\rTlist)} =  \\ 
	\mid\ \RT{[\ ]}\ =>\ \bm{[\ ]}  \\ 
	\mid\ \RT{\Type::\Tlist}\ =>\ \bm{\erasure(\RT{\Type})::\erasure(\RT{\Tlist})} 
\end{array}$
\caption{Erasure function over types and type lists} 	\label{Fig: erasure ty}
\end{subfigure}
\\ \\ \\
\begin{subfigure}{\textwidth} 
$\begin{array}{l}
\bm{\erasure(\rElist)} = \\ 
	\mid\ \RT{\Elist,\ \Expr}\ =>\ \bm{\erasure(\RT{\Elist}),\ \erasure(\rExpr)}  \\ 
	\mid\ \rExpr\ =>\ \bm{\erasure(\rExpr)}  \\ 
	\mid\ \rVoid\ =>\ \bm{\Void} \\  \\
\bm{\erasure(\rplist)} =  \\ 
	\mid\ \RT{\plist,\ \Type\ \var}\ =>\ \\ \ \quad \bm{\erasure(\RT{\plist}),\ \erasure(\RT{\Type\ \var})}  \\ 
	\mid\ \RT{\Type\ \var}\ =>\ \bm{\erasure(\rType)\ \erasure(\RT\var)}  \\ 
	\mid\ \rVoid\ =>\ \bm{\Void}
\end{array}$
\caption{Erasure function over lists} 	\label{Fig: erasure list}
\end{subfigure}
\end{minipage}
\\ \\ \\
\centering
\begin{subfigure}{0.8\textwidth}
$\begin{array}{l}
\bm{\erasure(\rbyte,\ \rType,\ \rn)} = \\ 
	\mid\ (\rbyte,\ \rPub\ \rbtype,\ \rn) =>\ \bm{\byte} 	\\ 
	\mid\ (\rbyte,\ \rPriv\ \rbtype,\ \RT1) =>\  \\ \-\ \quad
		\RT{\val_1 = \Decode(\Type,\ 1,\ \byte)};\ 
		\val_2 = \RT{\Decrypt(\val_1)};\ 
		\byte_1 = \Encode(\btype,\ \val_2);\ 
		\bm{\byte_1}	\\ 
	\mid\ (\rbyte,\ \rPriv\ \rbtype,\ \rn) =>\  
		\RT{\val_1 = \Decode(\Type,\ n,\ \byte)};\ \\ \ \quad
		[\val_1', ..., \val_n'] = \RT{[\Decrypt(\val_1), \Decrypt(...), \Decrypt(\val_n)]};\ 
		\byte_1 = \Encode(\btype,\ [\val_1',\ ...,\ \val_n']);\ 
		\bm{\byte_1} \\ 	
	\mid\ (\rbyte,\ \RT{\Pub\ \btype\ *},\ \RT1) =>\ 
		[\RT1,\ [(\rloc, \RT\offset)],\ [\RT1],\ \RT\indir] = \RT{\DecodePtr(\Pub\ \btype\ *,\ 1,\ \byte)};\ \\ \-\ \quad
		\byte_1 = \EncodePtr(\btype\ *,\ [1,\ [(\loc, \offset)],\ [1],\ \erasure(\RT{\Type'}),\ \indir]);\ 
		\bm{\byte_1} \\ 
	\mid\ (\rbyte,\ \RT{\Priv\ \btype\ *},\ \RT1) =>\ 
		[\RT1,\ [(\rloc, \RT\offset)],\ [\RT1],\ \RT\indir] = \RT{\DecodePtr(\Priv\ \btype\ *,\ 1,\ \byte)};\ \\ \-\ \quad
		\If (\RT\indir = 1) \Then \{ \RT{\Type_1} = \rPub\ \rbtype; \RT{\Type_2} = \rPriv\ \rbtype \} \Else \{ \RT{\Type_1} = \rPub\ \rbtype\RT*; \RT{\Type_2} = \rPriv\ \rbtype\RT* \}; \ \\ \-\ \quad 
		\offset_1 = \frac{\RT\offset \cdot \tau(\RT{\Type_1})}{\tau(\RT{\Type_2})};\  
		\byte_1 = \EncodePtr(\btype\ *,\ [1,\ [(\loc, \offset_1)],\ [1],\ \erasure(\RT{\Type'}),\ \indir]);\ 
		\bm{\byte_1} \\ 
	\mid\ (\rbyte,\ \RT{\Priv\ \btype\ *},\ \rn) =>\ 
		[\RT\nl,\ \RT\locL,\ \RT\tagbL,\ \RT\indir] = \RT{\DecodePtr(\Priv\ \btype\ *,\ n,\ \byte)};\  \\ \-\ \quad
		(\loc, \RT\offset) = \rrDeclassifyPtr([\RT\nl,\ \RT\locL,\ \RT\tagbL,\ \RT\indir],\ \RT{\rPriv\ \btype*});\ \\ \-\ \quad
		\If (\RT\indir = 1) \Then \{ \RT{\Type_1} = \rPub\ \rbtype; \RT{\Type_2} = \rPriv\ \rbtype \} \Else \{ \RT{\Type_1} = \rPub\ \rbtype\RT*; \RT{\Type_2} = \rPriv\ \rbtype\RT* \}; \ \\ \-\ \quad
		\offset_1 = \frac{\RT\offset \cdot \tau(\RT{\Type_1})}{\tau(\RT{\Type_2})};\  
		\byte_1 = \EncodePtr(\btype\ *,\ [1,\ [(\loc, \offset_1)],\ [1],\ \ \indir]);\ 
		\bm{\byte_1} \\ 
	\mid\ (\rbyte,\ \RT{\Tlist \to \Type},\ \RT1) =>\  \\ \-\ \quad
		\RT{(\rstmt,\ \rn,\ \rplist) = \DecodeFun(\byte)};\ 
		\byte_1 = \EncodeFun(\erasure(\rstmt),\ \bsq,\ \erasure(\rplist));\ 
		\bm{\byte_1}
\end{array}$
\caption{Erasure function over bytes} 	\label{Fig: erasure bytes}
\end{subfigure}
\caption{The Erasure function, broken down into various functionalities. } 	\label{Fig: app erasure}
\end{figure*}

Figure~\ref{Fig: erasure config} shows erasure over an entire configuration, calling $\erasure$ on the four-tuple of the environment, memory, and two empty maps needed as the base for the \vanillaC\ environment and memory; removing the accumulator (i.e., replacing it with $\Box$); and calling $\erasure$ on the statement. 
Figure~\ref{Fig: erasure ty} shows erasure over types and type lists (i.e., for function types).
Here, we remove any privacy labels given to the types, with unlabeled types being returned as is. 
For function types, we must iterate over the entire list of types as well as the return type. 
Figure~\ref{Fig: erasure list} shows erasure over expression lists (i.e., from function calls) and parameter lists (i.e., from function definitions). 

Figure~\ref{Fig: erasure stmt} shows erasure over statements. 
For statements, we case over the various possible statements. 
When we reach a private value (i.e., $\Encrypt(n)$), we decrypt and then return the decrypted value. 
For function $\PMalloc$, we replace the function name with $\Malloc$, modifying the argument to appropriately evaluate the expected size of the type. 
For functions $\pfree$, $\smcinput$, and $\smcoutput$, we simply replace the function name with its \vanillaC\ equivalent. 
All other cases recursively call the erasure function as needed, with the last case $(\_)$ handling all cases that are already identical to the \vanillaC equivalent (i.e., $\Null$, locations). 

Figure~\ref{Fig: erasure bytes} shows erasure over bytes stored in memory, which is used from within the erasure on the environment and memory. 
This function takes the byte-wise data representation, the type that it should be interpreted as, and the size expected for the data. 
For regular public types, we do not need to modify the byte-wise data. 
For regular private types (i.e., single values and array data), we get back the value(s) from the representation, decrypt, and obtain the byte-wise data for the decrypted value(s). 
For pointers with a single location, we must get back the pointer data structure, then simply remove the privacy label from the type stored there. 
For private pointers with multiple locations, we must declassify the pointer, retrieving it's true location and returning the byte-wise data for the pointer data structure with only that location. 
For functions, we get back the function data, then call $\erasure$ on the function body, remove the tag for whether the function has public side effects (i.e., replace with $\Box$), and call $\erasure$ on the function parameter list.

\begin{figure*}\footnotesize
$\begin{array}{l}
\bm{\erasure(\rrgamma,\ \rrsigma, \hgamma, \hsigma)} = \\ 
	\mathrm{match}\ (\rrgamma, \rrsigma)\ \mathrm{with} \\ 
	\mid\ (\RT{[\ ]}, \RT{[\ ]})\ =>\ \bm{(\hgamma,\ \hsigma)} \\ 
	\mid\ (\RT{[\ ]}, \RT{\sigma_1[\loc \to (\Null,\ \Void*,\ n,\ \PermL(\PermF,\ \Void*,\ \Pub,\ n))])}\ =>\ \\ \-\ \quad
		 \bm{(\erasure(\RT{[\ ]},\ \rsigma{_1},\ \hgamma,\ \hsigma[\loc \to (\Null,\ \Void*,\ \hn,} \
			\bm{\PermL(\perm,\ \Void*,\ \Pub,\ \hn))]))} 
	\\ 
	\mid\ (\RT{[\ ]}, \RT{\sigma_1[\loc \to (\Null,\ \RT{\Void*},\ n,\ \PermL(\PermF,\ \RT{\Type},\ \Priv,\ n))])}\ =>\ \\ \-\ \quad
		\hn = \Big(\frac{\rn}{\tau(\rType)} \Big) \cdot \tau(\erasure(\rType)) \\ \-\ \quad
		\bm{(\erasure(\RT{[\ ]},\ \rsigma{_1},\ \hgamma,\ \hsigma[\loc \to (\Null,\ \Void*,\ \hn,} \
			\bm{\PermL(\perm,\ \Void*,\ \Pub,\ \hn))]))} 
	\\ 
	\mid\ (\RT{[\ ]}, \RT{\sigma_1[\loc \to (\byte,\ \Type,\ n,\ \PermL(\perm,\ \Type,\ \llabel,\ n))])}\ =>\ \\ \-\ \quad
		\bm{(\erasure(\RT{[\ ]}, \rsigma{_1}, \hgamma, \hsigma[\loc \to (\erasure(\rbyte, \rType, \rn), \erasure(\rType), n,}
			\bm{\PermL(\perm, \erasure(\rType), \Pub, n))]))} 
	\\ 
	\mid\ (\RT{[\ ]}, \RT{\sigma_1[\loc \to (\byte,\ \Type,\ n,\ \PtrPermL(\perm,\ \Type,\ \llabel,\ n))])}\ =>\ \\ \ \ \ 
		\bm{(\erasure(\RT{[\ ]}, \rsigma{_1}, \hgamma, \hsigma[\loc \to (\erasure(\rbyte, \rType, \rn), \erasure(\rType), n,}
			\bm{\PtrPermL(\perm, \erasure(\rType), \Pub, n))]))}
	\\ 
	\mid\ (\RT{[\ ]},\RT{\sigma_1[\loc \to (\byte,\ \Type,\ 1,\ \FunPermL(\Pub))])}\ =>\ \\ \-\ \quad
		\bm{(\erasure(\RT{[\ ]},\ \rsigma{_1},\ \hgamma,\ \hsigma[\loc \to (\erasure(\rbyte, \rType, \RT1),\ \erasure(\rType),\ 1,}\
			\bm{\FunPermL(\Pub))]))}
	\\ 
	\mid\ (\RT{\gamma_1[\x \to (\loc,\ \llabel\ \btype)]},\ \rsigma{_1}\RT{[\loc \to (\byte,\ \llabel\ \btype,\ 1,\ \PermL(\perm,\ \llabel\ \btype,\ \llabel,\ 1))]})\ =>\ \\ \-\ \quad
	 	\bm{(\erasure(\rgamma{_1}, \rsigma{_1}, \hgamma[\x \to (\loc, \btype)], \hsigma[\loc \to}
			\bm{(\erasure(\rbyte, \rlabel\ \rbtype, \RT1), \btype, 1,} 
			\bm{\PermL(\perm, \btype, \Pub, 1))]))}
	\\ 
	\mid\ (\RT{\gamma_1[\res\_n \to (\loc,\ \Priv\ \btype)]},\ \rsigma{_1}\RT{[\loc \to (\byte,\ \Priv\ \btype,\ 1,}\ \RT{\PermL(\perm,\ \Priv\ \btype,\ \Priv,\ 1))]})\ =>\ 
		\\ \-\ \quad
		\bm{(\erasure(\rgamma{_1},\ \rsigma{_1},\ \hgamma,\ \hsigma))}
	\\ 
	\mid\ (\RT{\gamma_1[\x\_then\_n \to (\loc, \llabel\ \btype)]}, \rsigma{_1}\RT{[\loc \to (\byte, \llabel\ \btype, 1,} \RT{\PermL(\perm, \llabel\ \btype, \llabel, 1))]})\ =>\ \\ \-\ \quad
		\bm{(\erasure(\rgamma{_1}, \rsigma{_1}, \hgamma, \hsigma))}
	\\ 
	\mid\ (\RT{\gamma_1[\x\_else\_n \to (\loc, \llabel\ \btype)]},\ \rsigma{_1}\RT{[\loc \to (\byte, \llabel\ \btype, 1,} \RT{\PermL(\perm, \llabel\ \btype, \llabel, 1))]})\ =>\ \\ \-\ \quad
		\bm{(\erasure(\rgamma{_1}, \rsigma{_1}, \hgamma, \hsigma))}
	\\ 
	\mid\ (\RT{\gamma_1[\x \to (\loc,\ \llabel\ \Const\ \btype*)]},\ \rsigma{_1}\RT{[\loc \to (\byte,\ \llabel\ \Const\ \btype*,\ 1,}\ \RT{\PermL(\perm,\ \llabel\ \Const\ \btype*,\ \llabel,\ 1))]})\ =>\ 
	\\ \-\ \quad
		\RT{\DecodePtr(\llabel\ \Const\ \btype*, 1, \byte)} = \RT{[1,\ [(\loc_1, 0)],\ [1],\ 1]}; \\ \-\ \quad
			\rsigma{_1} = \rsigma{_2}\RT{[\loc_1 \to (\byte_1,\ \llabel\ \btype,\ n,\ \PermL(\perm,\ \llabel\ \btype,\ \llabel,\ n))]}; \\ \-\ \quad
		\bm{(\erasure(\rgamma{_1},\ \rsigma{_2},\ \hgamma[\x \to (\loc,\ \erasure(\llabel\ \Const\ \btype*))],}\ \\ \-\ \quad
			\bm{\hsigma[\loc \to (\erasure(\rbyte, \RT{\llabel\ \Const\ \btype*}, \RT1), \Const\ \btype*), 1,} 
				\bm{\PtrPermL(\perm,\ \Const\ \btype*,\ \Pub,\ 1))]} \
			\\ \-\ \quad
			\bm{[\loc_1 \to (\erasure(\RT{\byte_1}, \rlabel\ \rbtype, \rn), \btype, n,} 
				\bm{\PermL(\perm,\ \btype,\ \Pub,\ n))]))}
	\\ 
	\mid\ (\RT{\gamma_1[\x\_then\_n \to (\loc,\ \llabel\ \Const\ \btype*)]}, \rsigma{_1}\RT{[\loc \to (\byte, \llabel\ \Const\ \btype*, 1,} 
		\RT{\PtrPermL(\perm, \llabel\ \Const\ \btype*, \llabel, 1))]}) =>\ \\ \-\ \quad
		\RT{\DecodePtr(\llabel\ \Const\ \btype*, 1, \byte)} = \RT{[1,\ [(\loc_1, 0)],\ [1],\ 1]}; \\ \-\ \quad
			\rsigma{_1} = \rsigma{_2}\RT{[\loc_1 \to (\byte_1,\ \llabel\ \btype,\ n,\ \PermL(\perm,\ \llabel\ \btype,\ \llabel,\ n))]}; 
		\bm{(\erasure(\rgamma{_1},\ \rsigma{_2},\ \hgamma,\ \hsigma))}
	\\ 
	\mid\ (\RT{\gamma_1[\x\_else\_n \to (\loc,\ \llabel\ \Const\ \btype*)]},\ \rsigma{_1}\RT{[\loc \to (\byte,\ \llabel\ \Const\ \btype*,\ 1,}\ 
		\RT{\PermL(\perm,\ \llabel\ \Const\ \btype*,\ \llabel,\ 1))]}) =>\ \\ \-\ \quad
		\RT{\DecodePtr(\llabel\ \Const\ \btype*, 1, \byte)} = \RT{[1,\ [(\loc_1, 0)],\ [1],\ 1]}; \\ \-\ \quad
			\rsigma{_1} = \rsigma{_2}\RT{[\loc_1 \to (\byte_1,\ \llabel\ \btype,\ n,\ \PermL(\perm,\ \llabel\ \btype,\ \llabel,\ n))]}; 
		\bm{(\erasure(\rgamma{_1},\ \rsigma{_2},\ \hgamma,\ \hsigma))}
	\\ 
	\mid\ (\RT{\gamma_1[\x \to (\loc,\ \llabel\ \btype*)]},\ \rsigma{_1}\RT{[\loc \to (\byte,\ \llabel\ \btype*,\ n,}\ 
			\RT{\PtrPermL(\perm,\ \llabel\ \btype*,\ \llabel,\ n))]})\ =>\ \\ \-\ \quad
		 \bm{(\erasure(\rgamma{_1},\ \rsigma{_1},\ \hgamma[\x \to (\loc,\ \erasure(\llabel\ \btype*))],\ } \\ \-\ \quad
			\bm{\hsigma[\loc \to (\erasure(\rbyte, \rType, \rn),\ \erasure(\rType),\ n,}\ 
				\bm{\PtrPermL(\perm,\ \erasure(\rType),\ \Pub,\ n))]))}
	\\ 
	\mid\ (\RT{\gamma_1[\x\_then\_n \to (\loc,\ \llabel\ \btype*)]},\ \rsigma{_1}\RT{[\loc \to (\byte,\ \llabel\ \btype*,\ n,}\ 
			\RT{\PtrPermL(\perm,\ \llabel\ \btype*,\ \llabel,\ n))]})\ %
			=>\ \\ \-\ \quad
		\bm{(\erasure(\rgamma{_1},\ \rsigma{_1},\ \hgamma,\ \hsigma))}
	\\ 
	\mid\ (\RT{\gamma_1[\x\_else\_n \to (\loc,\ \llabel\ \btype*)]},\ \rsigma{_1}\RT{[\loc \to (\byte,\ \llabel\ \btype*,\ n,}\ 
			\RT{\PtrPermL(\perm,\ \llabel\ \btype*,\ \llabel,\ n))]})\ %
			=>\ \\ \-\ \quad
		 \bm{(\erasure(\rgamma{_1},\ \rsigma{_1},\ \hgamma,\ \hsigma))}
	\\ 
	\mid\ (\RT{\gamma_1[\x \to (\loc,\ \Tlist \to \Type)]}, \RT{\sigma_1[\loc \to (\byte,\ \Tlist \to \Type,\ 1,\ \FunPermL(\Pub))]}\ =>\ \\ \-\ \quad
		\bm{(\erasure(\rgamma{_1},\ \rsigma{_1},\ \hgamma[\x \to (\loc,\ \erasure(\RT{\Tlist \to \Type}))],}\ \\ \-\ \quad
			\bm{\hsigma[\loc \to (\erasure(\rbyte, \RT{\Tlist \to \Type}, \RT1),\ \erasure(\rTlist \to \rType),\ 1,\ \FunPermL(\Pub))]))}
\end{array}$
\caption{Erasure function over the environment and memory} 	\label{Fig: erasure env mem}
\end{figure*}

Figure~\ref{Fig: erasure env mem} shows erasure over the environment and memory. 
In order to properly handle all types of variables and data stored, we must iterate over both the \piccoC\ environment and memory maps, and pass along the \vanillaC\ environment and memory maps as we remove elements from the \piccoC\ maps and either add to them to the \vanillaC\ maps or discard them. 
The first case is the base case, when the \piccoC\ environment and memory are both empty, and we return the \vanillaC\ environment and memory. 
Next, we have three cases which continue to iterate through the \piccoC\ memory after the environment has been emptied. These cases are possible due to the fact that in \piccoC\ we remove mappings from the environment once they are out of scope, but we never remove mappings from memory. 

Then we have three cases to handle regular variables. The first adds mappings to the \vanillaC\ environment and memory without the privacy annotations on the types, and calls $\erasure$ on the byte-wise data stored at that location (the behavior of this is shown in Figure~\ref{Fig: erasure bytes} and described later in this section). The other two remove temporary variables (an their corresponding data) inserted by an \TT{if else} statement branching on private data. 
The cases for arrays, pointers, and functions behave similarly; however, when we have an array we handle the array pointer as well as the array data within those cases.

%% file: short-correctness.tex
\begin{definition}[$\psi$] 
\label{Def: psi}
A \LocMap\ $\psi$ is defined as a list of lists of locations, in symbols $\psi = [\ ]\ |\ \psi[\locL]$,
that is formed by tracking which locations are privately switched during the execution of the statement $\RT{\pfree(\x)}$ in a \piccoC\ program $\rstmt$ to enable comparison with the \emph{congruent} \vanillaC\ program $\hstmt$. 
\end{definition}

\begin{definition}[$\rval \Pcong \hval$]
\label{Def: val psi cong}
A \piccoC\ value and \vanillaC\ value are \emph{$\psi$-congruent}, 
in symbols $\rval \Pcong \hval$, \\
if and only if either 
$\rval \neq (\rloc, \RT\offset)$, $\hval \neq (\hloc, \Hoffset)$ and $\rval \cong \hval$, \\
or 
$\rval = (\rloc, \RT\offset)$, $\hval = (\hloc, \Hoffset)$ and $(\rloc, \RT\offset) \Pcong (\hloc, \Hoffset)$. 
\end{definition}

\begin{definition}[$\RT\code \cong \codeV$]
\label{Def: code cong}
We define \emph{congruence} over \piccoC\ codes $\RT\code \in \piccoCodes$ and $\codeV \in \vanillaCodes$, in symbols $\RT\code \cong \codeV$, by cases as follows: 
\\ if $\RT\code = \codeV$, then $\RT\code \cong \codeV$, 
\\ if $\RT\code = \RT{\mathit{iep}} \lxor \RT{\mathit{iepd}}$, then $\codeV = \codeVV{mpiet} \lxor \codeVV{mpief}$ and $\RT\code \cong \codeV$, 
\\ if $\RT\code = \RT{\mathit{mpcmp}}$, then $\codeV = \codeVV{mpcmpt} \lxor \codeVV{mpcmpf}$ and $\code \cong \codeV$, 
\\ otherwise we have 
$[\RT{\mathit{malp}}] \cong [\codeVV{ty}, \codeVV{bm}, \codeVV{mal}]$, 
$\RT{\mathit{fc1}} \cong \codeVV{fc}$, $\RT{\mathit{pin3}} \cong \codeVV{pin}$, 
$\RT{\mathit{cl1}} \cong \codeVV{cl}$, $\RT{\mathit{mpwdp2}} \cong \codeVV{mpwdp1}$, 
$\RT{\mathit{cv1}} \cong \codeVV{cv}$, $\RT{\mathit{mpwdp}} \cong \codeVV{mpwdp}$, 
$\RT{\mathit{pin4}} \cong \codeVV{pin1}$, $\RT{\mathit{pin5}} \cong \codeVV{pin2}$, 
$\RT{\mathit{mpwdp3}} \cong \codeVV{mpwdp}$, $\RT{\mathit{pin6}} \cong \codeVV{pin1}$, 
$\RT{\mathit{pin7}} \cong \codeVV{pin2}$, $\RT{\mathit{r1}} \cong \codeVV{r}$, 
$\RT{\mathit{w1}} \cong \codeVV{w}$, $\RT{\mathit{w2}} \cong \codeVV{w}$, 
$\RT{\mathit{d1}} \cong \codeVV{d}$, $\RT{\mathit{wdp2}} \cong \codeVV{wdp1}$, 
$\RT{\mathit{dp1}} \cong \codeVV{dp}$, $\RT{\mathit{wdp3}} \cong \codeVV{wdp}$, 
$\RT{\mathit{rp1}} \cong \codeVV{rp}$, $\RT{\mathit{wdp4}} \cong \codeVV{wdp}$, 
$\RT{\mathit{wp1}} \cong \codeVV{wp}$, $\RT{\mathit{rdp1}} \cong \codeVV{rdp1}$, 
$\RT{\mathit{wp2}} \cong \codeVV{wp}$, $\RT{\mathit{da1}} \cong \codeVV{da}$, 
$\RT{\mathit{ra1}} \cong \codeVV{ra}$, $\RT{\mathit{wea2}} \cong \codeVV{wea}$, 
$\RT{\mathit{wea1}} \cong \codeVV{wea}$, $\RT{\mathit{rao1}} \cong \codeVV{rao}$, 
$\RT{\mathit{wa1}} \cong \codeVV{wa}$, $\RT{\mathit{wa2}} \cong \codeVV{wa}$, 
$\RT{\mathit{wa1p}} \cong \codeVV{wa}$, $\RT{\mathit{wa2p}} \cong \codeVV{wa}$, 
$\RT{\mathit{wao2}} \cong \codeVV{wao}$, $\RT{\mathit{wao1}} \cong \codeVV{wao}$, 
$\RT{\mathit{inp3}} \cong \codeVV{inp1}$, $\RT{\mathit{inp2}} \cong \codeVV{inp}$, 
$\RT{\mathit{out3}} \cong \codeVV{out1}$, and $\RT{\mathit{out2}} \cong \codeVV{out}$. 
\end{definition}

\begin{definition}[$\RT\Pi\Pcong\Sigma$]
\label{Def: deriv cong}
Two derivations and \emph{$\psi$-congruent}, in symbols $\RT\Pi\cong_{\psi}\Sigma$, 
if and only if 
\\ $\RT\Pi \deriv ((\RT\pidA, \rrgamma^{\RT{\pidA}}_{},$ $\rrsigma^{\RT{\pidA}}_{},$ $\RT\DMap^{\RT{\pidA}}_{}$, $\rAcc^{\RT{\pidA}}_{},$ $\rstmt^{\RT{\pidA}})\ \Mid ...\Mid$ 
	$(\RT\pidZ, \rrgamma^{\RT{\pidZ}}_{},$ $\rrsigma^{\RT{\pidZ}}_{},$ $\RT\DMap^{\RT{\pidZ}}_{}$, $\rAcc^{\RT{\pidZ}}_{},$ $\rstmt^{\RT{\pidZ}}))$ 
	 \\$\Deval{\RT\locLL}{\RT\codeLL}$ 
	$((\RT\pidA, {\rrgamma^{\RT{\pidA}}_{\RT1}},$ ${\rrsigma^{\RT{\pidA}}_{\RT1}},$ $\RT\DMap^{\RT{\pidA}}_{\RT1}$, $\rAcc^{\RT{\pidA}}_{\RT1},$ ${\RT\val^{\RT{\pidA}}_{}}) \Mid ... \Mid$ 
	$(\RT\pidZ, {\rrgamma^{\RT{\pidZ}}_{\RT1}},$ ${\rrsigma^{\RT{\pidZ}}_{\RT1}},$ $\RT\DMap^{\RT{\pidZ}}_{\RT1}$, $\rAcc^{\RT{\pidZ}}_{\RT1},$ $\RT\val^{\RT{\pidZ}}))$
and 
\\$\Sigma \deriv ((\pidA,$ $\hgamma^\pidA,$ $\hsigma^\pidA,$ $\bsq,$ $\bsq,$ $\hstmt^\pidA)\Mid ...\Mid $
	$(\pidZ,$ $\hgamma^\pidZ,$ $\hsigma^\pidZ,$ $\bsq,$ $\bsq,$ $\hstmt^\pidZ))$ 
\\	$\Deval{}{\codeVLL}$ 
	$((\pidA,$ $\hgamma^\pidA_1,$ $\hsigma^\pidA_1,$ $\bsq,$ $\bsq,$ $\hval^\pidA)\Mid...\Mid$
	$(\pidZ,$ $\hgamma^\pidZ_1,$ $\hsigma^\pidZ_1,$ $\bsq,$ $\bsq,$ $\hval^\pidZ))$
such that 
\\$\{(\RT\pid, \rrgamma^{\RT{\pid}}_{},$ $\rrsigma^{\RT{\pid}}_{},$ $\RT\DMap^{\RT{\pid}}_{}$, $\rAcc^{\RT{\pid}}_{},$ $\rstmt^{\RT{\pid}})$ $\cong_{\psi_1}$ 
$(\pid,$ $\hgamma^\pid,$ $\hsigma^\pid,$ $\bsq,$ $\bsq,$ $\hstmt^\pid)\}^{\pidZ}_{\pid = \pidA}$, 
$\RT\codeLL \cong \codeVLL$, and 
\\$\{(\RT\pid, {\rrgamma^{\RT{\pid}}_{\RT1}},$ ${\rrsigma^{\RT{\pid}}_{\RT1}},$ $\RT\DMap^{\RT{\pid}}_{\RT1}$, $\rAcc^{\RT{\pid}}_{\RT1},$ $\RT\val^{\RT{\pid}}_{})$ $\cong_{\psi}$ 
$(\pid,$ $\hgamma^\pid_1,$ $\hsigma^\pid_1,$ $\bsq,$ $\bsq,$ $\hval^\pid)\}^{\pidZ}_{\pid = \pidA}$ such that $\psi$ was derived from $\psi_1$ and the derivation $\RT\Pi$. 
\end{definition}

\begin{definition}[$\RT{\val^1}\sim\RT{\val^2}$]
\label{def: val sim}
Two values are \emph{corresponding}, in symbols $\RT{\val^1}\sim\RT{\val^2}$, 
if and only if either both $\RT{\val^1},\RT{\val^2}$ are public (including locations) and $\RT{\val^1}=\RT{\val^2}$, 
or $\RT{\val^1},\RT{\val^2}$ are private and $\erasure(\RT{\val^1}) = \erasure(\RT{\val^2})$.
\end{definition}

\begin{definition}[$\RT{\Config^1} \sim \RT{\Config^2}$]
\label{def: config sim}
Two configurations are \emph{corresponding}, in symbols $\RT{\Config^1} \sim \RT{\Config^2}$ or 
$(\RT1, \RT{\gamma^1}, \RT{\sigma^1}, \RT{\DMap^1}, \RT{\Acc^1}, \RT{\stmt^1})$ $\sim$ $(\RT2, \RT{\gamma^2}, \RT{\sigma^2}, \RT{\DMap^2}, \RT{\Acc^2}, \RT{\stmt^2})$, 
if and only if 
$\RT{\gamma^1}=\RT{\gamma^2}$, $\RT{\sigma^1} \sim \RT{\sigma^2}$, $\RT{\DMap^1} \sim \RT{\DMap^2}$, $\RT{\Acc^1}=\RT{\Acc^2}$, and $\RT{\stmt^1}=\RT{\stmt^2}$. 
\end{definition}

\begin{axiom}[$MPC_b$]
\label{axiom: mpc bop}
Given $\binop\in\{+,-,\cdot, \div\}$, values $\{\RT{\n^{\pid}_{1}},$ $\RT{\n^{\pid}_{2}}, \hn_{1},$ $\hn_{2}\}^{\pidZ}_{\pid=\pidA}\in\N$, 
\\
if $MPC_b(\binop,$ $[\RT{\n^{\pidA}_{1}}, ...,$ $\RT{\n^{\pidZ}_{1}}],$ $[\RT{\n^{\pidA}_{2}}, ..., \RT{\n^{\pidZ}_{2}}])$ $= (\RT{\n^{\pidA}_{3}}, ..., \RT{\n^{\pidZ}_{3}})$, 
$\{\RT{\n^{\pid}_{1}} \cong \hn_1\}^{\pidZ}_{\pid=\pidA}$, 
	and
$\{\RT{\n^{\pid}_{2}} \cong \hn_2\}^{\pidZ}_{\pid=\pidA}$,  
\\
then $\{\RT{\n^{\pid}_{3}} \cong \hn^\pid_3\}^{\pidZ}_{\pid=\pidA}$ 
such that $\hn_1\ \binop\ \hn_2 = \hn_3$. 
\end{axiom}

\begin{lemma}[$\RT{\Config^1} \sim \RT{\Config^2} \implies \RT{\Config^1}\Pcong\hConfig \land \RT{\Config^2}\Pcong\hConfig$]
\label{lem: sim implies cong same}
Given two configurations $\RT{\Config^1},\RT{\Config^2}$ such that 
$\RT{\Config^1} = (\RT1, \RT{\gamma^1}, \RT{\sigma^1}, \RT{\DMap^1}, \RT{\Acc^1}, \RT{\stmt^1})$ and $\RT{\Config^2} = (\RT2, \RT{\gamma^2}, \RT{\sigma^2}, \RT{\DMap^2}, \RT{\Acc^2}, \RT{\stmt^2})$ and $\psi$, 
if $\RT{\Config^1} \sim \RT{\Config^2}$ 
then $\{\RT{\Config^\pid}\Pcong(\pid, \hgamma, \hsigma, \bsq, \bsq, \stmt)\}^2_{\pid=1}$. 
\end{lemma}
\begin{proof}[Proof Sketch]
Using the definition of $\erasure$ and Definition~\ref{def: config sim}, there is only one possible \vanillaC\ configuration $\hConfig$ (modulo party ID) that can be obtained from both 
$\erasure(\RT{\Config^1})$ and $\erasure(\RT{\Config^2})$. 
\end{proof}

\begin{lemma}[Unique party-wise transitions]
\label{lem: smc unique rules}
Given $((\RT\pid, \RT\gamma, \RT\sigma, \RT\DMap, \RT\Acc, \RT\stmt) \Mid \RT\Config)$
if $((\RT\pid, \RT\gamma, \RT\sigma, \RT\DMap, \RT\Acc, \RT\stmt) \Mid \RT\Config)$ $\Deval{\RT\locLL}{\RT\codeLL}$
   $((\RT\pid, \RT{\gamma_1}, \RT{\sigma_1}, \RT{\DMap_1}, \RT\Acc, \RT\val) \Mid \RT{\Config_1})$
then there exists no other rule by which $(\RT\pid, \RT\gamma, \RT\sigma, \RT\DMap, \RT\Acc, \RT\stmt)$ can step. 
\end{lemma}
\begin{proof}[Proof Sketch]
By induction on $(\RT\pid, \RT\gamma, \RT\sigma, \RT\DMap, \RT\Acc, \RT\stmt)$. 
We verify that for every configuration, given $\RT\stmt$, $\RT\Acc$, and stored type information, there is only one corresponding semantic rule.
\end{proof}

\begin{theorem}[Confluence]
\label{thm: smc confluence}
Given $\RT{\Config^\pidA} \Mid ... \Mid \RT{\Config^\pidZ}$ such that $\{\RT{\Config^\pidA} \sim \RT{\Config^\pid}\}^{\RT\pidZ}_{\RT\pid = \RT\pidA}$ 
\\
if $(\RT{\Config^\pidA} \Mid ... \Mid \RT{\Config^\pidZ})$ $\Deval{\RT{\locLL_1}}{\RT{\codeLL_1}}$ $(\RT{\Config^\pidA_1} \Mid ... \Mid \RT{\Config^\pidZ_1})$ such that $\exists \RT\pid\in\{\RT\pidA...\RT\pidZ\} \RT{\Config^\pidA_1} \not\sim \RT{\Config^\pid_1}$, 
\\
then $\exists$ $(\RT{\Config^\pidA_1} \Mid ... \Mid \RT{\Config^\pidZ_1})$ $\Deval{\RT{\locLL_2}}{\RT{\codeLL_2}}$ $(\RT{\Config^\pidA_2} \Mid ... \Mid \RT{\Config^\pidZ_2})$
\\ 
such that $\{\RT{\Config^\pidA_2} \sim \RT{\Config^\pid_2}\}^{\RT\pidZ}_{\RT\pid = \RT\pidA}$, 
$\{(\RT{\locLL^\pidA_1}\addL\RT{\locLL^\pidA_2}) = (\RT{\locLL^\pid_1}\addL\RT{\locLL^\pid_2})\}^{\RT\pidZ}_{\RT\pid = \RT\pidA}$, 
and $\{(\RT{\codeLL^\pidA_1}\addC\RT{\codeLL^\pidA_2}) = (\RT{\codeLL^\pid_1}\addC\RT{\codeLL^\pid_2})\}^{\RT\pidZ}_{\RT\pid = \RT\pidA}$.
\end{theorem}

\begin{proof}[Proof Sketch]
By Lemma~\ref{lem: smc unique rules}, we have that there is only one possible execution trace for any given party based on the starting configuration. 
~\\ 
By definition of $\{\RT{\Config^\pidA} \sim \RT{\Config^\pid}\}^{\RT\pidZ}_{\RT\pid = \RT\pidA}$, 
we have that the starting states of all parties are corresponding, with identical statements. 
~\\ 
Therefore, all parties must follow the same execution trace and will eventually reach another set of corresponding states.
\end{proof}

\begin{axiom}
\label{correctness assumption}
For purposes of correctness, 
we assume all parties are executing program $\RT\stmt$ 
from initial state $(\RT\pid, \RT{[\ ]}, \RT{[\ ]}, \RT{[\ ]}, \RT\AccZ, \RT\stmt)$
and corresponding input data. 
\end{axiom}

\begin{theorem}[Semantic Correctness]
\label{Thm: app correctness}
~\\
For every configuration $\{(\RT\pid,\ \rrgamma^{\RT{\pid}}_{},$ $\rrsigma^{\RT{\pid}}_{},$ $\RT\DMap^{\RT{\pid}}_{}$, $\rAcc^{\RT{\pid}}_{},$ $\rstmt^{\RT\pid})\}^{\RT\pidZ}_{\RT\pid = \RT\pidA}$, 
$\{(\pid,$ $\hgamma^\pid,$ $\hsigma^\pid,$ $\bsq,$ $\bsq,$ $\hstmt^\pid)\}^{\pidZ}_{\pid = \pidA}$ and \LocMap\ $\psi$ 
\\ such that $\{(\RT\pid, \rrgamma^{\RT{\pid}}_{},$ $\rrsigma^{\RT{\pid}}_{},$ $\RT\DMap^{\RT{\pid}}_{}$, $\rAcc^{\RT{\pid}}_{},$ $\rstmt^{\RT{\pid}})$ $\Pcong$ 
$(\pid,$ $\hgamma^\pid,$ $\hsigma^\pid,$ $\bsq,$ $\bsq,$ $\hstmt^\pid)\}^{\pidZ}_{\pid = \pidA}$, 
\\ 
if $\RT\Pi \deriv ((\RT\pidA, \rrgamma^{\RT{\pidA}}_{},$ $\rrsigma^{\RT{\pidA}}_{},$ $\RT\DMap^{\RT{\pidA}}_{}$, $\rAcc^{\RT{\pidA}}_{},$ $\rstmt^{\RT{\pidA}})\ \Mid ...\Mid$ 
	$(\RT\pidZ, \rrgamma^{\RT{\pidZ}}_{},$ $\rrsigma^{\RT{\pidZ}}_{},$ $\RT\DMap^{\RT{\pidZ}}_{}$, $\rAcc^{\RT{\pidZ}}_{},$ $\rstmt^{\RT{\pidZ}}))$ 
	\\ \-\ \-\ \-\ $\Deval{\RT\locLL}{\RT\codeLL}$ 
	$((\RT\pidA, {\rrgamma^{\RT{\pidA}}_{\RT1}},$ ${\rrsigma^{\RT{\pidA}}_{\RT1}},$ $\RT\DMap^{\RT{\pidA}}_{\RT1}$, $\rAcc^{\RT{\pidA}}_{\RT1},$ ${\RT\val^{\RT{\pidA}}_{}}) \Mid ... \Mid$ 
	$(\RT\pidZ, {\rrgamma^{\RT{\pidZ}}_{\RT1}},$ ${\rrsigma^{\RT{\pidZ}}_{\RT1}},$ $\RT\DMap^{\RT{\pidZ}}_{\RT1}$, $\rAcc^{\RT{\pidZ}}_{\RT1},$ $\RT\val^{\RT{\pidZ}}))$ 
\\ for codes $\RT\codeLL \in \RT{\piccoCodes}$,
then there exists a derivation 
\\ 
$\Sigma \deriv\ ((\pidA,$ $\hgamma^\pidA,$ $\hsigma^\pidA,$ $\bsq,$ $\bsq,$ $\hstmt^\pidA)\Mid ...\Mid $
	$(\pidZ,$ $\hgamma^\pidZ,$ $\hsigma^\pidZ,$ $\bsq,$ $\bsq,$ $\hstmt^\pidZ))$ 
	\\ $\Deval{}{\codeVLL}$ 
	$((\pidA,$ $\hgamma^\pidA_1,$ $\hsigma^\pidA_1,$ $\bsq,$ $\bsq,$ $\hval^\pidA)\Mid...\Mid$
	$(\pidZ,$ $\hgamma^\pidZ_1,$ $\hsigma^\pidZ_1,$ $\bsq,$ $\bsq,$ $\hval^\pidZ))$ 
\\ for codes $\codeVLL \in \vanillaCodes$ 
and 
a \LocMap\ $\psi_1$ 
such that 
\\ 
$\RT\codeLL \cong \codeVLL$, 
$\{(\RT\pid, {\rrgamma^{\RT{\pid}}_{\RT1}},$ ${\rrsigma^{\RT{\pid}}_{\RT1}},$ $\RT\DMap^{\RT{\pid}}_{\RT1}$, $\rAcc^{\RT{\pid}}_{\RT1},$ $\RT\val^{\RT{\pid}}_{})$ $\cong_{\psi_1}$ 
$(\pid,$ $\hgamma^\pid_1,$ $\hsigma^\pid_1,$ $\bsq,$ $\bsq,$ $\hval^\pid)\}^{\pidZ}_{\pid = \pidA}$, 
and $\RT\Pi \cong_{\psi_1} \Sigma$.
\end{theorem}

\begin{proof}[Proof Sketch]
By induction over all \piccoC\ semantic rules. 

The bulk of the complexity of this proof lies with rules pertaining to Private If Else, handling of pointers, and freeing of memory.  We first provide a brief overview of the intuition for the simpler cases and then dive deeper
into the details for the more complex cases.  Full proofs are available in our artifact submission.

For the rules evaluating over public data, correctness follows simply as the \vanillaC\ and \piccoC\ rules for public data are nearly identical. For all the semantic rules that use general helper algorithms (i.e., algorithms in common to both \vanillaC\ and \piccoC), we also reason about the correctness of the helper algorithms, comparing the \vanillaC\ version and the \piccoC version. Correctness over such algorithms is easily proven, as these algorithms are nearly identical, differing on privacy labels as we do not have private data in \vanillaC. 

For all \piccoC\ multiparty semantic rules, we relate them to the multiparty versions of the \vanillaC\ rules. To reason about the multiparty protocols, we leverage Axioms, such as Axiom~\ref{axiom: mpc bop}, to prove these rules correct. These Axioms should be proven correct by a library developer to ensure the completeness of the formal model. The correctness of most multiparty semantic rules follows easily, with Multiparty Private Free being an exception. For this rule, we also must reason about our helper algorithms that are specific to the \piccoC\ semantics (e.g., $\UpdateBytesFree$, $\UpdatePtrLocs$). We leverage the correctness of the behavior of the multiparty protocol $\PFree$, to show that correctness of these algorithms follows due to the deterministic definitions of the algorithms. In this case, we must also show that the locations that are swapped within this rule (which is done to hide the true location) are deterministic based on our memory model definition. We use $\psi$ to map the swapped locations, enabling us to show that, if these swaps were reversed, we would once again have memories that are directly congruent. This concept of locations being $\psi$-congruent is particularly necessary when reasoning about pointers in other rule cases. 
For all the rules using private pointers, we will rely upon the pointer data structure containing a set of locations and their associated tags, only one of which being the true location. With this proven to be the case, it is then clear that the true location indicated within the private pointer's data structure in \piccoC\ will be $\psi$-congruent  with the location given by the pointer data structure in \vanillaC.
In our proof, we make the assumption that locations are not hard-coded, as hard-coded locations would lead to potentially differing results between \vanillaC\ and \piccoC\ execution due to the behavior of \TT{pfree}. Additionally, given the distributed nature of the \piccoC, it would not make sense to allow hard-coded locations, as a single program will be executed on several different machines.

For rule Private Malloc, we must relate this rule to the sequence of \vanillaC\ rules for Malloc, Multiplication, and Size Of Type. This is due to the definition of \TT{pmalloc} as a helper that allows the user to write programs without knowing the size of private types. This case follows from the definition of translating the \piccoC\ program to a \vanillaC\ program, $\bm{\erasure}({\PMalloc(\Expr,\ \Type)} 
= {(\Malloc(\sizeof(\bm{\erasure}(\Type)) \cdot \bm{\erasure}(\Expr)))})$.

For the Private If Else rules, we must reason that our end results in memory after executing both branches and resolving correctly match the end result of having only executed the intended branch. 
The cases for both of these rules will have two subcases - one for the conditional being true, and the other for false. 
To obtain correctness, we use multiparty versions of the if else true and false rules that execute both branches - this allows us to reason that both branches will evaluate properly, and that we will obtain the correct ending state once completed. 
For both rules, we must first show that $\DynExtract$ will correctly find all non-local variables that are modified within both branches, including non-assignment modifications such as use of the pre-increment operator $++\x$, and that all such modified variables will be added to the list (excluding pointers modified exclusively by pointer dereference write statements). We must also show that it will correctly find and tag if a pointer dereference write statement was found. These properties follow deterministically from the definition of the algorithm. 

For rule Private If Else Variable Tracking, we will leverage the correctness of $\DynExtract$, and that if $\DynExtract$ returns the tag 0, no pointer dereference writes were found. We then reason that $\Initialize$ will correctly create the assignment statements for our temporary variables, and that the original values for each of the modified variables will be stored into the \TT{else} temporary variables. The temporaries being stored into memory correctly through the evaluation of these statements follows by induction. Next we have the evaluation of the \TT{then} branch, which will result in the values that are correct for if the condition had been true - this holds by induction. 
We then proceed to reason that $\Restore$ will properly create the statements to store the ending results of the \TT{then} branch into the \TT{then} temporary variables, and restore all of the original values from the \TT{else} variables (the original values being correctly stored follows from $\Initialize$ and the evaluation of it's statements). The correct evaluation of the this set of statements follows by induction. 
Next we have the evaluation of the \TT{else} branch, which will result in the values that are correct for if the condition had been false - this holds by induction and the values having been restored to the original values properly.
We will then reason about the correctness of the statements created by $\Resolve$. These statements must be set up to correctly take the information from the \TT{then} temporary variable, the temporary variable for the condition for the branch, and the ending result for all variables from the \TT{else} branch. For the resolution of pointers, we insert a call for a resolution function ($\resolve$), because the resolution of pointer data is more involved. The evaluation of this function is shown in rule Multiparty Resolve Pointer Locations. By proving that this rule will correctly resolve the true locations for pointers, we will then have that the statements created by $\Resolve$ will appropriately resolve all

For rule Private If Else Location Tracking, the structure of the case is similar to the rule for variable tracking, but with a few differences we will discuss here. 
For this rule, we will need to reason about $\DynUpdate$, and that we will catch all modifications by pointer dereference writes and properly add them to $\DMap$ if the location being modified is not already tracked. If a new mapping is added, we store the current value in $\val_\mathit{orig}$ (as this location has not yet been modified) and the tag has to be set to 0. This behavior will be used to ensure the correctness during resolution.
For $\DynInit$, we must reason that we correctly initialize the map $\DMap$ with all of the locations we found within $\DynExtract$ to be modified by means other than pointer dereference writes and store their original values in $\val_\mathit{orig}$. 
Then we can evaluate the \TT{then} branch, which will result in the values that are correct for if the condition had been true - this holds by induction. 
For $\DynRestore$, we reason that we properly store the results of the \TT{then} branch, and update the tag for the location to signify that we should use $\val_\mathit{then}$ instead of $\val_\mathit{orig}$. We will then restore the original values, leveraging the correctness of $\DynInit$ to prove this will happen correctly. 
Then we can evaluate the \TT{else} branch, which will result in the values that are correct for if the condition had been false - this holds by induction. 
For $\DynResolve$, we reason that we will create the appropriate resolution statements to be executed. For the \TT{then} result, these statements must use the value stored in $\val_\mathit{orig}$ if the tag is set to 0 (this occurs if the first modification to the location was a pointer dereference write within the \TT{else} branch), and the value stored in $\val_\mathit{then}$ if the tag is set to 1. We prove this to be the correct \TT{then} result through the correctness of $\DynUpdate$ and $\DynRestore$. The \TT{else} result must use the current value for that location in memory, which is proven to be the correct \TT{else} result through the correctness of $\DynInit$ and $\DynResolve$. In this way, we can prove the correctness the contents of the statements created by $\DynResolve$, and then the correctness of the evaluation of the statements created by $\DynRestore$ will hold as we discussed for with those created by $\Resolve$ for Private If Else Variable tracking. 
\end{proof}

%% file: short-ni.tex
\begin{definition}
\label{def: loweq tree}
Two \piccoC\ evaluation trees $\Pi$ and $\Sigma$ are  \emph{low-equivalent}, in symbols $\Pi \loweq \Sigma$, if and only if $\Pi$ and $\Sigma$ have the same structure as trees, and for each node in $\Pi$ proving 
\\ \-\ \quad \-\ $((\pidA, \gamma^{\pidA}_{},$ $\sigma^{\pidA}_{},$ $\DMap^{\pidA}_{}$, $\Acc^{\pidA}_{},$ $\stmt)\Mid ...\Mid (\pidZ, \gamma^{\pidZ}_{},$ $\sigma^{\pidZ}_{},$ $\DMap^{\pidZ}_{}$, $\Acc^{\pidZ}_{},$ $\stmt))$ 
\\ $\Deval{\locLL}{\codeLL}$ $((\pidA, \gamma^{\pidA}_{1},$ $\sigma^{\pidA}_{1},$ $\DMap^{\pidA}_{1}$, $\Acc^{\pidA}_{1},$ $\val^{\pidA}_{})\Mid ...\Mid (\pidZ, \gamma^{\pidZ}_{1},$ $\sigma^{\pidZ}_{1},$ $\DMap^{\pidZ}_{1}$, $\Acc^{\pidZ}_{1},$ $\val^{\pidZ}_{}))$, the corresponding node in $\Sigma$ proves 
\\ \-\ \quad \-\ \-\ $((\pidA, \gamma^{\pidA}_{},$ $\sigma^{\pidA}_{},$ $\DMap^{\pidA}_{}$, $\Acc^{\pidA}_{},$ $\stmt)\Mid ...\Mid (\pidZ, \gamma^{\pidZ}_{},$ $\sigma^{\pidZ}_{},$ $\DMap^{\pidZ}_{}$, $\Acc^{\pidZ}_{},$ $\stmt))$ 
\\ $\Deval{\locLL'}{\codeLL'}$ $((\pidA, \gamma^{\pidA}_{1},$ $\sigma^{\pidA}_{1},$ $\DMap^{\pidA}_{1}$, $\Acc^{\pidA}_{1},$ $\val^{\pidA}_{})\Mid ...\Mid (\pidZ, \gamma^{\pidZ}_{1},$ $\sigma^{\pidZ}_{1},$ $\DMap^{\pidZ}_{1}$, $\Acc^{\pidZ}_{1},$ $\val^{\pidZ}_{}))$, $\codeLL = \codeLL'$ and $\locLL = \locLL'$.
\end{definition}

\begin{axiom}[$MPC_{ar}$]
\label{axiom: mpc ar ni}
Given indices $\{\ind^{\pid}_{},$ $\ind'^{\pid}_{}\}^{\pidZ}_{\pid=\pidA}$, arrays $\{[\val^{\pid}_{1},$ $...,$ $\val^{\pid}_{n}]$, $[\val'^{\pid}_{1},$ $...,$ $\val'^{\pid}_{n}]\}^{\pidZ}_{\pid=\pidA}$, 
\\
if $MPC_{ar}((\ind^{\pidA}_{},$ $[\val^{\pidA}_{1},$ $...,$ $\val^{\pidA}_{n}]),$ $...,$ $(\ind^{\pidZ}_{},$ $[\val^{\pidZ}_{1},$ $...,$ $\val^{\pidZ}_{n}]))$ $= (\val^{\pidA}_{},$ $...,$ $\val^{\pidZ}_{})$, 
\\ \-\ \-\ 
$MPC_{ar}((\ind'^{\pidA}_{},$ $[\val'^{\pidA}_{1},$ $...,$ $\val'^{\pidA}_{n}]),$ $...,$ $(\ind'^{\pidZ}_{},$ $[\val'^{\pidZ}_{1},$ $...,$ $\val'^{\pidZ}_{n}]))$ $= (\val'^{\pidA}_{},$ $...,$ $\val'^{\pidZ}_{})$,  
\\ \-\ \-\ 
$\{\ind^{\pid}_{} = \ind'^{\pid}_{}\}^{\pidZ}_{\pid=\pidA}$, 
and $\{[\val^{\pid}_{1},$ $...,$ $\val^{\pid}_{n}] = [\val'^{\pid}_{1},$ $...,$ $\val'^{\pid}_{n}]\}^{\pidZ}_{\pid=\pidA}$
\\
then $\{\val^{\pid}_{} = \val'^{\pid}_{}\}^{\pidZ}_{\pid=\pidA}$. 
\end{axiom}

\begin{axiom}[$\MPC{b}$]
\label{axiom: mpc b ni}
Given values $\{\val^{\pid}_{1},$ $\val^{\pid}_{2},$ $\val^{\pid}_{3},$ $\val'^{\pid}_{1},$ $\val'^{\pid}_{2},$ $\val'^{\pid}_{3}\}^{\pidZ}_{\pid=\pidA}$ and binary operation $\binop\in\{\cdot, +, -, \div\}$, 
\\
if $MPC_{b}(\binop,$ $\val^{\pidA}_{1},$ $\val^{\pidA}_{2}, ...,$ $\val^{\pidZ}_{1}, \val^{\pidZ}_{2})$ $= (\val^{\pidA}_{3}, ..., \val^{\pidZ}_{3})$, 
\\ \-\ \-\ 
$MPC_{b}(\binop,$ $\val'^{\pidA}_{1},$ $\val'^{\pidA}_{2}, ...,$ $\val'^{\pidZ}_{1}, \val'^{\pidZ}_{2})$ $= (\val'^{\pidA}_{3}, ..., \val'^{\pidZ}_{3})$,  
$\{\val^{\pid}_{1} = \val'^{\pid}_{1}\}^{\pidZ}_{\pid=\pidA}$, 
and $\{\val^{\pid}_{2} = \val'^{\pid}_{2}\}^{\pidZ}_{\pid=\pidA}$
\\
then $\{\val^{\pid}_{3} = \val'^{\pid}_{3}\}^{\pidZ}_{\pid=\pidA}$. 
\end{axiom}

\begin{theorem}[Multiparty Noninterference]
\label{thm: NI app}
For every environment $\{\gamma^{\pid}_{},$ $\gamma^{\pid}_{1},$ $\gamma'^{\pid}_{1}\}^{\pidZ}_{\pid = \pidA}$; 
memory $\{\sigma^{\pid}_{}$, $\sigma^{\pid}_{1}$, $\sigma'^{\pid}_{1} \}^{\pidZ}_{\pid = \pidA}\in\Mem$; 
\changeMap $\{\DMap^{\pid}_{}$, $\DMap^{\pid}_{1}$, $\DMap'^{\pid}_{1}\}^{\pidZ}_{\pid = \pidA}$;
accumulator $\{\Acc^{\pid}_{}$, $\Acc^{\pid}_{1}$, $\Acc'^{\pid}_{1}\}^{\pidZ}_{\pid = \pidA}\in\N$; 
statement $\stmt$, values $\{\val^{\pid}_{}$, $\val'^{\pid}_{}\}^{\pidZ}_{\pid = \pidA}$; 
step evaluation code lists $\codeLL,\codeLL'$ and their corresponding lists of locations accessed $\locLL,\locLL'$, 
party $\pid \in \{\pidA...\pidZ\}$; 
\\
if 
$\Pi \deriv ((\pidA, \gamma^{\pidA}_{},$ $\sigma^{\pidA}_{},$ $\DMap^{\pidA}_{}$, $\Acc^{\pidA}_{},$ $\stmt)\Mid ...\Mid (\pidZ, \gamma^{\pidZ}_{},$ $\sigma^{\pidZ}_{},$ $\DMap^{\pidZ}_{}$, $\Acc^{\pidZ}_{},$ $\stmt))$ 
\\ $\-\ \quad\Deval{\locLL}{\codeLL}$ $((\pidA, \gamma^{\pidA}_{1},$ $\sigma^{\pidA}_{1},$ $\DMap^{\pidA}_{1}$, $\Acc^{\pidA}_{1},$ $\val^{\pidA}_{})\Mid ...\Mid (\pidZ, \gamma^{\pidZ}_{1},$ $\sigma^{\pidZ}_{1},$ $\DMap^{\pidZ}_{1}$, $\Acc^{\pidZ}_{1},$ $\val^{\pidZ}_{}))$ 
\\ and   
$\Sigma \deriv ((\pidA, \gamma^{\pidA}_{},$ $\sigma^{\pidA}_{},$ $\DMap^{\pidA}_{}$, $\Acc^{\pidA}_{},$ $\stmt)\Mid ...\Mid (\pidZ, \gamma^{\pidZ}_{},$ $\sigma^{\pidZ}_{},$ $\DMap^{\pidZ}_{}$, $\Acc^{\pidZ}_{},$ $\stmt))$ 
\\ $\-\ \quad \-\ \Deval{\locLL'}{\codeLL'}$ $((\pidA, \gamma'^{\pidA}_{1},$ $\sigma'^{\pidA}_{1},$ $\DMap'^{\pidA}_{1}$, $\Acc'^{\pidA}_{1},$ $\val'^{\pidA}_{})\Mid ...\Mid (\pidZ, \gamma'^{\pidZ}_{1},$ $\sigma'^{\pidZ}_{1},$ $\DMap'^{\pidZ}_{1}$, $\Acc'^{\pidZ}_{1},$ $\val'^{\pidZ}_{}))$
\\ then $\{\gamma^{\pid}_{1}=\gamma'^{\pid}_{1}\}^{\pidZ}_{\pid = \pidA}$, 
$\{\sigma^{\pid}_{1}=\sigma'^{\pid}_{1}\}^{\pidZ}_{\pid = \pidA}$, 
$\{\DMap^{\pid}_{1} =\DMap'^{\pid}_{1}\}^{\pidZ}_{\pid = \pidA}$, 
$\{\Acc^{\pid}_{1}=\Acc'^{\pid}_{1}\}^{\pidZ}_{\pid = \pidA}$, 
$\{\val^{\pid}_{}=\val'^{\pid}_{}\}^{\pidZ}_{\pid = \pidA}$, 
$\codeLL=\codeLL'$, 
$\locLL = \locLL'$, 
$\Pi \loweq \Sigma$.
~\\
\end{theorem}

\begin{proof}[Proof Sketch]
By induction over all \piccoC\ semantic rules. We make the assumption that both evaluation traces are over the same program (this is given by having the same $\stmt$ in the starting states) and all public data will remain the same, including data read as input during the evaluation of the program. A portion of the complexity of this proof is within ensuring that memory accesses within our semantics remain data oblivious. 
Several rules follow fairly simply and leverage similar ideas, which we will discuss first, and then we will provide further intuition behind the more complex cases. 

For all rules leveraging helper algorithms, we must reason about the helper algorithms, and that they behave deterministically by definition and have data-oblivious memory accesses. Given this and that these helper algorithms do no modify the private data, we maintain the properties of noninterference of this theorem. First we reason that our helper algorithms to translate values into their byte representation will do so deterministically, and therefore maintain indistinguishability between the value and byte representation. We can then reason that our helper algorithms that take these byte values and store them into memory will also do so deterministically, so that when we later access the data in memory we will obtain the same indistinguishable values we had stored. 

It is also important to take note here our functions to help us retrieve data from memory, particularly in cases such as when reading out of bounds of an array. When proving these cases to maintain noninterference, we leverage our definition of how memory blocks are assigned in a monotonically increasing fashion, and how the algorithms for choosing which memory block to read into after the current one are deterministic. This, as well as our original assumptions of having identical public input, allows us to reason that if we access out of bounds (including accessing data at a non-aligned position, such as a chunk of bytes in the middle of a memory block), we will be pulling from the same set of bytes each time, and therefore we will end up with the same interpretation of the data as we continue to evaluate the remainder of the program. It is important to note again here that by definition, our semantics will always interpret bytes of data as the type it is expected to be, not the type it actually is (i.e., reading bytes of data that marked private in memory by overshooting a public array will not decrypt the bytes of data, but instead give you back a garbage public value). To reiterate this point, even when reading out of bounds, we will not reveal anything about private data, as the results of these helper algorithms will be indistinguishable.

To reason about the multiparty protocols, we leverage Axioms, such as Axiom~\ref{axiom: mpc b ni}, to reason that the protocols will maintain our definition of noninterference. With each of these Axioms, we ensure that over two different evaluations, if the values of the first run ($\val^\pid_1, \val^\pid_2$) are not distinguishable from those of the second ($\val'^\pid_1, \val'^\pid_2$), then the resulting values are also not distinguishable ($\val^\pid_3 = \val'^\pid_3$). These Axioms should be proven by a library developer to ensure the completeness of the formal model. 

For private pointers, it is important to note that the obtaining multiple locations is deterministic based upon the program that is being evaluated. A pointer can initially gain multiple locations through the evaluation of a private if else. Once there exists a pointer that has obtained multiple locations in such a way, it can be assigned to another pointer to give that pointer multiple locations. The other case for a pointer to gain multiple location is through the use of \TT{pfree} on a pointer with multiple locations (i.e., the case where a pointer has locations $\loc_1$, $\loc_2$, $\loc_3$ and we free $\loc_1$) - when this occurs, if another pointer had referred to only $\loc_1$, it will now gain locations in order to mask whether we had to move the true location or not. 
When reasoning about pointers with multiple locations, we maintain that given the tags for which location is the true location are indistinguishable, then it is not possible to distinguish between them by their usage as defined in the rules or helper algorithms using them. 
Additionally, to reason about \TT{pfree}, we leverage that the definitions of the helper algorithms are deterministic, and that (wlog), we will be freeing the same location. We will then leverage our Axiom about the multiparty protocol $\PFree$. After the evaluation of $\PFree$, it will deterministically update memory and all other pointers as we mentioned in the brief example above.

For both Private If Else rules, the most important element we must leverage is how values are resolved, showing that given our resolution style, we are not able to distinguish between the ending values. 
In order to do this, we also must reason about the entirety of the rule, including all of if else helper algorithms. 
First, we note that the evaluation of the \TT{then} branches follows by induction, as does the evaluation of the \TT{else} branch once we have reasoned through the restoration phase. 
For variable tracking, it is clear from the definitions of $\DynExtract$, $\Initialize$, and $\Restore$ that the behavior of these algorithms is deterministic and given the same program, we will be extracting, initializing, and restoring the same set variables every time we evaluate the program. For location tracking, $\DynInit$ is also immediately clear that it will be initializing the same locations each time. We must then reason about $\DynUpdate$, and how given a program, we will deterministically find the pointer dereference writes and array writes at public indices at corresponding positions in memory and add them to our tracking structure $\DMap$. Then we can reason that the behavior of $\DynRestore$ will deterministically perform the same updates, because $\DMap$ will contain the same information in every evaluation. 
Now, we are able to move on to reasoning about resolution, and show that given all of this and the definitions of the resolution helper algorithms and rule, we are not able to distinguish between the ending values.  

One of the main complexities of this proof revolves around ensuring \emph{data-oblivious memory accesses} (i.e. that we always access locations deliberately based on public information), particularly when handling arrays and pointers. Within the proof, we must consider all helper algorithms, and what locations are accessed within the algorithms as well as within the rules. What locations are accessed within the algorithms follows deterministically from the definition of the algorithms, and we return from the algorithms which locations were accessed in order to properly reason about the entire evaluation trace of the program. Our semantics are designed in such a way that we give the multiparty protocols all of the information they need, with all memory accesses being completed within the rule itself or our helper algorithms. This also helps show that memory accesses are purely local, not distributed operations.  
Within the array rules, the main concern is in reading from and writing at a private index. We currently handle this complexity within our rules by accessing all locations within the array in rules Multiparty Array Read Private Index and Multiparty Array Write Private Index. 
In Multiparty Array Read Private Index, we clearly read data from every index of the array ($\{\forall \ind \in \{0...\nl-1\} \quad \DecodeArr({\llabel\ \btype}, \ind, {\byte^\pid_1}) =  \n^\pid_\ind\}^{\pidZ}_{\pid = \pidA}$), then that data is passed to the multiparty protocol. 
Similarly, in Multiparty Array Write Private Index, we read data from every index of the array, pass it to the multiparty protocol, then proceed to update every index of the array with what was returned from the protocol. 
Within the multiparty protocols used in these two rules, we will ensure the usage of the data is data-oblivious within the main noninterference proof in the following subsection. 
All other array rules use public indices, and in turn only access that publicly known location.
Within the pointer rules, our main concern is that we access all locations that are referred to by a private pointer when we have multiple locations. For this, we will reason about the contents of the rules and the helper algorithms used by the pointer rules, which can be shown to deterministically do so.
\end{proof}